\documentclass[11pt,a4paper]{article}

\usepackage[T1]{fontenc}
\usepackage[utf8]{inputenc}
\usepackage{lmodern}
\usepackage[margin=1.15in]{geometry}
\usepackage{microtype}
\usepackage{amsmath,amssymb,amsfonts,amsthm,mathtools}
\usepackage{mathrsfs}
\usepackage{bbold}
\usepackage{stmaryrd}
\usepackage{physics}
\usepackage{tikz-cd}
\usepackage{prftree}
\usepackage{enumitem}
\usepackage{xcolor}
\usepackage{xparse}
\usepackage{csquotes}
\usepackage[round,authoryear]{natbib}
\usepackage[colorlinks=true,citecolor=blue,linkcolor=blue,urlcolor=blue]{hyperref}
\usepackage[nameinlink,capitalize]{cleveref}

\allowdisplaybreaks
\newtheorem{theorem}{Theorem}[section]
\newtheorem{proposition}[theorem]{Proposition}
\newtheorem{lemma}[theorem]{Lemma}

\theoremstyle{definition}
\newtheorem{definition}[theorem]{Definition}

\theoremstyle{remark}

\ExplSyntaxOn
\NewDocumentCommand{\maybeopsup}{m}{
  \tl_if_blank:nF {#1}
    {^{\tl_if_eq:nnTF {#1} {o} {\mathrm{op}} {#1}}}
}
\NewDocumentCommand{\maybebracket}{m}{
  \tl_if_blank:nF {#1}{\!\left[#1\right]}
}
\ExplSyntaxOff

\DeclareDocumentCommand{\Set}{O{}}{\mathbf{Set}\maybeopsup{#1}}
\DeclareDocumentCommand{\Cat}{O{}}{\mathbf{Cat}\maybeopsup{#1}}
\DeclareDocumentCommand{\Vect}{O{}}{\mathbf{Vect}\maybeopsup{#1}}
\DeclareDocumentCommand{\R}{O{}}{\mathbb{R}\maybeopsup{#1}}
\providecommand{\N}{\mathbb{N}}
\providecommand{\Id}{\mathrm{Id}}
\providecommand{\id}{\mathrm{id}}

\providecommand{\const}{\mathrm{const}}
\providecommand{\AST}{\mathsf{AST}}
\providecommand{\Syn}{\mathsf{Syn}}
\providecommand{\Fam}[1]{\mathsf{Fam}(#1)}
\providecommand{\TE}{T^{\ast}\E}
\DeclareDocumentCommand{\TR}{O{}}{T^{\ast}\R\maybeopsup{#1}}
\providecommand{\E}{\mathbb{E}}
\providecommand{\blankdash}{\mathord{-}}
\DeclareDocumentCommand{\srcreal}{O{}}{\mathsf{real}\maybeopsup{#1}}
\providecommand{\srcop}{\mathsf{op}}
\providecommand{\tgtop}{\mathsf{op}}
\providecommand{\casubst}[3]{#1[#3/#2]}
\providecommand{\casest}[5]{\mathsf{case}\;#1\;\mathsf{of}\;\mathsf{inl}(#2)\Rightarrow #3\;|\;\mathsf{inr}(#4)\Rightarrow #5}
\providecommand{\letst}[3]{\mathsf{let}\;#1=#2\;\mathsf{in}\;#3}
\providecommand{\supp}{\operatorname{supp}}
\providecommand{\base}[1]{\left(#1\right)_{\mathrm{set}}}
\providecommand{\fibr}[1]{\left(#1\right)_{\mathrm{fbr}}}
\providecommand{\dompt}[1]{\left(#1\right)_{\mathrm{dom}}}
\providecommand{\codpt}[1]{\left(#1\right)_{\mathrm{cod}}}
\providecommand{\indicator}[1]{\mathbb{1}_{#1}}
\providecommand{\st}{\;\middle|\;}
\providecommand{\lt}{<}
\providecommand{\gt}{>}
\providecommand{\lintype}[1]{\underline{#1}}
\DeclareDocumentCommand{\chad}{O{}}{\mathsf{CHAD}\maybebracket{#1}}
\providecommand{\setint}[1]{\llbracket #1\rrbracket_{\Set}}
\providecommand{\genint}[1]{\llbracket #1\rrbracket}
\providecommand{\lrint}[1]{\llbracket #1\rrbracket_{\cat[]{LR}}}
\providecommand{\fvoint}[1]{\llbracket #1\rrbracket_{\cat[]{F}}}
\providecommand{\linint}[2]{\llbracket #1\rrbracket_{#2}^{\mathrm{lin}}}
\providecommand{\linto}{\multimap}
\providecommand{\MFin}{\mathsf{MFin}}
\DeclareDocumentCommand{\LSyn}{O{}}{\mathsf{LSyn}\maybebracket{#1}}
\DeclareDocumentCommand{\CSyn}{O{}}{\mathsf{CSyn}\maybeopsup{#1}}
\providecommand{\Tgt}{\mathsf{Tgt}}

\NewDocumentCommand{\cat}{O{} m}{\mathsf{#2}}

\NewDocumentCommand{\opn}{O{} m}{\operatorname{#2}}
\NewDocumentCommand{\fn}{O{} m m}{#2\!\left(#3\right)}
\NewDocumentCommand{\Fun}{O{} m m}{#2\!\left(#3\right)}
\NewDocumentCommand{\Nat}{O{} m m}{\left(#2\right)_{#3}}
\NewDocumentCommand{\tuple}{m}{\left\langle #1\right\rangle}
\NewDocumentCommand{\tupling}{m}{\left\langle #1\right\rangle}
\NewDocumentCommand{\cotupling}{m}{\left[ #1\right]}
\NewDocumentCommand{\inthom}{m m}{\left[#1,#2\right]}
\NewDocumentCommand{\commacat}{m m}{\left(#1\downarrow #2\right)}
\RenewDocumentCommand{\abs}{m}{\left|#1\right|}
\DeclareDocumentCommand{\opmor}{O{} m}{#2^{\mathrm{op}}}
\NewDocumentCommand{\casenst}{m m m m m}{\operatorname{case}^{#3}_{#2}\left(#1;#4.#5\right)}
\NewDocumentCommand{\caseprod}{m m m}{\operatorname{caseprod}\left(#1;#2.#3\right)}
\NewDocumentCommand{\linapp}{m m}{#1\, #2}

\providecommand{\categ}{\operatorname{categorical}}
\providecommand{\iwrt}{\int}
\providecommand{\ws}{\operatorname{ws}}
\providecommand{\bary}{\operatorname{bary}}
\providecommand{\cont}{\mathsf{K}}
\providecommand{\runit}{1_R}
\providecommand{\rcdot}{\cdot_R}
\providecommand{\vplus}{+_V}
\providecommand{\vzero}{0_V}
\providecommand{\vrdot}{\cdot_R^V}

\providecommand{\str}{\operatorname{st}}
\providecommand{\vect}[1]{\mathbf{#1}}

\title{Backpropagation for Effectful Languages I: Finite Probability and Discrete Output Algebraic Effects}
\author{Diogo Simm\\
  Utrecht University, The Netherlands\\
  \texttt{d.l.simmsallesvianna@uu.nl}
  \and
  Fernando Lucatelli Nunes\\
  Centre for Mathematics, University of Coimbra, Portugal\\
  \texttt{fln@uc.pt}
  \and
  Matthijs V\'ak\'ar\\
  Utrecht University, The Netherlands\\
  \texttt{m.i.l.vakar@uu.nl}}
\date{}

\begin{document}

\maketitle

\begin{abstract}
We analyse \textit{reverse-mode} automatic differentiation (AD) for discrete probabilistic programs. Our construction is formulated in the framework of Combinatory Homomorphic Automatic Differentiation (CHAD), treating AD as a structure-preserving transformation of programs, guided by a denotational semantics.

	The main case study is the finite atomic distribution monad, whose computations have finite support and differentiable weights. The key point is that differentiating probabilistic programs requires cotangents to flow backwards not only through deterministic computations, but also through the probabilistic structure itself. We define the corresponding reverse-mode code transformation and prove its correctness, for handled real-output programs, by a categorical logical-relations argument.

	Although the paper focuses on finite discrete probability, the construction gives a reusable pattern for differentiating discrete-output algebraic effects, including finite multiset non-determinism (e.g., from fork-join parallelism), exceptions, and writer-style accumulation (e.g., for in-place accumulation of high-dimensional vectors). More broadly, we view this work as a foundational step towards extending CHAD to richer probabilistic languages and to other algebraic effects with handlers.
\end{abstract}


\section{Introduction}

Probabilistic programming languages provide constructs for specifying
probabilistic models directly and for manipulating their statistical
properties within the language. Their promise is to make common tasks in
probabilistic data analysis, such as sampling, inference, and the optimisation
of expected objectives, more compositional and closer to the structure of the
models themselves.

Reverse-mode automatic differentiation (AD), often called
backpropagation, underlies much of modern optimisation and machine learning.
Its practical importance comes from the need to compute gradients of scalar or
low-dimensional objectives with respect to large numbers of parameters
\cite{baydinAutomaticDifferentiationMachine2018}.
Compared to the simpler forward-mode of AD, reverse-mode AD inverts data flow:
forward-mode propagates
sensitivities along the computation, whereas reverse-mode propagates and
accumulates cotangents against the direction of computation. This makes the
design and semantics of reverse-mode automatic differentiation substantially
more delicate.

Combinatory Homomorphic Automatic Differentiation (CHAD) gives a particularly
clean categorical account of reverse-mode AD. In the CHAD framework, automatic
differentiation is not specified as an \emph{ad hoc} syntactic traversal, but
as a semantics-driven program transformation: a structure-preserving map out of
the syntactic categorical model of the source language
\cite{lucatellinunesCHADExpressiveTotal2023,
	lucatellinunesUnravelingIterativeCHAD2025}. This account is not merely
explanatory, but it guides us in how to differentiate more complex forms of computation.
Once the relevant semantic structure has been
identified, the corresponding program transformation is obtained by following
the universal structure of the language. In the deterministic setting, this has
led both to correctness proofs by categorical logical relations and to efficient
implementations via familiar functional-programming techniques, such as
sparse representations and defunctionalisation or closure
conversion \cite{smedingEfficientCHAD2024}.

The goal of the present paper is to explain how the reverse-mode CHAD transformation
applies to programs with finite-support probabilistic effects and differentiable
weights. In this impure setting, backpropagation is not only a matter of propagating
cotangents through the deterministic input-output behaviour of a program. The
weights of probabilistic traces, and more generally the effectful structure of
 computations, become part of the reverse derivative communication
channel. We formulate this using monadic types and handlers: monadic operations
generate probabilistic computations, while handlers consume them, for instance
by taking expected values.

There is already foundational work on forward-mode automatic differentiation
for probabilistic programming, notably ADEV
\cite{lewADEVSoundAutomatic2023}. The present paper addresses the corresponding
reverse-mode problem. The semantic challenges of the two modes are different:
forward-mode propagates sensitivities with the computation, whereas reverse
mode must accumulate cotangent contributions backwards through both ordinary
program structure and probabilistic weights.
Our contribution is to give a
compositional reverse-mode transformation for this finite probabilistic
setting, formulated using monadic types and handlers, and to prove its
correctness by a categorical logical-relations argument. The CHAD framework is
essential here: it derives the transformation from the semantic structure
rather than retrofitting a correctness proof after the fact, while keeping
the construction close to implementation.

\subsection{Motivating example: expected-profit optimisation under finite probabilistic choice}
\label{sec:motivating-example}

To make the problem concrete, we begin with a small finite probabilistic
optimisation example. The model is deliberately simple, but it contains the
features that drive the paper: parameter-dependent probabilistic weights,
deterministic arithmetic, and a handler that consumes a probabilistic
computation by taking an expected value. Reverse-mode AD for such a program
must propagate cotangents through all of these components, not only through the
deterministic arithmetic. We first describe the model mathematically, and later
return to its encoding in the source language and to the corresponding code
transformation (Section~\ref{sec:returning-to-the-example}).

We consider the problem of introducing a product to a population. To obtain a
finite probabilistic model, we make the simplifying assumption that the
population is highly trend-following. There are then two possible outcomes:
either the trend does not catch on and no one buys the product, or it catches
on and everyone who can afford the product buys one.

The parameters of the model are the per-product profit \(p\) and the marketing
budget \(m\). Marketing influences the probability of the successful-trend
scenario, while the per-product profit determines how many people can purchase
the product if that scenario occurs.

In symbols, the probability of zero total sales is
\[
\fn[n]{\opn{sig}}{-s_m\cdot(m-m_0)},
\]
where \( \fn[n]{\opn{sig}}{x}=(1+e^{-x})^{-1} \) is the sigmoid function and
\(s_m,m_0\in\R\) are constants: \(s_m\) represents the model's
sensitivity to marketing, while \(m_0\) is the marketing budget at which the
success and failure scenarios are equally likely.
If, instead, the campaign is successful, the number of purchases is
\[
P\cdot \fn[n]{\opn{sig}}{-s_p\cdot(p-p_0)},
\]
where \(P\) is the total population, \(s_p\) is a known
constant controlling the sensitivity to changes in \(p\), and \(p_0\) is the
per-product profit at which the product is affordable to half the population.
Thus, by reducing the per-product profit, a larger portion of the population is
able to afford the product.

The expected total profit, which we intend to maximise, is therefore
\begin{equation*}
	\fn[n]{\opn{ETP}}{p,m}
	=
	p \cdot
	\underbrace{
		P
		\cdot \fn[n]{\opn{sig}}{-s_p\cdot(p-p_0)}
		\cdot \fn[n]{\opn{sig}}{s_m\cdot(m-m_0)}
	}_{\text{expected number of sales}}
	- m .
\end{equation*}
The stationarity conditions obtained by differentiating this expected objective
are a pair of coupled transcendental equations. Thus, even in this simple
example, optimisation is naturally performed numerically, for instance by
gradient descent.

\subsection{Contributions and Outline}

The example above illustrates the central semantic issue addressed in this
paper: probabilistic programs carry differentiable information not only through
their deterministic computations, but also through the effectful structure by
which probabilistic computations are generated and handled. The body of the
paper addresses this issue as follows.

First, we give a factorisation-based account of the finite atomic distribution
(FAD) monad. In \( \Set \), a FAD is a finitely supported nonnegative weight
function; we reconstruct it by factorising the weighted-sum strong monad
morphism induced by the barycentre algebra on lists of point/log-weight pairs
(Section~\ref{sec:barycentre-and-weighted-sums}). This is the bridge to the
differentiated semantics: in \( \Fam{\Vect[o]} \), a distribution \(w\) on a
space \(X\) with cotangent spaces \(A(x)\) has fibre
\[\bigoplus_{x\in\supp(w)}(A(x)\oplus\R),\] whose summands record cotangents
for the atom \(x\) and for its log-weight. Thus the lifted FAD monad captures
movement of atoms and variation of weights
(Sections~\ref{sec:generalising-the-set-procedure}
and~\ref{sec:fad-monad-on-fam-vectop}).

Second, we use this lift to extend reverse-mode CHAD to a first-order
probabilistic language with products, coproducts, Euclidean base types
\(\R[n]\), differentiable primitives, the monadic type constructor
\(\opn[b]{M}\), categorical distributions, and expectation-like handler
algebras. Its ordinary semantics lives in \( \Set \), its differentiated
semantics in \( \Fam{\Vect[o]} \), and the transformation is the induced
structure-preserving translation into a target language for derivatives.

Third, we prove correctness by logical relations. We construct a category
\(\cat[]{LR}\) relating the ordinary \( \Set \)-semantics of source programs to
the \( \Fam{\Vect[o]} \)-semantics of their differentiated programs. Its
projections recover, respectively, the ordinary source denotation and the CHAD
denotation. The resulting correctness theorem
(Theorem~\ref{th:correctness-theorem}) says that, for real-input/real-output
terms, the transformed program computes the original value and its reverse
derivative.

Fourth, we encode the motivating example in the source language using
categorical distributions, monadic composition, and an expectation handler. The
transformed program exhibits the intended cotangent flow through deterministic
arithmetic, atom locations, and log-weights.

Finally, although the paper focuses on finite discrete probability, the
construction also identifies a reusable pattern for discrete-output monadic
effects. Section~\ref{sec:beyond-probability} explains this general pattern,
and Appendix~\ref{sec:appendix-additional-effects} spells it out for finite
multiset non-determinism, exceptions, and \( \R[n] \)-writer accumulation.
As such, this paper makes important steps
towards both richer probabilistic-differentiable languages and towards general differentiable programming with algebraic effects and handlers.

\section{The Finite Atomic Distribution Monad}

The probabilistic effect studied in the main body of the paper is represented
by the \emph{finite atomic distribution} (\emph{FAD}) monad. For a set \(X\),
its elements are finitely supported, nonnegative weight distributions on \(X\).
Operationally, this is the monad of finite probabilistic computations: monadic
composition sums the weights of finite execution traces, while handlers, such
as expectation, consume the resulting weighted distribution. For reverse-mode
CHAD, however, the ordinary set-level monad is not enough. We need a lifted
version that also records how cotangents propagate through the probabilistic
structure itself, including both the locations of atoms and their weights.

This section develops the FAD monad in the form needed later in the paper. We
begin in \( \Set \), recalling its familiar description as the monad of
finitely supported nonnegative weight functions. We then give an equivalent
factorisation presentation in \( \Set \): lists of point/log-weight pairs serve
as raw weighted syntax, their weighted-sum semantics gives a strong monad
morphism, and factorising this morphism identifies precisely the presentations
that determine the same finite atomic distribution. The following subsection packages these data as a FAD factorisation
structure. This structure is then instantiated in \( \Fam{\Vect[o]} \), where
the associated FAD monad carries cotangent information for atoms and
log-weights, and in \( \cat[]{LR} \), where the same construction is used in
the logical-relations correctness argument.

\subsection{FAD Monad on \texorpdfstring{\( \Set \)}{Set}}
\label{sec:fad-monad-on-set}

We first recall the direct set-level description of the FAD monad. This is the
familiar finite-support distribution monad on \( \Set \), used here in its
unnormalised weighted form; it appears both in the probabilistic-programming
literature \cite{erwigProbabilisticFunctionalProgramming2006} and in
implementations \cite{erwigProbabilisticFunctionalProgramming2025}.

We use standard categorical notation: products and coproducts are written
\(1\), \(X_1\times X_2\), \( \prod_{i\in I}X_i \), and
\(0\), \(X_1\amalg X_2\), \( \coprod_{i\in I}X_i \), with projections
\(\pi_i\), inclusions \(\iota_i\), tuplings \( \tupling{f_i}_{i\in I} \), and
cotuplings \( \cotupling{f_i}_{i\in I} \). Exponentials are denoted by
\( \inthom{X}{Y} \) or by \( Y^{X} \).

A finite atomic distribution (FAD) on a set \(X\) is a finitely supported
nonnegative function \(w:X\to\R\). Explicitly, we define a FAD as the following notion of \emph{unnormalised weight distribution}:
\[
\fn[n]{w}{x}\geq 0 \ \text{for all }x\in X,
\qquad
\abs{\supp w}<\infty,
\qquad
\supp w:=\{x\in X\mid \fn[n]{w}{x}\neq 0\}.
\]
The set \( \supp w \) is the \emph{support} of \(w\); its elements are the
\emph{atoms} of the distribution, and \( \fn[n]{w}{x} \) is the
\emph{weight} of the atom \(x\). We write \( \Fun[nn]{M}{X} \) for the set of
finite atomic distributions on \(X\).

The functorial action pushes weights forward. For \(f:X\to Y\), define
\( \Fun[nn]{M}{f}:\Fun[nn]{M}{X}\to\Fun[nn]{M}{Y} \) by
\[
\fn[n]{\Fun[nn]{M}{f}}{w}(y)
:=
\sum_{\substack{x\in\supp w\\ \fn[n]{f}{x}=y}}
\fn[n]{w}{x}.
\]
The unit \( \Nat[nn]{\eta}{X}:X\to\Fun[nn]{M}{X} \) sends each point to its
Dirac distribution, that is,
\( \fn[n]{\Nat[nn]{\eta}{X}}{x}(x')=1 \) if \(x=x'\), and
\( \fn[n]{\Nat[nn]{\eta}{X}}{x}(x')=0 \) otherwise. The multiplication
\( \Nat[nn]{\mu}{X}:\Fun[nn]{M^2}{X}\to\Fun[nn]{M}{X} \), where
\(M^2:=M\circ M\), flattens a FAD of FADs by summing all contributions to each
point:
\[
\fn[n]{\Nat[nn]{\mu}{X}}{\omega}(x)
:=
\sum_{\substack{w\in\supp\omega\\ x\in\supp w}}
\fn[n]{\omega}{w}\cdot\fn[n]{w}{x}.
\]
The maps \( \eta:\Id_{\Set}\to M \) and \( \mu:M^2\to M \) are natural and
equip \(M\) with a monad structure.

This completes the direct description of the FAD monad on \( \Set \). We next
give a second presentation of the same monad, not by changing the object being
defined, but by changing how it is obtained: finite atomic distributions are
recovered as the image of a weighted-sum map from ordered point/log-weight
presentations. This factorisation presentation is the one that admits the
axiomatic formulation used in the rest of the section.

\subsection{Recovering FAD by Weighted-Sum Factorisation}
\label{sec:barycentre-and-weighted-sums}

The preceding subsection described \( \Fun[nn]{M}{X} \) directly as the set of
finitely supported nonnegative weight functions on \(X\). We now recover the
same object as the image of a weighted-sum transformation. The domain of this
transformation is the ordered-list representation of weighted finite support:
a list of point/log-weight pairs over \(X\) acts on a test function
\(f:X\to\R\) by evaluating \(f\) at the listed points and summing with weights
\( \exp\gamma_j \). Its image factorisation has \( \Fun[nn]{M}{X} \) as the
middle object: the left factor forms the corresponding finite atomic
distribution by adding the weights of coincident points, while the right
factor embeds such a distribution as its integration functional. This is the
calculation in \( \Set \) whose categorical content is isolated in the next
subsection.

We begin with the real-valued case, which provides the barycentre algebra used
below. For a list of point/log-weight pairs
\(
\tuple{
	\tuple{t_{1}, \gamma_{1}},
	\cdots,
	\tuple{t_{n}, \gamma_{n}}
}
\in \coprod_{n \in \N} \left( \R \times \R \right)^{n}\),
where \( t_{j},\gamma_{j}\in\R \), define its unnormalised barycentre by\footnote{A normalised barycentre would require division by the total mass. Here we deliberately omit this normalisation.}
\begin{equation}
	\label{eq:barycentre-in-set}
	\fn[n]{ \bary }{
		\tuple{
			\tuple{t_{j}, \gamma_{j}}
		}_{j = 1}^{n}
	} :=
	\sum_{j = 1}^{n} t_{j} \cdot \exp \gamma_{j}.
\end{equation}
Here \( \gamma_{j} \) is a log-weight, and \( \exp\gamma_{j} \) is the
corresponding positive weight assigned to \(t_j\). If the list contains
\emph{collisions}, meaning that \( t_{i}=t_{j}=:t \) for distinct
\( i,j\in\{1,\cdots,n\} \), then the total weight assigned to \(t\) is the sum
of \( \exp\gamma_{k} \) over all \(k\) such that \(t_k=t\).

For an arbitrary set \(X\), we write
\[
\Fun[nn]{ \AST }{ X }
:=
\coprod_{n \in \N} \left( X \times \R \right)^{n}
\]
for the set of lists of point/log-weight pairs over \(X\). Given a function
\( f:X\to\R \), define
\begin{equation}
	\label{eq:definition-ast-in-set}
	\begin{aligned}
		\Fun[nn]{ \AST }{ f }:
		\Fun[nn]{ \AST }{ X } &\to \Fun[nn]{ \AST }{ \R };\qquad
		\tuple{\tuple{x_{j}, \gamma_{j}}}_{j = 1}^{n}
		\mapsto
		\tuple{\tuple{\fn[n]{ f }{ x_{j} }, \gamma_{j}}}_{j = 1}^{n}.
	\end{aligned}
\end{equation}
Here \( \AST \) stands for \emph{abstract syntax tree}, for reasons explained
in Section~\ref{sec:generalising-the-set-procedure}. Thus
\( \Fun[nn]{ \AST }{ f } \) acts on a raw weighted presentation by applying
\(f\) to its points and leaving the log-weights unchanged. Composing with the
barycentre gives the weighted-sum map \( \Nat[nn]{ \ws }{ X } \):
\begin{equation}
	\label{eq:weighted-sum-definition-in-set}
	\begin{aligned}
		\Nat[nn]{ \ws }{ X }:
		\Fun[nn]{ \AST }{ X }
		&\to \inthom{\inthom{X}{\R}}{\R}
		\\
		\tuple{\tuple{x_{j}, \gamma_{j}}}_{j = 1}^{n}
		&\mapsto
		\left(
		f \mapsto  \fn[p]{ \bary \circ \Fun[nn]{ \AST }{ f } }{
			\tuple{\tuple{x_{j}, \gamma_{j}}}_{j = 1}^{n}
		}
		=  f \mapsto \sum_{j = 1}^{n} \fn[n]{ f }{ x_{j} } \cdot \exp \gamma_{j}
		\right).
	\end{aligned}
\end{equation}

We factorise \( \Nat[nn]{ \ws }{ X } \) through
\( \Fun[nn]{ M }{ X } \). Define maps
\( \Nat[nn]{ \categ }{ X } \) and \( \Nat[nn]{ \iwrt }{ X } \) such that
\begin{equation*}
	\Nat[nn]{ \ws }{ X } = \Nat[nn]{ \iwrt }{ X } \circ \Nat[nn]{ \categ }{ X }:
	\Fun[nn]{ \AST }{ X }
	\to
	\Fun[nn]{ M }{ X }
	\to
	\inthom{\inthom{X}{\R}}{\R}.
\end{equation*}
The left factor \( \Nat[nn]{ \categ }{ X } \) sends a raw
presentation to the corresponding finite atomic distribution, accumulating the
weights of all occurrences of each point. For any
\( \tuple{\tuple{x_{j}, \gamma_{j}}}_{j = 1}^{n} \in \Fun[nn]{ \AST }{ X } \),
define
\begin{equation}
	\label{eq:categorical-definition-in-set}
	\fn[n]{ \Nat[nn]{ \categ }{ X } }{
		\tuple{\tuple{x_{j}, \gamma_{j}}}_{j = 1}^{n}
	}:
	X \to \R
	\qq{s.t.}
	x \mapsto \sum_{\substack{1 \le j \le n \\ x_{j} = x}} \exp \gamma_{j}.
\end{equation}
Its support is exactly \( \{x_{1},\cdots,x_{n}\} \), since the displayed sum
is zero precisely when it is empty. The name \( \categ \) reflects the standard
probability theory nomenclature of finitely supported weight distributions.
\( \Nat[nn]{ \categ }{ X } \) is surjective, but not injective; it forgets about ordering and collisions.

The right factor \( \Nat[nn]{ \iwrt }{ X } \) sends a finite atomic
distribution to the corresponding integration functional on test functions.
For \( w \in \Fun[nn]{ M }{ X } \), define
\begin{equation}
	\label{eq:iwrt-definition-in-set}
	\fn[n]{ \Nat[nn]{ \iwrt }{ X } }{ w }:
	\inthom{X}{\R} \to \R
	\qq{s.t.}
	f \mapsto \sum_{x \in \supp w} \fn[n]{ f }{ x } \cdot \fn[n]{ w }{ x }.
\end{equation}
This map is injective because indicator functions separate finite atomic
distributions: for each \(x\in X\),
\(
\fn[p]{ \fn[n]{ \Nat[nn]{ \iwrt }{ X } }{ w } }{
	\indicator{\left\{ x \right\}}
}
=
\fn[n]{w}{x}.
\)
Hence, for \( w,w'\in\Fun[nn]{M}{X} \),
\( \fn[n]{ \Nat[nn]{ \iwrt }{ X } }{ w }
=
\fn[n]{ \Nat[nn]{ \iwrt }{ X } }{ w' } \) if and only if \( w=w' \).

The composite has the required weighted-sum semantics. Indeed, for any list
\( \tuple{\tuple{x_{j}, \gamma_{j}}}_{j = 1}^{n}
\in \Fun[nn]{ \AST }{ X } \) and any function \( f:X\to\R \),
\begin{equation*}
	\begin{aligned}
		\fn[p]{ \fn[p]{ \Nat[nn]{ \iwrt }{ X } \circ \Nat[nn]{ \categ }{ X } }{ \tuple{\tuple{x_{j}, \gamma_{j}}}_{j = 1}^{n} } }{ f }
		&=
		\sum_{x \in \left\{ x_{1}, \cdots, x_{n} \right\} }
		\left(
		\fn[n]{ f }{ x } \cdot
		\sum_{\substack{1 \le j \le n \\ x_{j} = x}} \exp \gamma_{j}
		\right)
		\\
		&=
		\sum_{j = 1}^{n} \fn[n]{ f }{ x_{j} } \cdot \exp \gamma_{j}=
		\fn[p]{ \fn[n]{ \Nat[nn]{ \ws }{ X } }{ \tuple{\tuple{x_{j}, \gamma_{j}}}_{j = 1}^{n} } }{ f }.
	\end{aligned}
\end{equation*}
Thus \( \Nat[nn]{ \ws }{ X } = \Nat[nn]{ \iwrt }{ X } \circ
\Nat[nn]{ \categ }{ X } \). Since \( \Nat[nn]{ \categ }{ X } \) is
surjective and \( \Nat[nn]{ \iwrt }{ X } \) is injective, this is the
\emph{image factorisation} of \( \Nat[nn]{ \ws }{ X } \) in \( \Set \).

The calculation has isolated the three ingredients that we will abstract
next. The ordered lists form the raw weighted-syntax monad \( \AST \); the
barycentre algebra induces the weighted-sum morphism into a continuation
monad; and the FAD monad on \( \Set \) is recovered as the middle object of
the resulting image factorisation. We next package these
ingredients as a FAD factorisation structure on a category \( \cat[]{X} \) and
define the associated FAD monad on \( \cat[]{X} \).

\subsection{FAD Monads by Axiomatic Factorisation}
\label{sec:generalising-the-set-procedure}

We now abstract the factorisation just described in \( \Set \). The aim is to
specify the structure on a category \( \cat[]{X} \) that is sufficient to form
weighted raw syntax, interpret it by a barycentre algebra, and factorise the
resulting strong monad morphism. Once these data are fixed, the FAD monad on
\( \cat[]{X} \) is not additional primitive structure: it is the monad obtained
as the middle term of the factorisation
\[
\AST^{\cat[]{X}}
\xrightarrow{\;\categ^{\cat[]{X}}\;}
M^{\cat[]{X}}
\xrightarrow{\;\iwrt^{\cat[]{X}}\;}
\cont{V}^{\cat[]{X}} .
\]
For \( \cat[]{X}=\Set \), with the choices specified below, this construction
recovers the FAD monad on \( \Set \).

We assume that \( \cat[]{X} \) is bicartesian closed with countable
coproducts. Thus we fix a terminal object \(1\), an initial object \(0\),
binary products and coproducts, countable coproducts, and, for each object
\(A\), a right adjoint
\(
\inthom{A}{\blankdash}:\cat[]{X}\to\cat[]{X}
\)
to \( (\blankdash)\times A:\cat[]{X}\to\cat[]{X} \).

We also fix a monoid object \(R\) in \( \cat[]{X} \), with unit
\( \runit:1\to R \) and multiplication \( \rcdot:R\times R\to R \). This is
the \emph{monoid of weights}; in the motivating instance, it is the additive
monoid of log-weights. Finally, we fix an \(R\)-semimodule \(V\), called the
\emph{barycentre semimodule}, with structure maps
\[
\vplus:V\times V\to V,
\qquad
\vzero:1\to V,
\qquad
\vrdot:V\times R\to V .
\]
Here an \(R\)-semimodule means a commutative monoid object \(V\) equipped with a right action \(V\times R\to V\) of the monoid object \(R\) by commutative-monoid endomorphisms.

In \( \Set \), the relevant choices are \(R=\R\) (thought of as log-weights), with
\begin{equation}
	\label{eq:monoid-of-weights-operations-in-set}
	\begin{aligned}
		\runit: 1 &\to \R
		\\
		\ast &\mapsto 0
	\end{aligned}
	\qq{and}
	\begin{aligned}
		\rcdot: \R \times \R &\to \R
		\\
		\tuple{a, b} &\mapsto a + b,
	\end{aligned}
\end{equation}
and \(V=\R\) (thought of as real values), with
\begin{equation}
	\label{eq:barycentre-semimodule-operations-in-set}
	\begin{aligned}
		\vzero: 1 &\to \R
		\\
		\ast &\mapsto 0
	\end{aligned}
	\qq{and}
	\begin{aligned}
		\vplus: \R \times \R &\to \R
		\\
		\tuple{a, b} &\mapsto a + b
	\end{aligned}
	\qq{and}
	\begin{aligned}
		\vrdot: \R \times \R &\to \R
		\\
		\tuple{v, r} &\mapsto v \cdot \exp r .
	\end{aligned}
\end{equation}
The case \(V=\R[p]\), for \(p\in\N\), is obtained in the same way.

From \(R\) and \(V\) we form two monads. The raw weighted-syntax monad and the
continuation monad on \(V\) have underlying endofunctors
\begin{equation}
	\label{eq:definition-ast-continuation-general}
	\AST^{\cat[]{X}} := \coprod_{n \in \N} \left( \blankdash \times R \right)^{n}
	\qq{and}
	\cont{V}^{\cat[]{X}} := \inthom{\inthom{\blankdash}{V}}{V}.
\end{equation}
The monad structure on \( \AST^{\cat[]{X}} \) is obtained by combining the
ordered-list monad \( A\mapsto\coprod_{n\in\N}A^n \)
\cite{dubucFreeMonoids1974} with the \(R\)-writer monad
\(A\mapsto A\times R\), induced by the monoid
\( \tuple{R,\rcdot,\runit} \); see, for instance, Chapter~7 of
\cite{maclaneCategoriesWorkingMathematician1997}. These monads are combined
using the canonical distributive law recalled in
Section~6.3 of \cite{bacciSumTensorQuantitative2024}. Thus
\( \AST^{\cat[]{X}} \) is an ordered-list presentation of weighted finite
support. The identifications imposed by permutations and collisions are not
built into \( \AST^{\cat[]{X}} \); they are produced by factorisation.

Both \( \AST^{\cat[]{X}} \) and \( \cont{V}^{\cat[]{X}} \) are strong monads.
For \( \AST^{\cat[]{X}} \), the strength is induced by distributivity of
products over countable coproducts and by the canonical maps \( A \times B^{n} \to A^{n} \times B^{n} \xrightarrow{\sim} \left( A \times B \right)^{n} \) obtained from the diagonal morphism and tuplings of projections. Then, by using \(B\times R\) in place of \(B\), this gives a natural
transformation
\begin{equation*}
A \times \Fun[nn]{ \AST^{\cat[]{X}} }{ B }
\to
\Fun[np]{ \AST^{\cat[]{X}} }{ A \times B }.
\end{equation*}
The continuation monad carries its standard strength in every cartesian closed
category.

The semimodule structure on \(V\) defines an
\( \AST^{\cat[]{X}} \)-algebra structure on \(V\), the
\emph{barycentre algebra}. Explicitly, it is the morphism
\begin{equation}
	\label{eq:definition-barycentre-algebra-general}
\bary^{\cat[]{X}}:
\Fun[nn]{ \AST^{\cat[]{X}} }{ V }
\to V;
\quad
\textnormal{defined by}
\quad
	\bary^{\cat[]{X}}
	=
	\cotupling{\underset{n}{\Sigma} \circ \left( \vrdot \right)^{n}}_{n \in \N},
\end{equation}
where \( \underset{n}{\Sigma}:V^n\to V \) is the \(n\)-ary sum determined by
\( \vplus \) and \( \vzero \), and
\( \left(\vrdot\right)^n:(V\times R)^n\to V^n \) is the product of \(n\)
copies of \( \vrdot \).

Under continuation-passing style, the barycentre algebra corresponds to a strong monad morphism $\ws^{\cat[]{X}}$ to the continuation monad $\inthom{\inthom{-}{V}}{V}$ whose \(A\)-component sends \(a \in \AST^{\cat[]{X}}A\) to \(\lambda k.\,\bary^{\cat[]{X}}\!\left(\AST^{\cat[]{X}}k(a)\right)\).
\begin{lemma}[Algebra CPS; see Appendix~\ref{sec:appendix-algebras-and-continuations}]
	\label{lem:algebras-and-strong-monad-morphisms-to-the-continuation-monad}
	Let \( \tuple{T, \mu^{T}, \eta^{T}, \str^{T}} \) be a strong monad on a
	cartesian closed category \( \cat[]{X} \), and let \( V \) be an object of
	\( \cat[]{X} \). Then there is a 1:1-correspondence between
	\( T \)-algebra structures on \( V \) and strong monad morphisms
	\( T \to \cont{V}^{\cat[]{X}} \).
\end{lemma}

Applying Lemma~\ref{lem:algebras-and-strong-monad-morphisms-to-the-continuation-monad}
to the barycentre algebra gives a strong monad morphism
\[
\ws^{\cat[]{X}}:
\AST^{\cat[]{X}}
\to
\cont{V}^{\cat[]{X}}
\qquad\textnormal{with components}\qquad
\Nat[nn]{ \ws^{\cat[]{X}} }{ A }:
\Fun[nn]{ \AST^{\cat[]{X}} }{ A }
\to
\inthom{\inthom{A}{V}}{V}.
\]

It remains to specify the factorisation. In \( \Set \), we
used the image factorisation. In \( \cat[]{X} \), we fix an orthogonal
factorisation system \( \tuple{\mathcal{E},\mathcal{M}} \). Thus every
 \(f:A\to B\) factors as \(f=m\circ e\), with
\(e\in\mathcal{E}\) and \(m\in\mathcal{M}\), and
\( \mathcal{E} \) is orthogonal to \( \mathcal{M} \): for every equation
\(v\circ e=m\circ u\), with \(e\in\mathcal{E}\) and \(m\in\mathcal{M}\), there
is a unique \(w\) such that \(u=w\circ e\) and \(v=m\circ w\). In \( \Set \),
taking \( \mathcal{E} \) to be the epis (i.e., surjections) and \( \mathcal{M} \) to be the
(regular) monos (i.e., injections) recovers the usual image factorisation.
We assume that \( \mathcal{E} \) is closed under binary products. Since
left classes of orthogonal factorisation systems are closed under coproducts,
this implies closure of \( \mathcal{E} \) under \( \AST^{\cat[]{X}} \).

For each object \(A\), factorise the component
\( \Nat[nn]{\ws^{\cat[]{X}}}{A} \) as
\[
\begin{tikzcd}[column sep = 2 cm]
	\Fun[nn]{ \AST^{\cat[]{X}} }{ A }
	&&
	\Fun[nn]{ \cont{V}^{\cat[]{X}} }{ A }
	\\
	& \Fun[nn]{ M^{\cat[]{X}} }{ A } &
	\arrow[from = 1-1, to = 1-3, "{\Nat[nn]{ \ws^{\cat[]{X}} }{ A }}"]
	\arrow[from = 1-1, to = 2-2, swap, "{\Nat[nn]{ \categ^{\cat[]{X}} }{ A }}"]
	\arrow[from = 2-2, to = 1-3, swap, "{\Nat[nn]{ \iwrt^{\cat[]{X}} }{ A }}"]
\end{tikzcd}
\]
with \( \Nat[nn]{ \categ^{\cat[]{X}} }{ A }\in\mathcal{E} \) and
\( \Nat[nn]{ \iwrt^{\cat[]{X}} }{ A }\in\mathcal{M} \). At this stage
\( \Fun[nn]{ M^{\cat[]{X}} }{ A } \) denotes only the intermediate objects of
the factorisation. The next theorem ensures that these objects
assemble into a strong monad.

\begin{lemma}[Factorisation of strong monad morphisms, {\cite[Theorem~2.5]{kammarFactorisationSystemsLogical2018}}]\label{lem:factorisation-of-strong-monad-morphisms}
	Let \( \cat[]{X} \) be a category, \( \tuple{\mathcal{E},\mathcal{M}} \) a
	factorisation system, \( S \) and \( T \) monads on \( \cat[]{X} \), and
	\( m:S\to T \) a monad morphism. Denote by \( m[S] \) the functor obtained
	by pointwise factorisation of \(m\), with \( m^{\mathbf{e}} \) and
	\( m^{\mathbf{m}} \) denoting the pointwise left- and right-class factors.
	If \( \mathcal{E} \) is closed under \(S\), then \(m[S]\) is a monad, and
	\( m^{\mathbf{e}} \) and \( m^{\mathbf{m}} \) are monad morphisms. If,
	moreover, \( \mathcal{E} \) is closed under binary products, \(S\) and \(T\)
	are strong monads, and \(m\) is a strong monad morphism, then \(m[S]\) is a
	strong monad and \( m^{\mathbf{e}} \) and \( m^{\mathbf{m}} \) are strong
	monad morphisms.
\end{lemma}

We can now name the structure and the monad it determines.

\begin{definition}[FAD factorisation structure and associated FAD monad]
	A \emph{FAD factorisation structure} on \( \cat[]{X} \) consists of the
	bicartesian closed structure with countable coproducts fixed above, a monoid
	of weights \( \tuple{R,\rcdot,\runit} \), a barycentre semimodule
	\( \tuple{V,\vplus,\vzero,\vrdot} \), and an orthogonal factorisation system
	\( \tuple{\mathcal{E},\mathcal{M}} \) whose left class is closed under binary
	products.

	Its associated \emph{FAD monad} is the
	monad \(M^{\cat[]{X}}\) obtained by the pointwise factorisation
  	\[
	\AST^{\cat[]{X}}
	\xrightarrow{\;\categ^{\cat[]{X}}\;}
	M^{\cat[]{X}}
	\xrightarrow{\;\iwrt^{\cat[]{X}}\;}
	\cont{V}^{\cat[]{X}}
	\qquad\text{of}\qquad
	\ws^{\cat[]{X}}:
	\AST^{\cat[]{X}}
	\to
	\cont{V}^{\cat[]{X}} .
	\]

\end{definition}

\begin{proposition}
	For every FAD factorisation structure on \( \cat[]{X} \), the associated
	FAD monad \(M^{\cat[]{X}}\) is strong, and the maps
	\(
	\categ^{\cat[]{X}}:
	\AST^{\cat[]{X}}\to M^{\cat[]{X}}\) and
	\(\iwrt^{\cat[]{X}}:
	M^{\cat[]{X}}\to \cont{V}^{\cat[]{X}}
	\)
	are strong monad morphisms.
\end{proposition}

\begin{proof}
	Apply Lemma~\ref{lem:factorisation-of-strong-monad-morphisms} to
	\( \ws^{\cat[]{X}}:\AST^{\cat[]{X}}\to\cont{V}^{\cat[]{X}} \). The required
	closure of \( \mathcal{E} \) under \( \AST^{\cat[]{X}} \) follows from
	closure under binary products and from closure of left classes under
	coproducts.
\end{proof}

With \( \cat[]{X}=\Set \), \(R=\R\), \(V=\R\), and the
(epi, regular mono) factorisation, the associated FAD monad
\(M^{\Set}\) is canonically isomorphic to the FAD monad on \( \Set \). The
point of the preceding abstraction is that the same construction can now be
instantiated in categories with additional semantic structure. The first
such instance is \( \Fam{\Vect[o]} \), where the associated FAD monad refines
the FAD monad on the underlying set by adding cotangent fibres for atom
locations and log-weights.

\subsection{FAD Monad on \texorpdfstring{\( \Fam{\Vect[o]} \)}{Fam(VectOp)}}
\label{sec:fad-monad-on-fam-vectop}

We now exhibit a FAD factorisation structure on
\( \cat[]{F}:=\Fam{\Vect[o]} \), and write \(M^{\cat[]{F}}\) for its associated
FAD monad. This is the semantic category for differentiated values in
reverse-mode CHAD: an object carries a set of possible values together with a
specified cotangent space over each value, and a morphism carries both the
underlying function on values and the corresponding contravariant action on
cotangents.
We think of this action as a transposed derivative.

The category \( \Fam{\Vect[o]} \) is the free coproduct completion of
\( \Vect[o] \)~\cite{adamekHowNiceAre2020}. We use it as the category of sets
equipped with cotangent fibres. An object is a pair \( \tuple{X,A} \), where
\(X\) is a set and \(A:X\to\Vect[o]\) is a functor, viewing \(X\) as a
discrete category. Equivalently, \(A\) is a family
\( \{\Fun[np]{A}{x}\}_{x\in X} \) of vector spaces indexed by \(X\).

A morphism \( \tuple{X,A}\to\tuple{Y,B} \) is a pair
\( \tuple{f,\alpha} \), where \(f:X\to Y\) is a function and
\( \alpha:A\to B\circ f \) is a natural transformation in \( \Vect[o] \).
Equivalently, for each \(x\in X\), it determines a linear map
\[
\Nat[nn]{\alpha}{x}:
\Fun[np]{B}{\fn[n]{f}{x}}
\to
\Fun[np]{A}{x}.
\]
Thus \(f\) acts on base points, while \( \Nat[nn]{\alpha}{x} \) acts
contravariantly on cotangents. In the intended interpretation,
\( \Fun[np]{A}{x} \) is the cotangent space at \(x\), and
\( \Nat[nn]{\alpha}{x} \) plays the role of the transposed derivative of \(f\)
at \(x\). We write
\(
\base{\tuple{X,A}}=X,
\quad
\fibr{\tuple{X,A}}=A,
\quad
\base{\tuple{f,\alpha}}=f,
\quad
\fibr{\tuple{f,\alpha}}=\alpha
\)
for the base and fibre parts. We use this notation throughout the
construction below.

The category \( \Fam{\Vect[o]} \) is complete and cocomplete. In particular,
for any family \( \{\tuple{X_i,A_i}\}_{i\in I} \) of objects, products and
coproducts are given by
\begin{equation}
	\label{eq:definition-product-coproduct-fam-vectop}
	\prod_{i \in I} \tuple{ X_{i}, A_{i} } =
	\tuple{
		\left(
		\prod_{i \in I} X_{i}
		\right),
		\left(
		\left\{ x_{i} \right\}_{i \in I}
		\mapsto
		\bigoplus_{i \in I} \Fun[nn]{ A_{i} }{ x_{i} }
		\right)
	}
	\qq{and}
	\coprod_{i \in I} \tuple{ X_{i}, A_{i} } =
	\tuple{
		\left(
		\coprod_{i \in I} X_{i}
		\right),
		\cotupling{A_{i}}_{i \in I}
	}.
\end{equation}
Moreover, \( \Fam{\Vect[o]} \) has exponentials
\cite[Proposition~74]{lucatellinunesCHADExpressiveTotal2023}:
\([\tuple{X,A},\tuple{Y,B}] = \tuple{\prod_{x\in X}\sum_{y \in Y}\Vect(\Fun[nn]{B}{y}, \Fun[nn]{A}{x}),f\mapsto \bigoplus_{x\in X}\Fun[np]{B}{
\fn[n]{\pi_1\!}{\fn[n]{f\!}{x}}}}\) .
Hence \( \cat[]{F}=\Fam{\Vect[o]} \) is bicartesian closed with countable
coproducts, as required by the axiomatic construction.

We use the (epi, regular mono) factorisation system on \( \Fam{\Vect[o]} \).
It admits the following explicit description.

\begin{lemma}
	Epis and regular monos form a factorisation system on \( \Fam{\Vect[o]} \).
	For a morphism
	\( \tuple{f,\alpha}:\tuple{X,A}\to\tuple{Y,B} \), membership in the left
	class \( \mathcal{E} \) is equivalent to \(f:X\to Y\) being surjective and
	\( \alpha \) being jointly injective, in the sense that, for each \(y\in Y\),
	the linear map
	\[
	\tupling{\Nat[nn]{ \alpha }{ x }}_{\substack{x \in X \\ \fn[n]{ f }{ x } = y}}:
	\Fun[nn]{ B }{ y }
	\to
	\prod_{\substack{x \in X \\ \fn[n]{ f }{ x } = y}}
	\Fun[nn]{ A }{ x }
	\]
	is injective. Membership in the right class \( \mathcal{M} \) is equivalent
	to \(f:X\to Y\) being injective and \( \alpha \) being pointwise surjective,
	in the sense that, for each \(x\in X\), the linear map
	\[
	\Nat[nn]{ \alpha }{ x }:
	\Fun[np]{ B }{ \fn[n]{ f }{ x } }
	\to
	\Fun[nn]{ A }{ x }
	\]
	is surjective. Moreover, the left class \( \mathcal{E} \) is closed under
	binary products.
\end{lemma}

Therefore \( \cat[]{F} \) has the categorical structure and factorisation
system required for a FAD factorisation structure. It remains to choose the
monoid of weights and the barycentre semimodule. Both are obtained by lifting
the corresponding \( \Set \)-structures to the real cotangent bundle
\[
\TR := \tuple{\R,\const\R},\qquad\text{where}\qquad \Fun[pn]{\const Y}{\ x}\ =\ Y.
\]

The monoid of weights is \( \TR \). Its unit and multiplication are
\begin{equation}
	\label{eq:monoid-of-weights-operations-in-fam-vectop}
	\begin{array}{lrcl}
		\runit &: &\tuple{ 1, \const 0 } &\to \TR
		\vspace{6pt}
		\\
		\base{\runit} &:& 1 &\to \R
		\\
		&&\ast &\mapsto 0
		\vspace{6pt}
		\\
		\Nat[nn]{ \fibr{\runit} }{\ast} &: &\R &\to 0
		\\
		&& v &\mapsto 0
	\end{array}
	\qquad
	\begin{array}{lccl}
		\rcdot &: &\TR \times \TR &\to \TR
		\vspace{6pt}
		\\
		\base{\rcdot} &:& \R \times \R &\to \R
		\\
		&&\tuple{ a, b } &\mapsto a + b
		\vspace{6pt}
		\\
		\Nat[nn]{ \fibr{\rcdot} }{ \tuple{ a, b } } &: &\R &\to \R \oplus \R
		\\
		&&v &\mapsto \tuple{ v, v } .
	\end{array}
\end{equation}
The base parts agree with
Equation~\eqref{eq:monoid-of-weights-operations-in-set}; the fibre parts are
the corresponding transposed derivatives.

We also take the barycentre semimodule to be \( \TR \).\footnote{As in
	\( \Set \), one may instead take \( \TR[p] \) for some \( p\in\N \).} Its zero
element and addition are
\begin{equation}
	\label{eq:barycentre-semimodule-operations-in-fam-vectop}
	\begin{array}{lrcl}
		\vzero &: &\tuple{ 1, \const 0 } &\to \TR
		\vspace{6pt}
		\\
		\base{\vzero} &:& 1 &\to \R
		\\
		&&\ast &\mapsto 0
		\vspace{6pt}
		\\
		\Nat[nn]{ \fibr{\vzero} }{\ast} &: &\R &\to 0
		\\
		&&v &\mapsto 0
	\end{array}
	\qquad
	\begin{array}{lrcl}
		\vplus &: &\TR \times \TR &\to \TR
		\vspace{6pt}
		\\
		\base{\vplus} &:& \R \times \R &\to \R
		\\
		&&\tuple{ a, b } &\mapsto a + b
		\vspace{6pt}
		\\
		\Nat[nn]{ \fibr{\vplus} }{ \tuple{ a, b } } &: &\R &\to \R \oplus \R
		\\
		&&v &\mapsto \tuple{ v, v } .
	\end{array}
\end{equation}
Again, the base parts are those of
Equation~\eqref{eq:barycentre-semimodule-operations-in-set}, and the fibre
parts are their transposed derivatives. The action of the monoid of weights on
the barycentre semimodule is
\begin{equation}
	\label{eq:barycentre-semimodule-monoid-action-in-fam-vectop}
	\begin{array}{lrcl}
		\vrdot &: &\TR \times \TR &\to \TR
		\vspace{6pt}
		\\
		\base{\vrdot} &:& \R \times \R &\to \R
		\\
		&&\tuple{ v, r } &\mapsto v \cdot \exp r
		\vspace{6pt}
		\\
		\Nat[nn]{ \fibr{\vrdot} }{ \tuple{ v, r } } &: &\R &\to \R \oplus \R
		\\
		&&t &\mapsto \tuple{ t \cdot \exp r, t \cdot v \cdot \exp r } .
	\end{array}
\end{equation}
The base part is the action used in \( \Set \), while the fibre part is its
transposed derivative.

These choices make \( \cat[]{F}=\Fam{\Vect[o]} \) into a FAD factorisation
structure. Let \(M^{\cat[]{F}}\) denote the associated FAD monad. For an object
\( \tuple{X,A} \), the factorisation construction gives
\begin{equation}
	\label{eq:fad-monad-on-objects-in-fam-vectop}
	\Fun[nn]{ M^{\cat[]{F}} }{ \tuple{ X, A } } :=
	\tuple{
		\Fun[nn]{ M }{ X },
		\left(
		w \mapsto \bigoplus_{x \in \supp w} \left( \Fun[np]{ A }{ x } \oplus \R \right)
		\right)
	}.
\end{equation}
Thus the base of \( \Fun[nn]{ M^{\cat[]{F}} }{ \tuple{X,A} } \) is the set of
finite atomic distributions on \(X\). The fibre over
\( w\in\Fun[nn]{M}{X} \) has one summand for each atom \(x\in\supp w\): the
component \( \Fun[np]{A}{x} \) carries cotangent information for the position
of the atom, while the real component carries cotangent information for
\( \log \fn[n]{w}{x} \).

For a morphism
\( \tuple{f,\alpha}:\tuple{X,A}\to\tuple{Y,B} \), the base part is
\[
\base{ \Fun[nn]{ M^{\cat[]{F}} }{ \tuple{ f, \alpha } } }
=
\Fun[nn]{ M }{ f } .
\]
For \( w\in\Fun[nn]{M}{X} \), the fibre part is
\begin{equation}
	\label{eq:fad-monad-on-morphisms-in-fam-vectop}
	\begin{aligned}
		\Nat[nn]{ \fibr{\Fun[nn]{ M^{\cat[]{F}} }{ \tuple{ f, \alpha } }} }{ w }:
		\bigoplus_{y \in \fn[n]{ f }{ \supp w }} \left( \Fun[np]{ B }{ y } \oplus \R \right)
		&\to
		\bigoplus_{x \in \supp w} \left( \Fun[np]{ A }{ x } \oplus \R \right)
		\\
		\sum_{y \in \supp \fn[n]{ \Fun{ M }{ f } }{ w }}
		\left(
		b_{y} \cdot \vect{ e }_{y}
		+
		s_{y} \cdot \vect{ \varepsilon }_{y}
		\right)
		&\mapsto
		\sum_{x \in \supp w}
		\left(
		\frac{
			\fn[n]{ w }{ x } \cdot
			\left(
			\fn[n]{ \Nat{ \alpha }{ x } }{ b_{\fn[n]{ f }{ x }} }
			\cdot \vect{ e }_{x}
			+
			s_{\fn[n]{ f }{ x }}
			\cdot \vect{ \varepsilon }_{x}
			\right)
		}{
			\fn[p]{ \fn[n]{ \Fun{ M }{ f } }{ w } }{ \fn[n]{ f }{ x } }
		}
		\right) .
	\end{aligned}
\end{equation}
Here \( \vect{e}_{\bullet} \) denotes the atom-location component and
\( \vect{\varepsilon}_{\bullet} \) the log-weight component in the relevant
direct sum. The formula says that cotangents are pulled back along
\( \tuple{f,\alpha} \) by distributing the cotangent at an image atom among
all atoms mapping to it, with the contribution of \(x\) scaled by the relative
weight
\[
\frac{
	\fn[n]{ w }{ x }
}{
	\fn[p]{ \fn[n]{ \Fun{ M }{ f } }{ w } }{ \fn[n]{ f }{ x } }
} .
\]
The atom component is then transported through \( \Nat[nn]{\alpha}{x} \),
while the log-weight component is copied to the corresponding
\( \vect{\varepsilon}_{x} \)-summand.

Under the base projection, the monoid of weights, the barycentre semimodule,
and the monad \(M^{\cat[]{F}}\) recover the corresponding structures in
\( \Set \). Thus \(M^{\cat[]{F}}\) is a lift of the FAD monad on \( \Set \)
along the base projection, with the additional fibre information required for
reverse-mode differentiation. Over a distribution \(w\in\Fun[nn]{M}{X}\), this
fibre consists of one cotangent component for each atom \(x\in\supp w\), and
one real cotangent component for the log-weight \(\log\fn[n]{w}{x}\).

\subsubsection{Expected-value gradients and score contributions}

We now spell out how the lifted FAD monad separates the usual contributions to
the derivative of an expected value. The point is not to introduce an external
analytic argument, but to show that the decomposition is already encoded by
the fibre maps in \( \Fam{\Vect[o]} \). In particular, for a function
\( \fn[n]{f}{\theta,z} \), with \(z\) sampled from a finite distribution
\( \fn[n]{q}{\theta} \) depending on the parameter \( \theta \), the lifted
semantics separates three sources of variation: the explicit dependence of
\(f\) on \( \theta \), the dependence of the atoms \(z\) on \( \theta \), and
the dependence of the log-weights on \( \theta \).

To make this comparison readable, we use differential notation for objects and
morphisms in \( \Fam{\Vect[o]} \). Thus we write
\[
T^{\ast}U
=
\tuple{
	U,
	\left\{T^{\ast}_{u}U\right\}_{u\in U}
}
\]
and write morphisms as
\begin{equation*}
	\hat{h} = \tuple{ h, \grad h }: T^{\ast} X \to T^{\ast} Y,
	\qquad
	\left( \grad h \right)_{x} :
	T^{\ast}_{\fn[n]{h}{x}}Y \to T^{\ast}_{x}X,
\end{equation*}
where \( \Nat[pn]{ \grad h }{x} \) denotes the \(x\)-fibre of the fibre part.

Consider the lifted real-valued expectation algebra
\[
\hat{\opn{E}}_{\R}^{\cat[]{F}}:
\Fun[np]{M^{\cat[]{F}}}{\TR}
\to
\TR .
\]
For \(u\in\Fun[nn]{M}{\R}\), its base and fibre parts are
\begin{equation}
	\label{eq:lifted-real-expectation-fibre}
	\fn[n]{ \opn{E}_{\R}^{\cat[]{F}} }{u}
	=
	\sum_{r\in\supp u} \fn[n]{ u }{ r } \cdot r,
	\qquad
	\fn[n]{ \Nat[pn]{ \grad \opn{E}_{\R}^{\cat[]{F}} }{ u } }{ t }
	=
	\sum_{r \in \supp u}
	\left(
	\fn[n]{ u }{ r } \cdot t \cdot\vect{e}_{r}
	+
	\fn[n]{ u }{ r } \cdot r \cdot t \cdot \vect{\varepsilon}_{r}
	\right).
\end{equation}
The two summands in the fibre component have different origins. The
\( \vect{e}_{r} \)-term is the cotangent contribution of the atom value \(r\);
the \( \vect{\varepsilon}_{r} \)-term is the cotangent contribution of the
log-weight \( \log \fn[n]{u}{r} \).

Let
\[
\hat{g}=\tuple{g,\grad g}:T^{\ast}X\to\TR,
\qquad
w\in\Fun[nn]{M}{X},
\qquad
u:=\fn[p]{\Fun[nn]{M}{g}}{w}.
\]
The Frobenius identity gives
\begin{equation*}
	\sum_{r\in\supp u} \fn[n]{ u }{ r } \cdot r
	=
	\sum_{x\in\supp w} \fn[n]{ w }{ x } \cdot \fn[n]{ g }{ x }.
\end{equation*}
On fibres, composing the map induced by
\( \Fun[nn]{M^{\cat[]{F}}}{\hat{g}} \) with
Equation~\eqref{eq:lifted-real-expectation-fibre} gives
\begin{equation}
	\label{eq:frobenius-cotangent-rule}
	\fn[n]{
		\Nat[pn]{
			\grad \! \left(
			\hat{ \opn{E} }_{\R}^{\cat[]{F}}
			\circ
			\Fun[nn]{ M^{\cat[]{F}} }{ \hat{g} }
			\right)
		}{ w }
	}{t}
	=
	\sum_{x \in \supp w}
	\left(
	\fn[n]{ w }{ x } \cdot \fn[n]{ \Nat[pn]{ \grad g }{ x } }{ t } \cdot \vect{e}_x
	+
	\fn[n]{ w }{ x } \cdot \fn[n]{ g }{ x } \cdot t \cdot \vect{ \varepsilon }_x
	\right).
\end{equation}
Indeed, the fibre map of \( \Fun[nn]{M^{\cat[]{F}}}{\hat{g}} \) contributes,
in the \(x\)-summand, the relative factor
\[
\frac{\fn[n]{w}{x}}{\fn[p]{\fn[n]{\Fun{M}{g}}{w}}{\fn[n]{g}{x}}},
\]
while Equation~\eqref{eq:lifted-real-expectation-fibre} contributes the
corresponding mass
\[
\fn[p]{\fn[n]{\Fun{M}{g}}{w}}{\fn[n]{g}{x}}.
\]
These factors cancel, leaving precisely the expression in
Equation~\eqref{eq:frobenius-cotangent-rule}.

Now let
\[
\hat{q}
=
\tuple{ q, \grad q }:
T^{\ast} \Theta
\to
\Fun[np]{ M^{\cat[]{F}} }{ T^{\ast} Z }
\]
be a parameterised FAD. For each \( \theta\in\Theta \), choose a local
collision-free presentation
\begin{equation*}
	\fn[n]{ q }{ \theta } =
	\fn[n]{ \categ }{
		\tuple{
			\tuple{
				\fn[n]{ z_{i} }{ \theta },
				\fn[n]{ \ell_{i} }{ \theta }
			}
		}_{i \in I}
	}
	\qq{and}
	\fn[n]{ w_{i} }{ \theta } := \exp \fn[n]{ \ell_{i} }{ \theta } .
\end{equation*}
Let
\[
\hat{\underline q}
:=
\Nat[nn]{ \str^{\cat[]{F}} }{ T^{\ast} \Theta, T^{\ast} Z }
\circ
\tupling{ \id_{T^{\ast} \Theta}, \hat{q} } .
\]
Then
\[
\fn[n]{ \underline{q} }{ \theta }
=
\fn[n]{ \categ }{
	\tuple{
		\tuple{ \theta, \fn[n]{ z_{i} }{ \theta } },
		\fn[n]{ \ell_{i} }{ \theta }
	}_{i \in I}
}.
\]
With respect to this local labelling, the fibre part of
\(\hat{\underline{q}}\) is
\begin{equation}
	\label{eq:suggestive-fibre-parametrised-fad}
	\fn[n]{
		\Nat[pn]{ \grad \underline{q} }{ \theta }
	}{
		\sum_{i \in I}
		\left(
		\tuple{d_i,a_i} \cdot \vect{e}_{\tuple{\theta,z_i(\theta)}}
		+
		s_i\cdot\vect{ \varepsilon }_{\tuple{\theta,z_i(\theta)}}
		\right)
	}
	=
	\sum_{i\in I} d_i
	+
	\sum_{i \in I} \fn[n]{ \Nat[pn]{ \grad z_{i} }{ \theta } }{ a_{i} }
	+
	\sum_{i \in I} \fn[n]{ \Nat[pn]{ \grad \ell_{i} }{ \theta } }{ s_{i} },
\end{equation}
where
\[
d_i\in T^{\ast}_{\theta}\Theta,
\qquad
a_i\in T^{\ast}_{\fn[n]{z_i}{\theta}}Z,
\qquad
s_i\in\R .
\]

Finally, let
\[
\widehat{f}
=
\tuple{f,\grad f}:
T^{\ast}\Theta\times T^{\ast}Z
\to
\TR .
\]
We write the two components of its fibre part as
\begin{equation*}
	\fn[n]{ \Nat[pn]{ \grad f }{\tuple{ \theta, z }} }{ t }
	=
	\tuple{
		\fn[n]{ \Nat[nn]{ \left( \grad f \right)^{\Theta} }{ \tuple{ \theta, z } } }{ t },
		\fn[n]{ \Nat[nn]{ \left( \grad f \right)^{Z} }{ \tuple{ \theta, z } } }{ t }
	}.
\end{equation*}
Apply Equation~\eqref{eq:frobenius-cotangent-rule} to the composite
\[
T^{\ast} \Theta
\xrightarrow{\hat{\underline q}}
\Fun[np]{ M^{\cat[]{F}} }{ T^{\ast} \Theta \times T^{\ast} Z }
\xrightarrow{\Fun[nn]{ M^{\cat[]{F}} }{ \widehat{f} }}
\Fun[nn]{ M^{\cat[]{F}} }{ \TR }
\xrightarrow{\opn{E}^{\cat[]{F}}_{\R}}
\TR
\]
and then use Equation~\eqref{eq:suggestive-fibre-parametrised-fad}. This gives
\begin{equation}
	\label{eq:three-term-weighted-particle-gradient}
	\begin{aligned}
		\fn[n]{
			\Nat[pn]{
				\grad\!
				\left(
				\hat{ \opn{E} }_{\R}^{\cat[]{F}}
				\circ
				\Fun[nn]{ M^{\cat[]{F}} }{ \widehat{f} }
				\circ
				\hat{\underline{q}}
				\right)
			}{ \theta }
		}{ t }
		&=
		\sum_{i \in I}
		\fn[n]{ w_{i} }{ \theta }
		\cdot
		\left[
		\fn[n]{
			\Nat[nn]{ \left( \grad f \right)^{\Theta} }{
				\tuple{ \theta, \fn[n]{ z_{i} }{ \theta } }
			}
		}{ t }
		+
		\fn[n]{
			\Nat[pn]{ \grad z_{i} }{ \theta }
		}{
			\fn[n]{
				\Nat[nn]{ \left( \grad f \right)^{Z} }{
					\tuple{ \theta, \fn[n]{ z_{i} }{ \theta } }
				}
			}{ t }
		}
		\right.
		\\
		&\qquad \qquad \qquad
		+
		\left.
		\fn[n]{ f }{ \tuple{ \theta, \fn[n]{ z_{i} }{ \theta } } }
		\cdot
		\fn[n]{ \Nat[pn]{ \grad \ell_{i} }{ \theta } }{ t }
		\right].
	\end{aligned}
\end{equation}
The three terms inside the brackets have distinct origins. The first is the
explicit parameter contribution, coming from the dependence of
\(\fn[n]{f}{\theta,z}\) on \( \theta \). The second is the moving-atom
contribution, obtained by transporting cotangents through the dependence of
\(z_i(\theta)\) on \( \theta \). The third is the score contribution, coming
from the dependence of the log-weight
\(\ell_i(\theta)=\log \fn[n]{w_i}{\theta}\) on \( \theta \).

This calculation explains the operational role of the lifted FAD monad in
\( \Fam{\Vect[o]} \). It separates the reverse-mode contributions carried by a
finite probabilistic computation into atom-location terms and log-weight
terms, and it does so internally, through the fibre maps of the associated FAD
monad. The next subsection lifts the construction once more, from
\( \Fam{\Vect[o]} \) to the logical-relations category used to prove that the
syntactic CHAD transformation computes these cotangent contributions
correctly.

\subsection{FAD Monad on the Logical-Relations Category}
\label{sec:fad-monad-on-the-logical-relations-category}

The final instance of the FAD factorisation construction is the
logical-relations category used in the correctness proof. Whereas
\( \Fam{\Vect[o]} \) records the differentiated semantics of programs,
\( \cat[]{LR} \) relates this differentiated semantics to the ordinary
semantics in \( \Set \). The purpose of this subsection is to equip
\( \cat[]{LR} \) with a FAD factorisation structure, and hence with an
associated FAD monad \(M^{\cat[]{LR}}\). This is the semantic ingredient needed
later to interpret the source language in the logical-relations model and to
prove correctness of the reverse-mode CHAD transformation.

We use the comma-category presentation of logical relations from
\cite[Section~11]{lucatellinunesCHADExpressiveTotal2023}. Let
\( G:\cat[]{C}\to\cat[]{D} \) be a functor. We write
\[
\cat[]{LR}:=\commacat{\cat[]{D}}{G}.
\]
An object of \( \cat[]{LR} \) is a triple \( \tuple{D,C,f} \), where
\(D\) is an object of \( \cat[]{D} \), \(C\) is an object of \( \cat[]{C} \),
and \(f:D\to\Fun[nn]{G}{C}\) is a morphism in \( \cat[]{D} \). A morphism
\[
\tuple{u,v}:\tuple{D,C,f}\to\tuple{D',C',f'}
\]
consists of morphisms \(u:D\to D'\) in \( \cat[]{D} \) and
\(v:C\to C'\) in \( \cat[]{C} \) such that
\(
\Fun[nn]{G}{v}\circ f = f'\circ u.
\)
There is therefore a forgetful functor
\[
U:\cat[]{LR}\to\cat[]{D}\times\cat[]{C},
\qquad
\textnormal{defined by}
\qquad
\Fun[nn]{U}{\tuple{D,C,f}}:=\tuple{D,C},
\qquad
\Fun[nn]{U}{\tuple{u,v}}:=\tuple{u,v}.
\]

For the correctness argument in this paper, we take
\[
\cat[]{D}:=\Set,
\qquad
\cat[]{C}:=\Set\times\Fam{\Vect[o]},
\]
and define \(G\) to be the hom-functor
\begin{equation}
	\label{eq:definition-g-functor-logical-relations}
	G := \fn[n]{ \cat[]{C} }{ \tuple{ \E, \TE }, \blankdash },
\end{equation}
where \( \E \) is a fixed finite-dimensional real vector space, regarded as a
set, and
\(
\TE:=\tuple{\E,\const\E^{\ast}}
\)
is its cotangent bundle.

The space \( \E \) is used as a probe object. For example, when
\( \E=\R \), morphisms from \( \tuple{\E,\TE} \) into an object of
\( \cat[]{C} \) can be read as parametrised curves equipped with cotangent
data. In this way, differentiability of a map can be tested by its action on
such probes: postcomposition should send admissible probes in the domain to
admissible probes in the codomain.

Thus, in the present instance, an object of \( \cat[]{LR} \) is a triple
\(
\tuple{D,C,p}
\)
where
\(
p:D\to\fn[n]{\cat[]{C}}{\tuple{\E,\TE},C}
\)
is a function. When \(p\) is injective, \(D\) may be regarded as a specified
class of probes \( \tuple{\E,\TE}\to C \), namely those satisfying the
predicate represented by \(p\). A morphism
\(
\tuple{u,v}:\tuple{D,C,p}\to\tuple{D',C',p'}
\)
asserts that postcomposition with \(v:C\to C'\) sends probes satisfying
\(p\) to probes satisfying \(p'\).

We write
\[
\dompt{\tuple{D,C,p}}:=D,
\qquad
\codpt{\tuple{D,C,p}}:=C,
\qquad
\dompt{\tuple{u,v}}:=u,
\qquad
\codpt{\tuple{u,v}}:=v
\]
for the domain and codomain parts of objects and morphisms of \( \cat[]{LR} \).

In the correctness proof, the predicates represented in \( \cat[]{LR} \) will
express differentiability and correctness of the syntactic reverse derivative.
Accordingly, the source language must admit an interpretation in
\( \cat[]{LR} \). This requires, in particular, that \( \cat[]{LR} \) carry
the categorical structure needed for the language: finite products and
coproducts, the FAD monad, and the distinguished objects and operations used
to interpret base types and primitives.\footnote{For the description of the
	source language and its interpretation requirements, see
	Section~\ref{sec:source-language}.} The remainder of this subsection verifies
the FAD part of this structure by applying the FAD factorisation construction
of Section~\ref{sec:generalising-the-set-procedure}.

The main structural input is the monadic-comonadic analysis of these comma
categories from the CHAD logical-relations construction.

\begin{lemma}[{\cite[Corollary~99]{lucatellinunesCHADExpressiveTotal2023}}]
	\label{lem:monadicity-and-comonadicity-of-u}
	Assume that \( \cat[]{C} \) has binary coproducts and \( \cat[]{D} \) has
	binary products. Then
	\( U:\commacat{\cat[]{D}}{G}\to\cat[]{D}\times\cat[]{C} \) is comonadic and
	monadic, provided that \(G\) has a left adjoint.
\end{lemma}

In the present case, \( \cat[]{C}=\Set\times\Fam{\Vect[o]} \) is complete and
cocomplete; hence so is \( \Set\times\cat[]{C} \). Since \( \cat[]{C} \) has
copowers, the hom-functor \(G\) has a left adjoint, and
Lemma~\ref{lem:monadicity-and-comonadicity-of-u} applies. Thus
\[
U:\cat[]{LR}\to\Set\times\cat[]{C}
\]
is both monadic and comonadic. The monadic part implies that \(U\) creates
limits, while the comonadic part implies that \(U\) creates colimits. Hence
\( \cat[]{LR} \) is complete and cocomplete, and its limits and colimits are
computed by taking the corresponding limits and colimits in the underlying
\( \Set \)- and \( \cat[]{C} \)-components, with the unique induced
logical-relation structure.

The same componentwise description gives the factorisation system required by
Section~\ref{sec:generalising-the-set-procedure}. A morphism
\( \tuple{u,v} \) lies in the left class precisely when \(u\) and \(v\) lie in
the corresponding left classes in \( \Set \) and
\( \Set\times\Fam{\Vect[o]} \); similarly for the right class. Writing
\( v=\tuple{v_1,v_2} \), the left-class condition says that \(u\) and \(v_1\)
are surjective and that \(v_2\) lies in the left class described in
Section~\ref{sec:fad-monad-on-fam-vectop}. The right-class condition says that
\(u\) and \(v_1\) are injective and that \(v_2\) lies in the corresponding
right class. This is the usual image factorisation, equivalently the
(epi, regular mono) factorisation system in \( \cat[]{LR} \). Since the left
classes in \( \Set \) and \( \Fam{\Vect[o]} \) are closed under binary
products, and \(U\) creates products, the left class in \( \cat[]{LR} \) is
closed under binary products as well.

It remains to specify the monoid of weights and the barycentre semimodule in
\( \cat[]{LR} \). Both are obtained by imposing the differentiability
predicate on the corresponding structures in
\( \cat[]{C}=\Set\times\Fam{\Vect[o]} \). The monoid of weights is
$\tuple{
	\opn{Diff} \left(\E,\R\right),
	\tuple{\R,\TR},
	\operatorname{dv}_{\R}
}$
where \( \opn{Diff} \left(\E,\R\right) \) denotes the set of
differentiable functions \( \E\to\R \), and \( \operatorname{dv}_{\R} \) is
the usual differentiability predicate. More generally, for each \(k\in\N\),
\begin{equation}
	\label{eq:definition-predicate-monoid-of-weights-in-lr}
	\operatorname{dv}_{\R[k]}:
	\opn{Diff} \left( \E, \R[k] \right)
	\to
	\fn[n]{ \Set }{ \E, \R[k] } \times \fn[n]{ \Fam{\Vect[o]} }{ \TE, \TR[k] }
	\qq{s.t.}
	f \mapsto \tuple{ f, \tuple{ f, \grad f } }.
\end{equation}
The barycentre semimodule is given by
$
\tuple{
	\opn{Diff} \left(\E,\R\right),
	\tuple{\R,\TR},
	\operatorname{dv}_{\R}
}$.\footnote{For the corresponding \(p\)-dimensional barycentre semimodule, one
	takes
	\( \tuple{ \opn{Diff} \left( \E,\R[p]\right),
		\tuple{\R[p],\TR[p]},\operatorname{dv}_{\R[p]} } \).}

The structure maps for the monoid and the semimodule are induced componentwise
from their counterparts in \( \Set \) and \( \Fam{\Vect[o]} \), with the
domain part given by postcomposition on differentiable functions.
For example, the operation \( \vplus \) is defined by
\begin{equation}
	\label{eq:definition-vplus-in-lr}
	\begin{array}{lllcl}
		\dompt{\vplus}
		&:=
		\vplus^{\Set} \circ \blankdash
		&:&
		\opn{Diff} \left( \E, \R[2] \right) &\to \opn{Diff} \left( \E, \R \right)
		\vspace{6pt}
		\\
		\codpt{\vplus}
		&:= \tuple{ \vplus^{\Set}, \vplus^{\cat[]{F}} }
		&:&
		\tuple{ \R[2], \TR[2] } &\to \tuple{ \R, \TR } .
	\end{array}
\end{equation}
Here we use the canonical identifications
\[
\opn{Diff} \left(\E,\R\right)\times \opn{Diff} \left(\E,\R\right)
\simeq
\opn{Diff} \left(\E,\R[2]\right)
\quad
\text{and}
\quad
\tuple{\R,\TR}\times\tuple{\R,\TR}
\simeq
\tuple{\R[2],\TR[2]}.
\]
The remaining structural morphisms are defined analogously. In each case, the
codomain part is the corresponding morphism in
\( \cat[]{C}=\Set\times\Fam{\Vect[o]} \), while the domain part restricts it
to the differentiable probes.

This componentwise construction is the sense in which the logical-relations
structure is fibred over \( \cat[]{C} \): objects, structural maps, limits,
colimits, and factorisations in \( \cat[]{LR} \) project to their corresponding
constructions in \( \cat[]{C} \).

We may therefore apply the FAD factorisation construction of
Section~\ref{sec:generalising-the-set-procedure} to \( \cat[]{LR} \). The
resulting monad \(M^{\cat[]{LR}}\) is fibred over
\( \cat[]{C}=\Set\times\Fam{\Vect[o]} \). Concretely, for an object
\( \tuple{D,C,p} \) and a morphism
\( \tuple{u,v}:\tuple{D,C,p}\to\tuple{D',C',p'} \) in \( \cat[]{LR} \),
\begin{equation}
	\label{eq:fad-monad-on-lr-fibred-over-c}
	\codpt{\Fun[nn]{ M^{\cat[]{LR}} }{ \tuple{ D, C, p } }}
	= \Fun[pn]{ M \times M^{\cat[]{F}} }{ C }
	\qq{and}
	\codpt{\Fun[nn]{ M^{\cat[]{LR}} }{ \tuple{ u, v } }}
	= \Fun[pn]{ M \times M^{\cat[]{F}} }{ v }.
\end{equation}

The domain part records the admissible differentiable probes into the FAD
object. Explicitly,
\begin{equation}
	\label{eq:fad-monad-on-lr-object-map-domain-part}
	\dompt{\Fun[nn]{ M^{\cat[]{LR}} }{ \tuple{ D, C, p } }} =
	\left\{
	\omega: D \to \opn{Diff} \left( \E, \R \right)
	\st
	\begin{aligned}
		&\abs{ \supp \omega } \lt \infty
		\text{ and }
		\\
		& \forall \left( x \in \supp \omega \right),
		\forall \left( v \in \E \right),
		\fn[n]{ \fn[n]{ \omega }{ x } }{ v } \gt 0
	\end{aligned}
	\right\}.
\end{equation}
Thus the domain part consists of finitely supported
\(\opn{Diff} \left(\E,\R\right)\)-valued weight distributions on \(D\), with
strictly positive non-zero weights at every probe point. On morphisms, it is
given by the usual pushforward of finite support:
\begin{equation}
	\label{eq:fad-monad-on-lr-morphism-map-domain-part}
	\dompt{\Fun[nn]{ M^{\cat[]{LR}} }{ \tuple{ u, v } }}:
	\dompt{\Fun[nn]{ M^{\cat[]{LR}} }{ \tuple{ D, C, p } }}
	\to
	\dompt{\Fun[nn]{ M^{\cat[]{LR}} }{ \tuple{ D', C', p' } }}
	\qq{s.t.}
	\omega
	\mapsto
	\left(
	y \mapsto
	\sum_{\substack{x \in \supp \omega \\ \fn[n]{ u }{ x } = y}} \fn[n]{ \omega }{ x }
	\right).
\end{equation}

Let \(X\) denote the \( \Set \)-component of \(C\). When \( \E=\R \), the
domain set
\(
\dompt{ \Fun[nn]{ M^{\cat[]{LR}} }{ \tuple{ D, C, p } } }
\)
can be read as the class of differentiable curves
\(h:\E\to\Fun[nn]{M}{X}\) selected by the logical relation. Locally, such a
curve is represented by a finite weighted presentation
\begin{equation*}
	\fn[n]{ h }{ v } =
	\fn[n]{ \categ }{
		\tuple{ \tuple{ \fn[n]{ f_{j} }{ v }, \fn[n]{ g_{j} }{ v } } }_{j = 1}^{n}
	},
\end{equation*}
where the curves \( f_j:\E\to X \) are represented by elements of \(D\) through
\(p\), and the functions \( g_j:\E\to\R \) are differentiable. The logical
relation simultaneously records the corresponding transposed derivatives.

We have therefore constructed the associated FAD monad
\(M^{\cat[]{LR}}\) on the logical-relations category. It is compatible with the
codomain projection to \( \Set\times\Fam{\Vect[o]} \), while its domain part
selects the differentiable probes into finite atomic distributions, including
differentiable atom data and differentiable log-weights. This completes the
semantic preparation for the syntactic categories: the source language will be
interpreted in \( \Set \), the differentiated target semantics in
\( \Fam{\Vect[o]} \), and the correctness argument in \( \cat[]{LR} \). The
next section introduces the syntactic category of the source language and the
structure that these semantic interpretations must preserve.

\section{Syntactic categories}

The CHAD perspective is to formulate automatic differentiation as a semantic
construction. Rather than defining a derivative transformation merely as an \textit{ad hoc } transformation, one first identifies the categorical structure
freely generated by the language and then defines a functor preserving that
structure. In this way, a code transformation is controlled by a categorical
transformation. The syntactic category is the object that makes this passage
precise.

Given a type theory, its syntactic category has as objects the types generated
by the theory\footnote{This assumes that the theory has product types. Without
	product types, one can instead work with the usual construction based on
	equivalence classes of contexts.} and as morphisms the
\( \alpha \)-equivalence classes of programs. Composition is given by
capture-avoiding substitution. The categorical structure of the syntactic
category records the corresponding type-forming operations: for instance, products,
coproducts, and exponentials interpret product, sum, and function types,
respectively. An interpretation of the theory in a category \( \cat[]{X} \) is
then a functor from the syntactic category to \( \cat[]{X} \) preserving the
chosen structure.

In this paper we use two type theories, the source language and the target
language. The source language is the language of programs to be differentiated.
The target language extends it with the additional structure needed to express
reverse-mode derivatives. The CHAD transformation sends source programs to
target programs by the corresponding structure-preserving functor between their
syntactic categories.
\subsection{Source language}
\label{sec:source-language}

In this section, we describe the \emph{source language}. Its syntactic category is the free distributive category generated by basic types \( \srcreal[n] \), \( n \in \N \), and basic built-in operations \( \srcop \), equipped with a chosen strong monad, a categorical operation, and a handler algebra.

Types in this language are given according to the following context-free grammar,
\begin{equation}
  \label{eq:source-language-type-construction-grammar}
  \tau, \sigma, \rho
  ::=
  0
  \mid
  1
  \mid
  \tau \times \sigma
  \mid
  \tau \sqcup \sigma
  \mid
  \srcreal[k] \left( k \in \N \right)
  \mid
  \opn[b]{M} \tau
\end{equation}
with \( \srcreal \) being used as abbreviation for \( \srcreal[1] \).

For terms, we denote variables by \( x, y, z, \dots \) and contexts as \( \Gamma = x_{1} : \tau_{1}, \cdots, x_{n} : \tau_{n} \). Term construction is defined by the following grammar
\begin{equation}
  \label{eq:source-language-term-construction-grammar}
  \begin{aligned}
    t, s, r ::=
    &
      \left.
      x
      \mid
      \tuple{\;}
      \mid
      \tuple{t, s}
      \mid
      \opn[b]{pr}_{1} t
      \mid
      \opn[b]{pr}_{2} t
      \mid
      \opn[b]{inl} t
      \mid
      \opn[b]{inr} t
      \right.
    \\
    &
      \left.
      \mid
      \casest{t}{x}{s}{y}{r}
      \right.
    \left.
    \mid
    \fn[n]{ \opn[b]{abort} }{ t }
    \right.
    \\
    &
      \left.
      \mid
      \letst{x}{t}{s}
      \right.
    \left.
    \mid
    \fn[n]{ \srcop }{ t_{1}, \dots, t_{k} }
    \right.
      \left.
      \mid
      \fn[n]{ \opn[b]{return} }{ t }
      \right.
    \\
    &
      \left.
      \mid
      \fn[n]{ \opn[b]{bind} }{ t; x . s }
      \right.
      \left.
      \mid
      \fn[n]{ \opn[b]{categorical} }{
      \tuple{
      \tuple{
      t_{i}, w_{i}
      }
      }_{i = 1}^{n}
      }
      \right.
    \left.
    \mid
    \fn[n]{ \opn[b]{E}_{\nu} }{ t }
    \right.
  \end{aligned}
\end{equation}
where \( \nu = \srcreal[p] \), for some \( p \in \N \), is a fixed type on which the handler algebra \( \opn[b]{E}_{\nu} \) is defined, and \( \srcop \) ranges over a family of chosen built-in operations, each with a signature \( \srcop: \srcreal[n_{1}] \times \cdots \times \srcreal[n_{k}] \to \srcreal[m] \), representing the basic operations available in the source language. This family is required to include all affine maps and pointwise linear algebra on products and sums, but extra operations may be added depending on the use case: a differential equation solver might require Bessel functions and spherical harmonics, while a signal analyser might require trigonometric functions, for example.

Typing judgements for the \( 0 \), \( 1 \), product, and coproduct types are standard, and they may be found in Appendix~\ref{sec:appendix-source-language-for-chad}. Any operation \( \srcop: \srcreal[n_{1}] \times \cdots \times \srcreal[n_{k}] \to \srcreal[m] \) defines a typing rule
\begin{equation*}
\prftree[r]{\textsc{Operations}}{
\Gamma \vdash t_{j}: \srcreal[n_{j}]
\quad
\left( j \in \left\{ 1, \dots, k \right\} \right)
}{
\Gamma \vdash
\fn[n]{ \srcop }{ t_{1}, \dots, t_{k} } :
\srcreal[m]
}
\end{equation*}

For types of the form \( \opn[b]{M} \tau \), we add the following extra rules
\begin{equation}
  \label{eq:source-language-monad-typing-rules}
  \begin{gathered}
    \prftree[r]{\( \opn[b]{M} \text{E} \)}{
      \Gamma \vdash t : \opn[b]{M} \nu
    }{
      \Gamma \vdash \fn[n]{ \opn[b]{E}_{\nu} }{ t } : \nu
    }
    \qquad
    \prftree[r]{\( \opn[b]{M} \text{I1} \)}{
      \Gamma \vdash t : \opn[b]{M} \sigma
    }{
      \Gamma, x : \sigma \vdash s : \opn[b]{M} \tau
    }{
      \Gamma \vdash
      \fn[n]{ \opn[b]{bind} }{ t; x . s } : \opn[b]{M} \tau
    }
    \\
    \prftree[r]{\( \opn[b]{M} \text{I2} \)}{
      \Gamma \vdash t : \tau
    }{
      \Gamma \vdash \fn[n]{ \opn[b]{return} }{ t }: \opn[b]{M} \tau
    }
    \qquad
    \prftree[r]{\( \opn[b]{M} \text{I3} \)}{
      \Gamma \vdash t_{i}: \tau
    }{
      \Gamma \vdash w_{i}: \srcreal
    }{
      \left( i \in \left\{ 1, \cdots, n \right\} \right)
    }
    {
      \Gamma \vdash
      \fn[n]{ \opn[b]{categorical} }{
        \tuple{
          \tuple{t_{i}, w_{i}}
        }_{i = 1}^{n}
      } :
      \opn[b]{M} \tau
    }
  \end{gathered}
\end{equation}

Finally, we add equations. In addition to the standard equational theory for the \( 0 \), \( 1 \), product, and coproduct types, we add equations for the bind, return, and algebra constructions. First, the Kleisli laws for the bind and return constructions are
\begin{equation}
  \label{eq:source-language-kleisli-laws}
  \begin{gathered}
    \fn[n]{ \opn[b]{bind} }{ \fn[n]{ \opn[b]{return} }{ t }; x . s }
    \equiv
    \casubst{s}{x}{t}
    \qquad
    \fn[n]{ \opn[b]{bind} }{ t; x . \fn[n]{ \opn[b]{return} }{ x } }
    \equiv t
    \\
    \fn[n]{ \opn[b]{bind} }{
      \fn[n]{ \opn[b]{bind} }{ t; x . s };
      y . r
    }
    \equiv
    \fn[n]{ \opn[b]{bind} }{
      t;
      x . \fn[n]{ \opn[b]{bind} }{ s ; y . r }
    }
  \end{gathered}
\end{equation}
where we use \( \casubst{s}{x}{t} \) for the capture-avoiding substitution of the term \( t \) for the variable \( x \) in the expression \( s \).

There are also equational rules for the handler algebra \( \opn[b]{E}_{\nu} \), which ensure that it behaves as an Eilenberg-Moore algebra for \( \opn[b]{M} \) over the type \( \nu = \srcreal[p] \),
\begin{equation}
  \label{eq:source-language-eilenberg-moore-algebra-laws}
  \fn[n]{ \opn[b]{E}_{\nu} }{ \fn[n]{ \opn[b]{return} }{ v } }
  \equiv v
  \qquad
  \fn[n]{ \opn[b]{E}_{\nu} }{ \fn[n]{ \opn[b]{bind} }{ t; x . u } }
  \equiv
  \fn[n]{ \opn[b]{E}_{\nu} }{
    \fn[n]{ \opn[b]{bind} }{
      t; x . \fn[n]{ \opn[b]{return} }{ \fn[n]{ \opn[b]{E}_{\nu} }{ u } }
    }
  }
\end{equation}

We define the category \( \Syn \) as the syntactic category of this source language. In categorical terms, \( \Syn \) can be seen as the freely generated distributive category on the objects \( \srcreal[n] \), \( n \in \N \), and operations \( \srcop \), equipped with a fixed strong monad, a chosen \( \opn[b]{categorical} \) operation, and a chosen Eilenberg--Moore algebra \( \opn[b]{E}_{\nu} \) over some object \( \nu = \srcreal[p] \) for some \( p \in \N \).

Given any other distributive category \( \cat[]{X} \), with a strong monad \( M^{\cat[]{X}}: \cat[]{X} \to \cat[]{X} \), chosen categorical constructors, and an Eilenberg--Moore algebra \( \opn{alg}: \Fun[nn]{ M^{\cat[]{X}} }{ V } \to V \) for some object \( V \), we may define an interpretation functor \( \genint{\cdot}: \Syn \to \cat[]{X} \) by choosing, for each \( n \in \N \), an object \( X_{n} \) in \( \cat[]{X} \) such that \( X_{p} = V \). The object map of this functor is given by the type interpretations, which are inductively defined on the structure of the types as
\begin{equation}
  \label{eq:source-language-interpretation-inductive-definition}
  \begin{aligned}
    \genint{0} &= 0_{\cat[]{X}}
    \\
    \genint{\tau \times \sigma} &= \genint{\tau} \times \genint{\sigma}
  \end{aligned}
  \quad
  \begin{aligned}
    \genint{1} &= 1_{\cat[]{X}}
    \\
    \genint{\tau \sqcup \sigma} &= \genint{\tau} \amalg \genint{\sigma}
  \end{aligned}
  \quad
  \begin{aligned}
    \genint{\srcreal[n]} &= X_{n}
    \\
    \genint{\opn[b]{M} \tau} &= \Fun[nn]{ M^{\cat[]{X}} }{ \genint{\tau} }
  \end{aligned}
\end{equation}
where \( 0_{\cat[]{X}} \) and \( 1_{\cat[]{X}} \) denote chosen initial and terminal objects in \( \cat[]{X} \), and \( \times \) and \( \amalg \) denote chosen product and coproduct structures, respectively.

Contexts are interpreted as products of the underlying types, that is, \( \genint{x_{1}: \tau_{1}, \cdots, x_{n} : \tau_{n}} = \prod_{i = 1}^{n} \genint{\tau_{i}} \), and a typing judgement \( \Gamma \vdash t : \tau \) is then interpreted as a morphism \( \genint{\Gamma \vdash t : \tau}: \genint{\Gamma} \to \genint{\tau} \).

Variables, constants, and constructors for the distributive structure are interpreted in the standard way (see Appendix~\ref{sec:appendix-source-language-for-chad}). A built-in operator \( \srcop: \srcreal[n_{1}] \times \cdots \times \srcreal[n_{k}] \to \srcreal[m] \) must be associated with a morphism \( \genint{\srcop}: \prod_{j = 1}^{k} X_{n_{j}} \to X_{m} \), so that
\begin{equation}
  \label{eq:source-language-built-in-operator-interpretation}
  \genint{\Gamma \vdash \fn[n]{ \srcop }{ t_{1}, \cdots, t_{k} } : \srcreal[m]}
  = \genint{\srcop}
  \circ
  \tupling{\genint{\Gamma \vdash t_{j} : \srcreal[n_{j}]}}_{j = 1}^{k}:
  \genint{\Gamma}
  \to
  X_{m}
\end{equation}
and the monadic operations are interpreted in terms of the unit \( \eta^{\cat[]{X}} \), multiplication \( \mu^{\cat[]{X}} \), and strength \( \str^{\cat[]{X}} \) of the monad \( M^{\cat[]{X}} \),
\begin{equation*}
  \begin{aligned}
    \genint{\Gamma \vdash \fn[n]{ \opn[b]{return} }{ t } : \opn[b]{M} \tau}
    &=
      \Nat[nn]{ \eta^{\cat[]{X}} }{ \genint{\tau} }
      \circ
      \genint{\Gamma \vdash t : \tau}
    \\
    \genint{\Gamma \vdash \fn[n]{ \opn[b]{bind} }{ t; x . s } : \opn[b]{M} \tau}
    &=
      \mu^{\cat[]{X}}_{\genint{\tau}}%
      \circ
      \Fun[nn]{ M^{\cat[]{X}} }{
      \genint{\Gamma, x : \sigma \vdash s : \opn[b]{M} \tau}
      }
      \circ
      \str^{\cat[]{X}}_{ \genint{\Gamma}, \genint{\sigma} }%
      \circ
      \tupling{
      \id_{\genint{\Gamma}},
      \genint{\Gamma \vdash t : \opn[b]{M} \sigma}
      }
  \end{aligned}
\end{equation*}

Terms constructed with \( \opn[b]{categorical} \) are interpreted using the chosen morphisms \( \categ_{\tau,n}^{\cat[]{X}} \).
\begin{equation*}
  \genint{
    \Gamma \vdash
    \fn[n]{ \opn[b]{categorical} }{
      \tuple{ \tuple{t_{i}, w_{i}} }_{i = 1}^{n}
    } :
    \opn[b]{M} \tau
  }
  =
  \categ_{\tau,n}^{\cat[]{X}}
  \circ
  \tupling{
    \tupling{
      \genint{\Gamma \vdash t_{i} : \tau},
      \genint{\Gamma \vdash w_{i} : \srcreal}
    }
  }_{i = 1}^{n}.
\end{equation*}
For the FAD language treated here, no additional equations for these constructors are imposed.

Finally, the handler algebra \( \opn[b]{E}_{\nu} \) is interpreted as the chosen algebra on \( X_{p} \),
\begin{equation}
  \label{eq:source-language-handler-algebra-interpretation}
  \genint{\Gamma \vdash \fn[n]{ \opn[b]{E}_{\nu} }{ t } : \nu}
  = \opn{alg} \circ \genint{\Gamma \vdash t : \opn[b]{M} \nu}
\end{equation}

Following this prescription, one can prove the following \emph{soundness theorem},
\begin{lemma}[Soundness of source language interpretations]
  \label{lem:soundness-of-source-language-interpretations}

  If \( \Gamma \vdash t \equiv s : \tau \) is derivable in the equational theory above, then \( \genint{\Gamma \vdash t : \tau} = \genint{\Gamma \vdash s : \tau} \) as morphisms \( \genint{\Gamma} \to \genint{\tau} \) in \( \cat[]{X} \).
\end{lemma}

We may now state the universal property of the syntactic category of the source language.

\begin{lemma}[Universal property of \( \Syn \)]
  \label{lem:universal-property-of-the-syntactic-category-of-the-source-language}

  Let \( \cat[]{X} \) be any distributive category equipped with
  \begin{itemize}
  \item A strong monad \( \tuple{ M^{\cat[]{X}}, \eta^{\cat[]{X}}, \mu^{\cat[]{X}}, \str^{\cat[]{X}} } \);
  \item For every \( n \in \N \), chosen objects \( X_{n} \);
  \item An Eilenberg-Moore algebra \( \opn{alg}: \Fun[nn]{ M^{\cat[]{X}} }{ X_{p} } \to X_{p} \) over the object \( X_{p} \);
  \item A suitably typed interpretation \( \genint{\srcop} \) of every built-in operation \( \srcop \);
  \item For each \( n \in \N \) and each object \( A \) interpreting a type \( \tau \), an arrow \( \categ^{\cat[]{X}}_{\tau, n}: \left( A \times X_{1} \right)^{n} \to \Fun[nn]{ M^{\cat[]{X}} }{ A } \), used to interpret \( \opn[b]{categorical} \) terms built from lists of length \( n \) of atoms of type \( \tau \) and weights in \( X_{1} \).
  \end{itemize}

  Then there exists a unique strict functor \( \genint{ \cdot }: \Syn \to \cat[]{X} \) preserving finite products and coproducts, the strong monad structure, the handler algebra, and \( \opn[b]{categorical} \) as above, and sending every equation of \( \Syn \) to an identity in \( \cat[]{X} \).
\end{lemma}

The proof of existence and uniqueness follows from the prescribed construction above.

\subsection{Interpretation of the source language in \texorpdfstring{\( \Set \)}{Set}}
\label{sec:interpretation-of-source-language-in-set}

Applying the universal property of the source language (Lemma~\ref{lem:universal-property-of-the-syntactic-category-of-the-source-language}) to the data below gives a unique interpretation functor \( \setint{\cdot}: \Syn \to \Set \) interpreting the source language in the category of sets. For every \( n \in \N \), this interpretation is characterised by
\begin{equation}
  \label{eq:interpretation-reals-in-set}
  \setint{\srcreal[n]} := \R[n]
\end{equation}
and, for any built-in operator \( \srcop: \srcreal[n_{1}] \times \cdots \times \srcreal[n_{k}] \to \srcreal[m] \) in the source language, we interpret \( \srcop \) as a differentiable function
\begin{equation}
  \label{eq:interpretation-operations-in-set}
  \setint{\srcop} : \R[n_{1}] \times \cdots \times \R[n_{k}] \to \R[m].
\end{equation}
The strong monad \( M^{\Set} \) is the FAD monad \( M \), with its unit, multiplication, and strength. The handler algebra for \( \nu = \srcreal[p] \) is interpreted as the expectation algebra
\begin{equation*}
  \opn{alg}^{\Set}_{\nu}: \Fun[nn]{ M }{ \R[p] } \to \R[p],
  \qquad
  w \mapsto \sum_{x \in \supp w} x \cdot \fn[n]{ w }{ x }.
\end{equation*}
Finally, \( \opn[b]{categorical} \) is interpreted using the construction from Equation~\eqref{eq:categorical-definition-in-set}: for each type \( \tau \), we use
\begin{equation*}
  \categ^{\Set}_{\tau}:\coprod_{n \in \N} \left( \setint{\tau} \times \R \right)^{n} \to \Fun[nn]{ M }{ \setint{\tau} }.\end{equation*}

\subsection{Target language}
\label{sec:target-language}

Code written in the source language (Section~\ref{sec:source-language}) can be automatically differentiated using the CHAD algorithm. The transformed code is written in a separate language, called the \emph{target language}.

The target language is dependently typed, split into a cartesian type theory and, for every context of cartesian types, a linear type theory. Semantically, the cartesian fragment of the target language allows one to reason about base spaces, while the linear type theories describe cotangents.

We denote by \( \CSyn \) the term model of the cartesian fragment of the target language. Crucially, \( \Syn \) can be seen as a full subcategory of \( \CSyn \). For each context \( \Gamma \) of cartesian types, we denote by \( \LSyn[\Gamma] \) the syntactic category of the associated linear type theory. Since any substitution \( \Gamma \to \Gamma' \) between cartesian contexts \( \Gamma \) induces a functor \( \LSyn[\Gamma'] \to \LSyn[\Gamma] \), we have a strict indexed category \( \LSyn: \CSyn[o] \to \Cat \). The syntactic category of the target language is then given by the Grothendieck construction \( \Tgt \) (see Section 2.3 of \cite{lucatellinunesCHADExpressiveTotal2023}). An explicit description of objects and morphisms in this category can be seen in Section~\ref{sec:interpretation-of-the-target-language-in-fam-vectop}.

As in the source language, we write \( \tau, \sigma, \rho \) for the \emph{cartesian types}, and reserve underlined Greek letters \( \lintype{\tau}, \lintype{\sigma}, \lintype{\rho} \) for linear types. We then have \emph{cartesian typing judgements} of the form \( \Gamma \vdash t : \tau \) and \emph{linear typing judgements} of the form \( \Gamma ; v : \lintype{\sigma} \vdash s : \lintype{\rho} \).

We can divide the description of the target language into three layers:
\begin{itemize}
\item \textbf{Non-monadic core}: deterministic operations. For brevity, we omit the description of this layer here (see Appendix~~\ref{sec:appendix-target-language-for-chad}).
\item \textbf{Oplax strong-monad core}: rules for a strong monad once a compatible cotangent lift has been supplied.
\item \textbf{Distribution and handler}: language features that are specific to the FAD monad.
\end{itemize}
In doing so, we isolate the data needed to apply CHAD to effects.

Cartesian types are described by the following grammar:
\begin{equation}
  \label{eq:target-language-cartesian-type-construction-grammar}
  \tau, \sigma, \rho
  ::=
  \left( \text{As in Equation~\eqref{eq:source-language-type-construction-grammar}} \right)
  \mid
  \lintype{\sigma} \linto \lintype{\rho}
  \mid
  \Sigma \left( x : \sigma \right) . \tau'
\end{equation}
meaning that, in the cartesian fragment, we also have function types between linear types and dependent sums. For linear types, we use the grammar
\begin{equation}
  \label{eq:target-language-linear-type-construction-grammar}
  \lintype{\tau}, \lintype{\sigma}, \lintype{\rho} ::=
  \lintype{\srcreal[k]}
  \; \left( k \in \N \right)
  \mid
  \lintype{1}
  \mid
  \lintype{\tau} \times \lintype{\sigma}
  \mid
  \opn[b]{case} \; t \; \opn[b]{of} \; \left\{ \opn[b]{in}_{i} x_{i} \mapsto \lintype{\sigma}_{i} \right\}_{i = 1}^{n}
  \mid
  \opn[b]{M}^{w}_{p : \opn[b]{M} \tau}
  \left[ x \mapsto \lintype{ \sigma } \right]
\end{equation}
meaning that the base Euclidean-space types have associated linear types (in the sense that the cotangent space of a Euclidean space at any point can be identified with the base space itself), as well as the fibre-wise zero object \( 1 \). Moreover, we allow fibre-wise binary biproducts \( \times \) and families of linear types indexed by finite disjunctions, and we introduce the cotangent linear type associated with monadic types.

We also distinguish \emph{cartesian} and \emph{linear terms}. The cartesian terms are described by the grammar
\begin{equation}
  \label{eq:target-language-cartesian-term-construction-grammar}
  t, s, r ::=
  \left( \text{As in Equation~\eqref{eq:source-language-term-construction-grammar}} \right)
  \mid
  \tuple{ a, t }
  \mid
  \opn[b]{case} \; p \; \opn[b]{of} \; \tuple{ x, y } \mapsto s
  \mid
  \lambda v \; . \; s
\end{equation}
where the constructor \( \tuple{ a, t } \) and the eliminator \( \opn[b]{case} \; p \; \opn[b]{of} \; \tuple{ x, y } \mapsto s \) are used in the typing rules for the cartesian \( \Sigma \)-type, and \( \lambda v \; . \; s \) is used for abstraction over the unique linear variable \( v \).

For linear terms, we use the following grammar:
\begin{equation}
  \label{eq:target-language-linear-term-construction-grammar}
  \begin{aligned}
    t, s, r
    ::=
    v
    &\mid
      \letst{v}{t}{s}
      \mid
      \tuple{ t, s }
      \mid
      \opn[b]{pr}_{i} t
      \mid
      \opn[b]{0}
      \mid
      t + s
      \mid
      \linapp{r}{t}
    \\
    &
      \mid
      \fn[n]{ \fn[n]{ \tgtop }{ t_{1}, \cdots, t_{k} } }{v}
      \mid
      \fn[n]{ \opn[b]{return}^{w}_{t} }{s}
      \mid
      \fn[n]{ \opn[b]{bind}^{w}_{s; x \; . \; t} }{r; u}
    \\
    &
      \mid
      \fn[n]{ \opn[b]{categorical}^{w}_{\tuple{ \tuple{ t_{i}, w_{i} } }_{i = 1}^{n}} }{
      \tuple{ \tuple{ t_{i}', w_{i}' } }_{i = 1}^{n}
      }
      \mid
      \fn[n]{ \opn[b]{E}^{w}_{\nu, t} }{ s }
  \end{aligned}
\end{equation}
The three lines correspond to the three layers: the first line consists of the non-monadic core operations (linear variables, let bindings, pairings and projections, the zero vector, sums, and \( \linapp{r}{t} \) for linear application with \( r \) a cartesian term of type \( \lintype{\sigma} \linto \lintype{\rho} \)). The second line has the derivatives of the built-in operations \( \srcop \) of the source language, as well as the linear fibre components of the monad operations \( \opn[b]{return} \) and \( \opn[b]{bind} \). Finally, the third line has the linear fibre components of \( \opn[b]{categorical} \) and of the handler algebra \( \opn[b]{E}_{\nu} \).

The typing rules for the non-monadic core are standard, so we omit them here (see Appendix~\ref{sec:appendix-target-language-for-chad}).

There are two typing rules for the monadic core, corresponding to the \( \opn[b]{return} \) and \( \opn[b]{bind} \) operations. For \( \opn[b]{return} \), the rule is
\begin{equation}
  \label{eq:target-language-return-typing-law}
  \prftree{
    \Gamma \vdash t : \tau
  }{
    \Gamma, x : \tau \vdash \lintype{\sigma} \text{ type}
  }{
    \Gamma; v : \casubst{\lintype{\sigma}}{x}{t} \vdash s : \lintype{\rho}
  }{
    \Gamma; v : \opn[b]{M}^{w}_{\fn[n]{ \opn[b]{return} }{ t } : \Fun[nn]{ \opn[b]{M} }{ \tau }} \left[ x \mapsto \lintype{\sigma} \right]
    \vdash
    \fn[n]{ \opn[b]{return}^{w}_{t} }{ s }:
    \lintype{\rho}
  }
\end{equation}
and for \( \opn[b]{bind} \), the rule is
\begin{equation}
  \label{eq:target-language-bind-typing-law}
  \begin{gathered}
    \Gamma, x : \sigma \vdash \lintype{\sigma} \text{ type}
    \qquad
    \Gamma, y : \tau \vdash \lintype{\tau} \text{ type}
    \qquad
    \Gamma \vdash \lintype{\rho} \text{ type}
    \qquad
    \Gamma \vdash s : \Fun[nn]{ \opn[b]{M} }{ \sigma }
    \\
    \prftree{
      \Gamma, x : \sigma \vdash t : \Fun[nn]{ \opn[b]{M} }{ \tau }
    }{
      \Gamma; v : \opn[b]{M}^{w}_{s : \Fun[nn]{ \opn[b]{M} }{ \sigma }}
      \left[ x \mapsto \lintype{\sigma} \right]
      \vdash
      r : \lintype{\rho}
    }{
      \Gamma, x : \sigma; v : \opn[b]{M}^{w}_{t : \Fun[nn]{ \opn[b]{M} }{ \tau }} \left[ y \mapsto \lintype{\tau} \right]
      \vdash u : \lintype{\rho} \times \lintype{\sigma}
    }{
      \Gamma; v : \opn[b]{M}^{w}_{\fn[n]{ \opn[b]{bind} }{ s; x \; . \; t } : \Fun[nn]{ \opn[b]{M} }{ \tau }}
      \left[ y \mapsto \lintype{\tau} \right]
      \vdash
      \fn[n]{ \opn[b]{bind}^{w}_{s; x \; . \; t} }{ r; u } : \lintype{\rho}
    }
  \end{gathered}
\end{equation}

These typing rules are the target-language interface for the chosen lifted strong monad. For the FAD monad, the lift \(M^{\cat[]{F}}\) also supplies transpose rules for \( \opn[b]{categorical} \) and for the handler algebra \( \opn[b]{E}_{\nu} \). For \( \opn[b]{categorical} \), we have, for all \( n \in \N \),
\begin{equation}
  \label{eq:target-language-typing-rule-categorical}
  \begin{gathered}
    \Gamma, x : \tau \vdash \lintype{\tau} \text{ type}
    \qquad
    \Gamma \vdash \lintype{\rho} \text{ type}
    \\
    \prftree{
      \Gamma \vdash w_{i} : \srcreal
    }{
      \Gamma \vdash t_{i}: \tau
    }{
      \Gamma; v : \lintype{\srcreal}
      \vdash
      w_{i}' : \lintype{\rho}
    }{
      \Gamma; v : \casubst{\lintype{\tau}}{x}{t_{i}}
      \vdash
      t_{i}' : \lintype{\rho}
    }{
      \left( 1 \le i \le n \right)
    }{
      \Gamma; v : \opn[b]{M}^{w}_{
        \fn[n]{ \opn[b]{categorical} }{
          \tuple{ \tuple{ t_{i}, w_{i} } }_{i = 1}^{n}
        } : \Fun[nn]{ \opn[b]{M} }{ \tau }
      } \left[ x \mapsto \lintype{\tau} \right]
      \vdash
      \fn[n]{
        \opn[b]{categorical}^{w}_{
          \tuple{ \tuple{ t_{i}, w_{i} } }_{i = 1}^{n}
        }
      }{
          \tuple{ \tuple{ t_{i}', w_{i}' } }_{i = 1}^{n}
      }
      : \lintype{\rho}
    }
  \end{gathered}
\end{equation}
and, for the handler algebra, recalling that \( \nu = \srcreal[p] \), we define \( \lintype{\nu} := \lintype{\srcreal[p]} \). Then,
\begin{equation}
  \label{eq:target-language-typing-rule-handler-algebra}
  \prftree{
    \Gamma \vdash t : \Fun[nn]{ \opn[b]{M} }{ \nu }
  }{
    \Gamma; v : \opn[b]{M}^{w}_{t: \Fun[nn]{ \opn[b]{M} }{ \nu }}
    \left[ \underline{\hphantom{x}} \mapsto \lintype{\nu} \right]
    \vdash s : \lintype{\rho}
  }{
    \Gamma; v : \lintype{\nu}
    \vdash
    \fn[n]{ \opn[b]{E}^{w}_{\nu, t} }{ s } : \lintype{\rho}
  }
\end{equation}

The equations for the cartesian fragment are standard, whereas those for the linear fragment can be seen as originating from a formal derivative procedure over the cartesian base counterparts (see Appendix~\ref{sec:appendix-target-language-for-chad}). For instance, the right unit law for \( \opn[b]{return} \) and  \( \opn[b]{bind} \) (see Equation~\eqref{eq:source-language-kleisli-laws})
\begin{equation*}
  \fn[n]{\opn[b]{bind}}{
    s;
    x.\fn[n]{\opn[b]{return}}{x}
  }
  \equiv
  s.
\end{equation*}
When imposed on the cartesian fragment, this law gives rise to the equality
\begin{equation}
  \label{eq:target-language-monad-w-right-unit}
  \fn[n]{
    \opn[b]{bind}^{w}_{
      s;
      x.\fn[n]{\opn[b]{return}}{x}
    }
  }{
    r;
    \tuple{
      \opn[b]{0},
      \fn[n]{\opn[b]{return}^{w}_{x}}{v}
    }
  }
  \equiv
  r .
\end{equation}

\subsection{Interpretation of the target language  in \texorpdfstring{\( \Fam{\Vect[o]} \)}{Fam(VectOp)}}
\label{sec:interpretation-of-the-target-language-in-fam-vectop}

In Section~\ref{sec:interpretation-of-source-language-in-set}, we provided set-based semantics for the source language via the interpretation functor \( \setint{\cdot} : \Syn \to \Set \). Here we provide concrete semantics for the target language syntactic category \( \Tgt \) by interpreting it in \( \Fam{\Vect[o]} \) (see Section~\ref{sec:fad-monad-on-fam-vectop}).

An object of \( \Tgt \) may be seen as a pair \( \tuple{ p : \tau, \lintype{\sigma} } \) (up to \( \alpha \)-equivalence) such that \( \lintype{\sigma} \) is a linear type in \( \LSyn[p : \tau] \). A morphism \( \tuple{ t, s }: \tuple{ p : \tau, \lintype{\rho} } \to \tuple{ q : \sigma, \lintype{\chi} } \) is a pair of a cartesian arrow \( t: p \to q \) in \( \CSyn \), and of a linear arrow \( s: \casubst{\lintype{\chi}}{q}{t} \to \lintype{\rho} \) in \( \LSyn[p : \tau] \).

We start by providing an interpretation of the cartesian fragment \( \CSyn \) of the target language. This fragment is interpreted in \( \Set \) by extending the interpretation of \( \Syn \), as described in Section~\ref{sec:interpretation-of-source-language-in-set}. For that reason, we keep the same notation, writing \( \setint{0} := \emptyset \), \( \setint{\srcreal[k]} := \R[k] \), \( \setint{\tau \times \sigma} := \setint{\tau} \times \setint{\sigma} \), and so on.
The linear interpretation of a linear type under a cartesian context \( p : \tau \) is also described inductively over the structure of the linear types, and is described in Appendix~\ref{sec:appendix-fam-vectop-denotational-semantics}. In particular, we give
\(
\fn[n]{
  \linint{
    \lintype{\srcreal[k]}
  }{p : \tau}
}{ \gamma } := \R[k]
\)
and, for any cartesian judgement
\(
p: \tau \vdash q:\Fun[nn]{\opn[b]{M}}{\rho}
\),
\(
\fn[n]{
  \linint{
    \opn[b]{M}^{w}_{q : \Fun[nn]{ \opn[b]{M} }{ \rho }}
    \left[ x \mapsto \lintype{\sigma} \right]
  }{p : \tau}
}{\gamma}
=
\bigoplus_{x \in \supp \fn[n]{ \setint{q} }{\gamma} }
\left(
  \fn[n]{ \linint{\lintype{\sigma}}{p : \tau, x : \rho} }{ \gamma, x }
  \oplus \R
\right)
\)
for any cartesian variable \( p : \tau \) and \( \gamma \in \setint{\tau} \).

So, for any object \( \tuple{ p : \tau, \lintype{\sigma} } \) in \( \Tgt \), we define its interpretation in \( \Fam{\Vect[o]} \) as
\begin{equation}
  \label{eq:target-language-fam-vectop-interpretation-objects}
  \fvoint{\tuple{ p : \tau, \lintype{\sigma} }} :=
  \tuple{
    \setint{\tau},
    \left(
    \gamma \mapsto
    \fn[n]{
    \linint{\lintype{\sigma}}{p : \tau}
    }{ \gamma }
    \right)
  }
\end{equation}
and morphisms are interpreted, displaying the corresponding ordinary \(\Vect\)-map in the fibre, as
\begin{equation}
  \label{eq:target-language-fam-vectop-interpretation-morphisms}
  \fvoint{\tuple{ t, s }} :=
  \tuple{
    \setint{t},
    \left(
      \gamma \mapsto
      \left(
        \opmor[p]{
          \fn[n]{ \linint{s}{p : \tau} }{ \gamma }
        }
        \in
        \fn[n]{ \Vect }{
          \fn[n]{ \linint{\lintype{\chi}}{q : \sigma} }{ \fn[n]{ \setint{t} }{ \gamma } },
          \fn[n]{ \linint{\lintype{\rho}}{p : \tau} }{ \gamma }
        }
      \right)
    \right)
  }
\end{equation}

\subsection{CHAD code transformation}

The core idea of CHAD is the translation of the automatic differentiation code transformation into a structure-preserving categorical transformation, i.e. a functor. This is a functor \( \chad: \Syn \to \Tgt \). For each term in context \( \Gamma \vdash t : \tau \) in the source language, we get a term under context
\begin{equation}
  \label{eq:chad-code-transformation-terms-under-context}
  \Fun[np]{ \chad }{ \Gamma }_{1}
  \vdash
  \Fun[np]{ \chad[\Gamma] }{ t } :
  \Sigma \left(
    p :
    \Fun[np]{ \chad }{ \tau }_{1}
  \right)
  \; . \;
  \left(
  \Fun[np]{ \chad }{ \tau }_{2}
  \linto
  \Fun[np]{ \chad }{ \Gamma }_{2}
  \right)
\end{equation}
In other words, the term under context \( \Gamma \vdash t : \tau \) is translated as a term of this dependent-sum type: its first projection has type \( \Fun[np]{ \chad }{ \tau }_{1} \), and the corresponding fibres are given by the function types \( \Fun[np]{ \chad }{ \tau }_{2} \linto \Fun[np]{ \chad }{ \Gamma }_{2} \), corresponding to transposed derivatives (backpropagators).

Similar to denotational semantics, type and term transformations are defined inductively over the structure of the types and terms. The definitions for base types, as well as for products and coproducts is straightforward, and can be seen in Appendix~\ref{sec:appendix-code-transformation-rules}. For the monadic types, we define
\begin{equation}
  \label{eq:chad-code-transformation-measure-type-derivative}
  \Fun[np]{ \chad }{ \Fun[nn]{ \opn[b]{M} }{ \tau } }_{1} :=
  \Fun[nn]{ \opn[b]{M} }{ \Fun[nn]{ \chad }{ \tau }_{1} }
  \qq{and}
  \Fun[np]{ \chad }{ \Fun[nn]{ \opn[b]{M} }{ \tau } }_{2} :=
  \opn[b]{M}^{w}_{p : \Fun[nn]{ \opn[b]{M} }{ \Fun[np]{ \chad }{ \tau }_{1} }}
  \left[
    x \mapsto \casubst{ \Fun[np]{ \chad }{ \tau }_{2} }{p}{x}
  \right]
\end{equation}
and the terms related to the monadic structure have their derivatives defined as
\begin{align}
  \Fun[np]{ \chad[\Gamma] }{ \fn[n]{ \opn[b]{return} }{ t } }
  &:=
  \opn[b]{case} \; \Fun[np]{ \chad[\Gamma] }{ t } \; \opn[b]{of} \;
  \tuple{ p, b } \to
  \tuple{
    \fn[n]{ \opn[b]{return} }{ p },
    \lambda v \; . \; \fn[n]{ \opn[b]{return}^{w}_{p} }{ \fn[n]{ b }{ v } }
  }
  \label{eq:chad-code-transformation-return}
  \\
  \Fun[np]{ \chad[\Gamma] }{ \fn[n]{ \opn[b]{bind} }{ t; x \; . \; s } }
  &:=
  \opn[b]{case} \; \Fun[nn]{ \chad[\Gamma] }{ t } \; \opn[b]{of}
  \tuple{ p, b } \to
  \opn[b]{let} \; r =
  \opn[b]{case} \; \Fun[nn]{ \chad[\Gamma, x : \sigma] }{ s } \; \opn[b]{of} \;
  \tuple{ q, c } \to q
  \; \opn[b]{in}
  \label{eq:chad-code-transformation-bind}
  \\
  &\qquad
  \tuple{
    \fn[n]{ \opn[b]{bind} }{ p; x \; . \; r },
    \lambda v \; . \;
    \fn[n]{ \opn[b]{bind}^{w}_{p; x \; . \; r} }{
      \fn[n]{ b }{ v };
      \fn[p]{
        \opn[b]{case} \; \Fun[nn]{ \chad[\Gamma, x : \sigma] }{ s } \; \opn[b]{of} \;
        \tuple{ q, c } \to c
      }{ v }
    }
  }
  \notag
\end{align}
while the term constructors specific to the FAD monad have their code transformations defined as
\begin{equation}
  \label{eq:chad-code-transformation-categorical}
  \begin{aligned}
  &\Fun[np]{ \chad[\Gamma] }{
    \fn[n]{ \opn[b]{categorical} }{ \tuple{ \tuple{ t_{i}, w_{i} } }_{i = 1}^{n} }
  } :=
    \\
  & \qquad
    \opn[b]{case} \; \Fun[nn]{ \chad[\Gamma] }{ t_{1} } \; \opn[b]{of}
    \tuple{ u_{1}, u_{1}' } \to
    \cdots
    \opn[b]{case} \; \Fun[nn]{ \chad[\Gamma] }{ t_{n} } \; \opn[b]{of}
    \tuple{ u_{n}, u_{n}' } \to
    \\
  & \qquad
    \opn[b]{case} \; \Fun[nn]{ \chad[\Gamma] }{ w_{1} } \; \opn[b]{of}
    \tuple{ a_{1}, a_{1}' } \to
    \cdots
    \opn[b]{case} \; \Fun[nn]{ \chad[\Gamma] }{ w_{n} } \; \opn[b]{of}
    \tuple{ a_{n}, a_{n}' } \to
    \\
  & \qquad \qquad
    \tuple{
    \fn[n]{ \opn[b]{categorical} }{ \tuple{ \tuple{ u_{i}, a_{i} } }_{i = 1}^{n} },
    \lambda v \; . \;
    \fn[n]{ \opn[b]{categorical}^{w}_{ \tuple{ \tuple{ u_{i}, a_{i} } }_{i = 1}^{n} } }{
    \tuple{ \tuple{ \fn[n]{ u_{i}' }{ v }, \fn[n]{ a_{i}' }{ v } } }_{i = 1}^{n}
    }
    }
  \end{aligned}
\end{equation}
\begin{equation}
  \label{eq:chad-code-transformation-handler-algebra}\qquad\qquad\qquad
  \Fun[np]{ \chad[\Gamma] }{ \fn[n]{ \opn[b]{E}_{\nu} }{ t } } :=
  \opn[b]{case} \; \Fun[np]{ \chad[\Gamma] }{ t } \; \opn[b]{of}
  \tuple{ p, b } \to
  \tuple{
    \fn[n]{ \opn[b]{E}_{\nu} }{ p },
    \lambda v \; . \; \fn[n]{ \opn[b]{E}^{w}_{\nu, p} }{ \fn[n]{ b }{ v } }
  }
\end{equation}

\section{Logical relations and correctness}

\subsection{Interpretation of source language in the category of logical relations}
\label{sec:interpretation-of-source-language-in-lr}

As in \( \Set \), we can use the universal property of the syntactic category of the source language to construct an interpretation of the source language in the category \( \cat[]{LR} \) of logical relations. It is enough to specify the interpretations of the elementary real types, the built-in operations, the monadic type constructor, \( \opn[b]{categorical} \), and the handler algebra.

For the elementary real number types, we define, for all \( n \in \N \),
\begin{equation}
  \label{eq:interpretation-reals-in-lr}
  \lrint{\srcreal[n]} :=
  \tuple{
    \opn{Diff} \left( \E, \R[n] \right),
    \tuple{ \R[n], \TR[n] },
    \operatorname{dv}_{\R[n]}
  }
\end{equation}
and, for any built-in operation \( \srcop: \srcreal[n_{1}] \times \cdots \times \srcreal[n_{k}] \to \srcreal[m] \) in the source language, its domain part is given by postcomposition with \( \setint{\srcop} \), while its codomain part is the base function together with its transposed derivative,
\begin{equation}
  \label{eq:interpretation-operations-in-lr}
  \begin{aligned}
    \dompt{\lrint{\srcop}}
    &:= \left( \setint{ \srcop } \circ \blankdash \right):
      \underbrace{
      \opn{Diff} \left( \E, \R[n_{1}] \right)
      \times
      \cdots
      \times
      \opn{Diff} \left( \E, \R[n_{k}] \right)
      }_{
      \text{identified with } \opn{Diff} \left( \E, \R[n_{1}] \times \cdots \times \R[n_{k}] \right)
      }
      \to
    \opn{Diff} \left( \E, \R[m] \right)
    \\
    \codpt{\lrint{\srcop}}
    &:=
      \tuple{ \setint{ \srcop }, \tuple{ \setint{ \srcop }, \grad \setint{ \srcop } } }:
      \tuple{
      \prod_{j = 1}^{k} \R[n_{j}],
      \prod_{j = 1}^{k} \TR[n_{j}]
      }
      \to
      \tuple{\R[m], \TR[m]}
  \end{aligned}
\end{equation}

The strong monad is taken to be the FAD monad \( M^{\cat[]{LR}} \) on \( \cat[]{LR} \) (see Section~\ref{sec:fad-monad-on-the-logical-relations-category}), with the corresponding interpretations of unit, multiplication, and strength. The interpretation of \( \opn[b]{categorical} \) is given by the factorisation procedure on \( \cat[]{LR} \) described in Section~\ref{sec:fad-monad-on-the-logical-relations-category}. For \( \nu = \srcreal[p] \), the handler algebra \( \opn[b]{E}_{\nu} \) is interpreted by the corresponding \( M^{\cat[]{LR}} \)-algebra
\(
  \opn{alg}^{\cat[]{LR}}_{\nu}:
  \Fun[nn]{ M^{\cat[]{LR}} }{ \lrint{\nu} }
  \to
  \lrint{\nu},
\)
whose codomain part is the pair consisting of the expectation algebra in \( \Set \) and the lifted expectation algebra in \( \Fam{\Vect[o]} \), and whose domain part sends a differentiable function of FADs to the differentiable expectation function obtained by applying that algebra pointwise.

It then follows from the universal property of the syntactic category of the source language (Lemma~\ref{lem:universal-property-of-the-syntactic-category-of-the-source-language}) that these data define a unique interpretation functor \( \lrint{\cdot}: \Syn \to \cat[]{LR} \).

Crucially, by composing this interpretation functor with \( U_{\Set} := \Pi_{1} \circ \Pi_{2} \circ U \), where \( U: \cat[]{LR} \to \Set \times \cat[]{C} \) is the comma-category forgetful functor, \( \Pi_{2}: \Set \times \cat[]{C} \to \cat[]{C} \) is the second projection into \( \cat[]{C} = \Set \times \Fam{\Vect[o]} \), and \( \Pi_{1}: \cat[]{C} \to \Set \) is the first projection, we get another distributive functor. Indeed, \( U \) creates limits and colimits, \( \Pi_{1} \) and \( \Pi_{2} \) preserve them, and the codomain part of \( M^{\cat[]{LR}} \) is fibred over the FAD monads in \( \Set \) and \( \Fam{\Vect[o]} \). This functor \( U_{\Set} \circ \lrint{\cdot} \) satisfies the same data as the \( \Set \)-interpretation, in particular Equations~\eqref{eq:interpretation-reals-in-set} and \eqref{eq:interpretation-operations-in-set}. Hence, by uniqueness,
\begin{equation}
  \label{eq:lr-interpretation-fibred-over-set}
  U_{\Set} \circ \lrint{\cdot} = \setint{\cdot}
\end{equation}

Analogously, composing \( \lrint{\cdot}: \Syn \to \cat[]{LR} \) with \( U_{\cat[]{F}} := \Pi_{2} \circ \Pi_{2} \circ U: \cat[]{LR} \to \Fam{\Vect[o]} \), we get a distributive functor \( \Syn \to \Fam{\Vect[o]} \) that interprets the monad type constructor (and its associated unit and multiplication) as the FAD monad in \( \Fam{\Vect[o]} \). It follows from the universal property of \( \Syn \) that
\begin{equation}
  \label{eq:lr-interpretation-fibred-over-fam-vectop}
  U_{\cat[]{F}} \circ \lrint{\cdot} = \fvoint{\cdot} \circ \chad
\end{equation}

\subsection{Correctness Theorem}

In this section, we use the previously defined machinery to prove the following correctness theorem.

\begin{theorem}[Correctness Theorem]
  \label{th:correctness-theorem}
  For any well-typed source term of the form
  \(
    x : \srcreal[n]
    \vdash
    t : \srcreal[m]
  \),
  where \( n, m \in \N \), we have that
  \begin{equation}
    \label{eq:correctness-equation}
    \setint{t}: \R[n] \to \R[m]\textnormal{ is differentiable and }
    \fvoint{\Fun[np]{ \chad }{ t }} =
    \tuple{
      \setint{t}, \grad{\setint{t}}
    }.
  \end{equation}
\end{theorem}

To prove this, fix a well-typed source term \( x: \srcreal[n] \vdash t : \srcreal[m] \). Interpreting this judgement in the logical-relations category via the functor \( \lrint{\cdot}: \Syn \to \cat[]{LR} \), we get a morphism
\begin{equation*}
  \lrint{t}:
  \tuple{
    \opn{Diff} \left( \E, \R[n] \right),
    \tuple{ \R[n], \TR[n] },
    \operatorname{dv}_{\R[n]}
  }
  \to
  \tuple{
    \opn{Diff} \left( \E, \R[m] \right),
    \tuple{ \R[m], \TR[m] },
    \operatorname{dv}_{\R[m]}
  }
\end{equation*}
in \( \cat[]{LR} \). By the defining commutative square for morphisms in a comma category,
\begin{equation*}
  \Fun[nn]{ G }{ \codpt{\lrint{t}} } \circ \operatorname{dv}_{\R[n]}
  = \operatorname{dv}_{\R[m]} \circ \dompt{\lrint{t}}:
  \opn{Diff} \left( \E, \R[n] \right)
  \to \fn[n]{ \cat[]{C} }{ \tuple{ \E, \TE }, \tuple{ \R[m], \TR[m] } }.
\end{equation*}

Since \( \E \) was an arbitrary finite-dimensional real vector space, we may take \( \E = \R[n] \). In this case, we apply both sides of the equality to \( \id_{\R[n]} \in \opn{Diff} \left( \R[n], \R[n] \right) \). Using \( \operatorname{dv}_{\R[n]} \left( \id_{\R[n]} \right) = \id_{\tuple{ \R[n], \TR[n] }} \), the left-hand side is
\begin{equation*}
  \fn[p]{ \Fun[nn]{ G }{ \codpt{\lrint{t}} } \circ \operatorname{dv}_{\R[n]} }{ \id_{\R[n]} }
    =
    \codpt{\lrint{t}}
    \circ
    \id_{\tuple{ \R[n], \TR[n] }}
    =
    \codpt{\lrint{t}}
    =
    \Fun[pp]{ \Pi_{2} \circ U }{ \lrint{t} },
\end{equation*}
and the right-hand side is
\begin{equation*}
  \fn[p]{ \operatorname{dv}_{\R[m]} \circ \dompt{\lrint{t}} }{ \id_{\R[n]} }
  =
  \tuple{
    \fn[n]{ \dompt{\lrint{t}} }{ \id_{\R[n]} },
    \tuple{
      \fn[n]{ \dompt{\lrint{t}} }{ \id_{\R[n]} },
      \grad \left(
        \fn[n]{ \dompt{\lrint{t}} }{ \id_{\R[n]} }
      \right)
    }
  }.
\end{equation*}
Applying the first projection \( \Pi_{1}: \cat[]{C} \to \Set \) to both expressions, it follows from Equation~\eqref{eq:lr-interpretation-fibred-over-set} that
\begin{equation*}
  \fn[n]{ \dompt{\lrint{t}} }{ \id_{\R[n]} }
  =
  \Fun[pp]{ \Pi_{1} \circ \Pi_{2} \circ U }{ \lrint{t} }
  =
  \Fun[nn]{ U_{\Set} }{ \lrint{t} }
  =
  \setint{t}
\end{equation*}
and thus, we can prove Theorem~\ref{th:correctness-theorem}, by using Equation~\eqref{eq:lr-interpretation-fibred-over-fam-vectop}:
\begin{equation*}
  \begin{aligned}
    \fvoint{\Fun[np]{ \chad }{ t }}\hspace{-2pt}
    &=
      \Fun[nn]{ U_{\cat[]{F}} }{ \lrint{t} }
    =
      \Fun[np]{ \Pi_{2} }{ \Fun[pp]{ \Pi_{2} \circ U }{ \lrint{t} } }
    =
      \tuple{
      { \dompt{\lrint{t}} }{\hspace{-13pt} (\id_{\R[n]}}),
      \grad \left(
       { \dompt{\lrint{t}} }{\hspace{-13pt} (\id_{\R[n]}) }
      \right)
      }=
      \tuple{ \setint{t}, \grad{\setint{t}} }.
  \end{aligned}
\end{equation*}

\section{Returning to the example}
\label{sec:returning-to-the-example}

In this section, we return to the stochastic optimisation problem discussed in Section~\ref{sec:motivating-example}. The goal is to encode it in the source language and then apply the code transformation. We then study how the transformed code relates to the mathematical description of the derivative semantics.

\subsection{Encoding the example in the source language}
\label{sec:encoding-the-example-in-the-source-language}

As a first step in encoding the example in the source language, we must set up the family of built-in operations. In addition to the basic operations for representing affine maps and linear algebra, we also add five nullary operations (i.e. constants) \( s_{m}, s_{p}, m_{0}, p_{0}, P : \srcreal \) and two unary operations \( \opn[b]{sig}, \opn[b]{lsig}: \srcreal \to \srcreal \), representing the sigmoid and the logarithm of the sigmoid function, respectively. Then, in the context \( \Gamma = p : \srcreal, m : \srcreal \), the program is given by the term
\begin{equation}
  \label{eq:example-source-code-encoding}
  \begin{aligned}
    &\opn[b]{let} \;
      s =
      \opn[b]{categorical}
      \left(
      \left \langle
      \tuple{ 0, \fn[n]{ \opn[b]{lsig} }{ -s_{m} \cdot \left( m - m_{0} \right) } },
      \right.
      \right.
    \\
    &\qquad \qquad \qquad \qquad \hphantom{a}
      \left.
      \left.
      \tuple{
      P \cdot \fn[n]{ \opn[b]{sig} }{ -s_{p} \cdot \left( p - p_{0} \right) },
      \fn[n]{ \opn[b]{lsig} }{ s_{m} \cdot \left( m - m_{0} \right) }
      }
      \right\rangle
      \right)
    \\
    &\opn[b]{in} \;
      \opn[b]{let} \;
      t = \fn[n]{ \opn[b]{bind} }{s; x . \fn[n]{ \opn[b]{return} }{ p \cdot x - m } }
      \; \opn[b]{in} \;
      \fn[n]{ \opn[b]{E}_{\srcreal} }{ t }
  \end{aligned}
\end{equation}
where we have chosen \( \nu = \srcreal \). First, we introduce a distribution \( s \) for the two possible total-sales scenarios. Then, we map the function \( x \mapsto p \cdot x - m \) over those scenarios using \( \opn[b]{bind} \) and \( \opn[b]{return} \), giving the profit distribution \( t \). Finally, we compute the expected total profit using the \( \opn[b]{E}_{\srcreal} \) handler.

In this simple example, the final result is an expected value. However, the source language allows different expected values to be mixed within the same expression. For instance, suppose that one wants to maximise a risk-adjusted profit given by the ratio of the square of the expected profit to the profit variance. The source-language version of this calculation is
\begin{equation*}
  \begin{aligned}
    &\opn[b]{let} \;
      s =
      \opn[b]{categorical}
      \left(
      \left \langle
      \tuple{ 0, \fn[n]{ \opn[b]{lsig} }{ -s_{m} \cdot \left( m - m_{0} \right) } },
      \right.
      \right.
    &&\text{\textcolor{lightgray}{// Sales distribution}}
    \\
    &\qquad \qquad \qquad \qquad \hphantom{a}
      \left.
      \left.
      \tuple{
      P \cdot \fn[n]{ \opn[b]{sig} }{ -s_{p} \cdot \left( p - p_{0} \right) },
      \fn[n]{ \opn[b]{lsig} }{ s_{m} \cdot \left( m - m_{0} \right) }
      }
      \right\rangle
      \right)
    \\
    & \opn[b]{in} \;
      \opn[b]{let} \;
      t = \fn[n]{ \opn[b]{bind} }{s; x . \fn[n]{ \opn[b]{return} }{ p \cdot x - m } }
      &&\text{\textcolor{lightgray}{// Profit distribution}}
    \\
    & \opn[b]{in} \;
      \opn[b]{let} \; a = \fn[n]{ \opn[b]{E}_{\srcreal} }{ t }
      &&\text{\textcolor{lightgray}{// Expected profit}}
    \\
    & \opn[b]{in} \;
      \opn[b]{let} \;
      b = \fn[n]{ \opn[b]{bind} }{ t; x . \fn[n]{ \opn[b]{return} }{ x^{2} } }
      &&\text{\textcolor{lightgray}{// Squared profit}}
    \\
    & \opn[b]{in} \;
      \opn[b]{let} \;
      v = \fn[n]{ \opn[b]{E}_{\srcreal} }{ b } - a^{2}
      &&\text{\textcolor{lightgray}{// Variance}}
    \\
    & \opn[b]{in} \; a^{2} / v
  \end{aligned}
\end{equation*}
The source language also supports more complex combinations of expected values, including expected values of expressions that themselves contain expected values.

\subsection{Code transformation}
\label{sec:code-transformation}

We now consider the differentiated code from the example in Equation~\eqref{eq:example-source-code-encoding}. The full code can be seen in Appendix~\ref{sec:appendix-example-of-differentiated-code}. Here we highlight some of the main features.

Recall that the differentiation is done with respect to the context \( \Gamma = p : \srcreal, m : \srcreal \), and we adopt the following abbreviations:
\begin{equation*}
  \begin{aligned}
    a_{0} &:= 0
    \\
    a_{1} &:= P \cdot \fn[n]{ \opn[b]{sig} }{ -s_{p} \cdot \left( p - p_{0} \right) }
  \end{aligned}
  \qquad
  \begin{aligned}
    \gamma_{0} &:= \fn[n]{ \opn[b]{lsig} }{ -s_{m} \cdot \left( m - m_{0} \right) }
    \\
    \gamma_{1} &:= \fn[n]{ \opn[b]{lsig} }{ s_{m} \cdot \left( m - m_{0} \right) }
  \end{aligned}
\end{equation*}

The first lines of the differentiated code are case statements that pattern match on
\begin{equation*}
  \begin{aligned}
    &\opn[b]{case} \; \tuple{ a_{0}, \grad a_{0} } \; \opn[b]{of}
      \tuple{ \theta_{0}, \theta_{0}' } \to
    \opn[b]{case} \; \tuple{ a_{1}, \grad a_{1} } \; \opn[b]{of}
      \tuple{ \theta_{1}, \theta_{1}' } \to
    \\
    &\opn[b]{case} \; \tuple{ \gamma_{0}, \grad \gamma_{0} } \; \opn[b]{of}
      \tuple{ \rho_{0}, \rho_{0}' } \to
    \opn[b]{case} \; \tuple{ \gamma_{1}, \grad \gamma_{1} } \; \opn[b]{of}
      \tuple{ \rho_{1}, \rho_{1}' } \to
    \\
    &\qquad
      \tuple{
      \fn[n]{ \opn[b]{categorical} }{
      \tuple{ \tuple{ \theta_{j}, \rho_{j} } }_{j = 0}^{1}
      },
      \lambda v \; . \;
      \fn[n]{
      \opn[b]{categorical}^{w}_{
      \tuple{ \tuple{ \theta_{j}, \rho_{j} } }_{j = 0}^{1}
      }
      }{
      \tuple{ \tuple{ \fn[n]{ \theta_{j}' }{ v }, \fn[n]{ \rho_{j}' }{ v } } }_{j = 0}^{1}
      }
      }
  \end{aligned}
\end{equation*}
These lines correspond to the derivative of \( \opn[b]{categorical} \). Since \( a_{i} = \theta_{i} \) and \( \gamma_{i} = \rho_{i} \), the first entry of the tuple is simply a recalculation of the same categorical distribution of sales. The second entry uses the \( \opn[b]{categorical}^{w} \) linear term construction to construct a cotangent value combining the transposed derivative of the atoms (\( \theta_{i}' \)) and of the log-weights (\( \rho_{i}' \)), as is the case with the denotational semantics in \( \Fam{\Vect[o]} \) (Section~\ref{sec:fad-monad-on-fam-vectop}).

Another point to highlight is the block concerning the transposed derivative of the handler:
\begin{equation*}
  \begin{array}{lll}
    &
      \opn[b]{case}
    &
      \opn[b]{case}
      \tuple{ t, \lambda v \; . \; \opn[b]{coproj}_{\fn[n]{ \opn{idx} }{ t; \Gamma, s, t }} v}
    \\
    &
    &
      \opn[b]{of} \tuple{ \zeta, \zeta' } \to
      \tuple{
      \fn[n]{ \opn[b]{E}_{\srcreal} }{ \zeta },
      \lambda v \; . \;
      \fn[n]{ \opn[b]{E}_{\srcreal, \zeta}^{w} }{ \fn[n]{ \zeta' }{ v } }
      }
    \\
    &
      \opn[b]{of}
    &\tuple{ \epsilon, \epsilon' }
      \to
      \tuple{
      \casubst{\epsilon}{t}{\delta},
      \lambda v \; . \;
      \letst{w}{
      \fn[n]{ \casubst{\epsilon'}{t}{\delta} }{ v }
      }{
      \opn[b]{pr}_{1} w + \fn[n]{ \delta' }{ \opn[b]{pr}_{2} w }
      }
      }
  \end{array}
\end{equation*}
Here, the transposed derivative appears as the transposed derivative of the variable \( t \) with respect to the extended context \( \Gamma, s : \Fun[nn]{ \opn[b]{M} }{ \srcreal }, t : \Fun[nn]{ \opn[b]{M} }{ \srcreal } \). We then compute its expected value, and use the \( \opn[b]{E}^{w}_{\srcreal} \) term constructor to compute the derivative of the handler algebra. In the last line, \( t \) is substituted by the \( \delta \) term, whose definition is given in a surrounding case construction.

\section{Beyond probability: discrete-output algebraic effects}
\label{sec:beyond-probability}

The FAD monad is the main case study, but our construction is not intrinsically probabilistic. The reusable data are a strong monad presented by discrete-output operations, a chosen Eilenberg--Moore handler algebra, and compatible lifts of the monad, the operations, and the handler to \( \Fam{\Vect[o]} \) and to the logical-relations category. With these data, the cotangent rules for algebraic operations, generic \( \opn[b]{return} \) and \( \opn[b]{bind} \), and handlers from Sections~\ref{sec:source-language}--\ref{sec:interpretation-of-the-target-language-in-fam-vectop} have the same form.

For a typed algebraic signature \(\Sigma\) with generic effects \( \omega:I_{\omega}\to T O_{\omega} \), the parameter types \(I_\omega\) may be differentiable: \(\fvoint{I_{\omega}}\) may have non-trivial cotangent fibres. The direct fibred-polynomial construction requires the operation output type to be finite and discrete \(O_\omega= 1\sqcup \cdots\sqcup 1\) (so \(\fvoint{O_{\omega}}\) has trivial fibre). Before quotienting by equations, the free AST monad has the form
\[
  \AST_{\Sigma}X
  \simeq
  \mu Y.\;
  X \sqcup
  \coprod_{\omega\in\Sigma}
  I_{\omega}\times
  \prod_{o\in O_{\omega}}Y:
\]
we take a finite product of continuation subtrees. The lift to \( \Fam{\Vect[o]} \) adds cotangents  for the parameter and branches, while the branch index \(o\) has zero fibre. (Outputs with non-trivial fibres can destroy fibredness.) Appendix~\ref{sec:appendix-additional-effects} gives the corresponding generic CHAD theorem.

\paragraph{Finite non-determinism}
For finite non-determinism, we replace the FAD monad with the finite multiset monad
\(
  \MFin X = \coprod_{k \in \N} \left(X^{k}/\mathfrak{S}_{k}\right),
\)
where multiplicity is retained but order is forgotten. From the signature with no equations, we can take the corresponding \( \opn{AST} \) monad as the list monad (where order is relevant). This is analogous to what happens to FADs: where FADs can be seen as lists of point/log-weight pairs modulo collisions and reordering, multisets can be seen as lists modulo reordering. These additional equations are then imposed by an \( \opn{AST} \)-algebra.

For multisets, a \( \opn{AST} \)-algebra (i.e. a \( \opn{List} \)-algebra) is precisely a monoid structure on the underlying object. By making this monoid commutative, one then imposes order-irrelevance. Thus \( \tuple{\R[n],0,+} \), for example, induces a \( \opn{List} \)-algebra \( \opn{sum}: \Fun[np]{ \opn{List} }{ \R[n] } \to \R[n] \) for multiplicity-sensitive finite non-deterministic computations valued in \( \R[n] \). The ordinary lifted monad also quotients the fibres: over a multiset represented by \( \tuple{x_{1},\dots,x_{k}} \), the fibre is the stabiliser-invariant part of \( \bigoplus_i A_{x_i} \), equivalently one cotangent slot per support point after collisions. Hence repeated indistinguishable outcomes do not keep independent occurrence cotangents. The finite-choice reverse rule is the FAD categorical rule with the log-weight components removed, followed by this quotienting of cotangents.

\paragraph{Exceptions}
For exceptions, \( \Fun[nn]{ M_{E} }{ X } = X \amalg E \) may be taken as its own \( \opn{AST} \). In general, a handler on a carrier \(A\) is determined by a recovery map \(h:E\to A\), giving the algebra \( [\id_A,h]:A\amalg E\to A\).
To recover the monad from the CPS factorisation, it is enough to use the free algebra on carrier \(A=\Fun[nn]{ M_E }{ 2 }=2\amalg E\) as a canonical algebra: monad multiplication \( \Nat[nn]{ \mu }{ 2 }: \left( 2 \amalg E \right) \amalg E \to 2 \amalg E \).
If \(E\) has continuous structure, a raised exception carries an exception-label cotangent until it is handled, and the handler uses the transposed derivative of the recovery map.
Since here the monad and the \( \opn{AST} \) coincide, exceptions fit the same pattern without requiring a quotient construction.

\paragraph{Accumulation}
For writer accumulation, let \( \Fun[nn]{ W_{M} }{ X } = X \times M\) be the writer monad associated to a monoid \(M\). For the basic case \( M = \R[n] \) we have a cancellable monoid, and again it may be taken as its own \( \opn{AST} \). Its algebras are right \(M\)-actions \(A\times M\to A\), and so for the \( \opn{AST} \)-algebra we use the carrier \(\R[n]\) with action \( \R[n] \times \R[n] \to \R[n] \), \( \tuple{ u, v } \mapsto u + v \).
The lifted writer monad attaches to \(\tuple{x,r}\) the fibre \(A_x\oplus\R[n]\), and the handler rule is the transposed derivative of addition.

\paragraph{Remark (source syntax for the additional first-order effects).}
For finite multiset non-determinism we add operations \(\opn[b]{choose}_k\) with permutation and commutative-monoid laws, with a fold handler for the chosen commutative monoid; operationally, a choice node can be implemented as a parallel fork and the handler as the corresponding join/reduction. For exceptions one may add \(\fn[n]{\opn[b]{raise}}{e}:\opn[b]{M}\tau\), with handlers determined by maps \(E\to\tau\). For the \(\R[n]\)-writer monad one may add \(\fn[n]{\opn[b]{tell}}{r}:\opn[b]{M}1\), together with the basic handler \(\opn[b]{M}\srcreal[n]\to\srcreal[n]\) induced by addition, which consumes the accumulated output. The target language then adds the corresponding linear terms for the transposed derivatives of operations and handlers; the generic \(\opn[b]{return}^{w}\) and \(\opn[b]{bind}^{w}\) rules are unchanged. These fork--join and accumulation patterns are common in high-performance code that we want to differentiate.

\section{Related work}
\label{sec:related-work}

This work connects probabilistic programming semantics, stochastic gradient
estimation, and semantic accounts of AD. The continuation view of distributions
has classical roots in Schwartz's theory of distributions~\cite{schwartzTheorieDistributions1950}.
In probabilistic programming, Vákár, Kammar, and Staton construct a
probabilistic powerdomain by factoring integration into a continuation
monad~\cite{vakarDomainTheoryStatistical2019}, while Kammar and McDermott
develop the factorisation-systems technology for monadic lifting and logical
relations~\cite{kammarFactorisationSystemsLogical2018}.

Stochastic gradient estimation is central in machine learning;
\cite{mohamedMonteCarloGradient2020} surveys score, pathwise, and
measure-valued estimators. Stochastic computation graphs showed how these
compose by AD-like backpropagation~\cite{schulmanGradientEstimationUsing2015},
and Storchastic allows different estimators at different stochastic
nodes~\cite{vankriekenStorchasticFrameworkGeneral2021}. ADEV is especially
close, extending forward-mode AD to expectations of probabilistic programs
with discrete and continuous distributions~\cite{lewADEVSoundAutomatic2023}.
We contribute the corresponding reverse-mode account: reverse~mode
exposes cotangent-flow and fibredness issues that do not arise
in forward mode. In particular, continuous probability is not merely the
discrete construction with sums replaced by integrals; the relevant
continuation/image construction need not be fibred over primal outcomes, so we
leave it to future work. Compared to ADEV, we also use image factorisation to
carve out distributions and their reverse-derivative structure from the
continuation monad, and formulate expectation as an algebraic-effect handler
\cite{plotkinAlgebraicOperationsGeneric2003,plotkinHandlingAlgebraicEffects2013}.

There is also growing work on deterministic reverse AD with algebraic effects
and handlers. De Vilhena and Pottier verify a define-by-run reverse-mode AD
library based on effect handlers~\cite{devilhenaVerifyingEffectHandlerBased2023},
while Sigal studies effect-handler implementations of AD in Frank, OCaml, and
related settings~\cite{sigalAutomaticDifferentiationEffects2021,sigalAutomaticDifferentiationEffects2024}.
Our focus is different: leaving implementation strategy open, with Efficient
CHAD as a natural starting point~\cite{smedingEfficientCHAD2024}, we give a
denotationally correct source transformation for languages whose source terms
themselves have algebraic effects. This builds on the CHAD paradigm for pure
source languages~\cite{vakarReverseADHigher2021,vakarCHADCombinatoryHomomorphic2022,lucatellinunesCHADExpressiveTotal2023}
and its extension to iteration and partiality~\cite{lucatellinunesUnravelingIterativeCHAD2025},
showing that the technique also applies to impure, discrete-output algebraic
effects and handlers.

\section*{Funding}

This research was supported by the ERC project FoRECAST. The second author was
also supported by the Centre for Mathematics of the University of Coimbra
(CMUC, \href{https://doi.org/10.54499/UID/00324/2025}{doi:10.54499/UID/00324/2025}),
funded by the Funda\c{c}\~ao para a Ci\^encia e a Tecnologia (FCT) through
grants UID/00324/2025 and UID/PRR/00324/2025.

\clearpage
\appendix
\section{Source language for CHAD}
\label{sec:appendix-source-language-for-chad}

In this section, we present the complete type theory of the source language for CHAD, whose abbreviated introduction can be found in Section~\ref{sec:source-language}.

This is a first-order language whose term model can be understood as a free distributive category generated by the basic types \( \srcreal[k] \), \( k \in \N \), and by a family of chosen operations \( \srcop \), and equipped with a strong monad, a categorical constructor, and a specific Eilenberg--Moore algebra.

\subsection{Type and term grammar}

The type grammar for the source language is the one displayed in Equation~\ref{eq:appendix-source-language-type-construction-grammar}
\begin{equation}
  \label{eq:appendix-source-language-type-construction-grammar}
  \tau, \sigma, \rho
  ::=
  0
  \mid
  1
  \mid
  \tau \times \sigma
  \mid
  \tau \sqcup \sigma
  \mid
  \srcreal[k] \left( k \in \N \right)
  \mid
  \opn[b]{M} \tau
\end{equation}
with \( \srcreal \) being used as an abbreviation for \( \srcreal[1] \).

Similarly, the term construction grammar is the one in Equation~\ref{eq:appendix-source-language-term-construction-grammar}, denoting variables by \( x, y, z, \dots \) and contexts as \( \Gamma = x_{1} : \tau_{1}, \cdots, x_{n} : \tau_{n} \)
\begin{equation}
  \label{eq:appendix-source-language-term-construction-grammar}
  \begin{aligned}
    t, s, r ::=
    &
      \left.
      x
      \mid
      \tuple{\;}
      \mid
      \tuple{t, s}
      \mid
      \opn[b]{pr}_{1} t
      \mid
      \opn[b]{pr}_{2} t
      \mid
      \opn[b]{inl} t
      \mid
      \opn[b]{inr} t
      \right. 
    \\
    &
      \left.
      \mid
      \casest{t}{x}{s}{y}{r}
      \right. 
    \left.
    \mid
    \fn[n]{ \opn[b]{abort} }{ t }
    \right. 
    \\
    &
      \left.
      \mid
      \letst{x}{t}{s}
      \right. 
    \left.
    \mid
    \fn[n]{ \srcop }{ t_{1}, \dots, t_{k} }
    \right. 
      \left.
      \mid
      \fn[n]{ \opn[b]{return} }{ t }
      \right. 
    \\
    &
      \left.
      \mid
      \fn[n]{ \opn[b]{bind} }{ t; x . s }
      \right. 
      \left.
      \mid
      \fn[n]{ \opn[b]{categorical} }{
      \tuple{
      \tuple{
      t_{i}, w_{i}
      }
      }_{i = 1}^{n}
      }
      \right. 
    \left.
    \mid
    \fn[n]{ \opn[b]{E}_{\nu} }{ t }
    \right. 
  \end{aligned}
\end{equation}
where \( \nu = \srcreal[p] \), for some \( p \in \N \), is the fixed type on which the handler algebra \( \opn[b]{E}_{\nu} \) is defined, and \( \srcop \) ranges over a family of chosen built-in operations, each with a signature \( \srcop: \srcreal[n_{1}] \times \cdots \times \srcreal[n_{k}] \to \srcreal[m] \), representing the basic operations available in the source language. These operations should be thought of as \emph{deterministic operations}, and are chosen according to the intended usage of the language. In the \( \opn[b]{categorical} \) constructor, the second components \( w_i \) are real-valued log-weights; the corresponding nonnegative weights are produced by the interpretation of \( \opn[b]{categorical} \).

\subsection{Typing rules}
\label{sec:appendix-source-language-typing-rules}

The typing rules for the language are as follows.

\subsubsection{Nullary product and coproduct}

The elimination rule for nullary sum can be seen as a type-theoretic instance of the explosion principle, whereas the introduction rule for the nullary product states the existence of an element of \(1\) in any context,
\begin{equation*}
  \prftree[r]{\( 0 \text{E} \)}{\Gamma \vdash t : 0}{\Gamma \vdash \fn[n]{ \opn[b]{abort} }{ t } : \rho}
  \qquad
  \prftree[r]{\( 1 \text{I} \)}{\;}{\Gamma \vdash \tuple{\;} : 1}
\end{equation*}

\subsubsection{Binary products}

The introduction and elimination rules of binary products behave similarly to their set-based semantics, cartesian products.
\begin{equation*}
  \prftree[r]{\( \times \text{E1} \)}{
    \Gamma \vdash u: \tau \times \sigma
  }{
    \Gamma \vdash \opn[b]{pr}_{1} u : \tau
  }
  \qquad
  \prftree[r]{\( \times \text{E2} \)}{
    \Gamma \vdash u: \tau \times \sigma
  }{
    \Gamma \vdash \opn[b]{pr}_{2} u : \sigma
  }
  \qquad
  \prftree[r]{\( \times \text{I} \)}{\Gamma \vdash t: \tau}{\Gamma \vdash s: \sigma}{
    \Gamma \vdash \tuple{t, s} : \tau \times \sigma
  }
\end{equation*}

\subsubsection{Binary sums}

As with binary products, binary sums typing rules are based on their intended set-based semantics, the disjoint union. The elimination rule is
\begin{equation*}
  \prftree[r]{\( \sqcup \text{E} \)}{
    \Gamma \vdash t: \tau \sqcup \sigma
  }{
    \Gamma, x: \tau \vdash u : \rho
  }{
    \Gamma, y: \sigma \vdash v : \rho
  }{
    \Gamma \vdash
    \casest{t}{x}{u}{y}{v}: \rho
  }
\end{equation*}
and the introduction rules are
\begin{equation*}
  \prftree[r]{\( \sqcup \text{I1} \)}{
    \Gamma \vdash t: \tau
  }{
    \Gamma \vdash \opn[b]{inl} t : \tau \sqcup \sigma
  }
  \qquad
  \prftree[r]{\( \sqcup \text{I2} \)}{
    \Gamma \vdash s: \sigma
  }{
    \Gamma \vdash \opn[b]{inr} s : \tau \sqcup \sigma
  }
\end{equation*}

\subsubsection{Operations}

Any built-in operation \( \srcop: \srcreal[n_{1}] \times \cdots \times \srcreal[n_{k}] \to \srcreal[m] \) defines a typing rule
\begin{equation*}
\prftree[r]{\textsc{Operations}}{
\Gamma \vdash t_{j}: \srcreal[n_{j}]
\quad
\left( j \in \left\{ 1, \dots, k \right\} \right)
}{
\Gamma \vdash
\fn[n]{ \srcop }{ t_{1}, \dots, t_{k} } :
\srcreal[m]
}
\end{equation*}

\subsubsection{Structural rules}

We add a structural rule for the let-binding,
\begin{equation*}
  \prftree[r]{\textsc{Let introduction}}{
    \Gamma \vdash t: \sigma
  }{
    \Gamma, x : \sigma \vdash s : \rho
  }{
    \Gamma \vdash \letst{x}{t}{s} : \rho
  }
\end{equation*}
and another one for assumptions, meaning that variables already in the context may be taken as terms,
\begin{equation*}
  \prftree[r]{\textsc{ Assumption }}{
    \left( x : \tau \right) \in \Gamma
  }{
    \Gamma \vdash x : \tau
  }
\end{equation*}

\subsubsection{Syntactic sugar and pattern-matching}

The \(n\)-ary product and coproduct notation used below is syntactic
sugar for the binary grammar.  We take both products and coproducts to be
right-associated:
\[
  \tau_1 \times \cdots \times \tau_n
  :=
  \begin{cases}
    1, & n=0,\\
    \tau_1, & n=1,\\
    \tau_1 \times (\tau_2 \times \cdots \times \tau_n), & n\geq 2,
  \end{cases}
\]
and similarly
\[
  \tau_1 \sqcup \cdots \sqcup \tau_n
  :=
  \begin{cases}
    0, & n=0,\\
    \tau_1, & n=1,\\
    \tau_1 \sqcup (\tau_2 \sqcup \cdots \sqcup \tau_n), & n\geq 2.
  \end{cases}
\]

The corresponding term notation is defined recursively by
\[
  \tuple{t_1,\dots,t_n}
  :=
  \begin{cases}
    \tuple{\;}, & n=0,\\
    t_1, & n=1,\\
    \tuple{t_1,\tuple{t_2,\dots,t_n}}, & n\geq 2,
  \end{cases}
\]
and by
\[
  \opn[b]{in}^{n}_{1}t:=\opn[b]{inl}t
  \quad (n\geq 2),
  \qquad
  \opn[b]{in}^{n}_{i}t
  :=
  \opn[b]{inr}(\opn[b]{in}^{n-1}_{i-1}t)
  \quad (2\leq i\leq n),
\]
with \(\opn[b]{in}^{1}_{1}t:=t\).  We write
\(\opn[b]{in}_{i}t\) when \(n\) is clear.

The \(n\)-ary coproduct case expression is defined by nested binary
cases:
\[
  \casenst{t}{i}{1}{x_i}{s_i}
  :=
  \letst{x_1}{t}{s_1},
\]
and, for \(n\geq 2\),
\[
  \casenst{t}{i}{n}{x_i}{s_i}
  :=
  \casest{t}
    {x_1}{s_1}
    {y}{
      \casenst{y}{j}{n-1}{x_{j+1}}{s_{j+1}}
    },
\]
where \(y\) is fresh.

Similarly, \(n\)-ary product pattern matching is defined by nested
projections:
\[
  \caseprod{t}{x_1}{s}
  :=
  \letst{x_1}{t}{s},
\]
and, for \(n\geq 2\),
\[
  \caseprod{t}{x_1,\dots,x_n}{s}
  :=
  \letst{z}{t}{
    \letst{x_1}{\opn[b]{pr}_{1}z}{
      \caseprod{\opn[b]{pr}_{2}z}{x_2,\dots,x_n}{s}
    }
  },
\]
where \(z\) is fresh.  With these definitions, the displayed
\(n\)-ary product and coproduct typing rules are admissible.

\subsubsection{Monadic types}

For monadic types, we add typing rules for the \( \opn[b]{return} \) and \( \opn[b]{bind} \) term constructions. Together with their equational rules, this ensures that they behave as the unit and Kleisli extension of a strong monad.
\begin{equation*}
  \prftree[r]{\( \opn[b]{M} \text{E} \)}{
    \Gamma \vdash t : \opn[b]{M} \sigma
  }{
    \Gamma, x : \sigma \vdash s : \opn[b]{M} \tau
  }{
    \Gamma \vdash
    \fn[n]{ \opn[b]{bind} }{ t; x . s } : \opn[b]{M} \tau
  }
  \qquad
  \prftree[r]{\( \opn[b]{M} \text{I} \)}{
    \Gamma \vdash t : \tau
  }{
    \Gamma \vdash \fn[n]{ \opn[b]{return} }{ t }: \opn[b]{M} \tau
  }
\end{equation*}

\subsubsection{Measure-specific rules}

All the rules above would remain the same regardless of which monadic effect one wishes to implement in the source language. The ones that follow are specific to the FAD monad and describe the construction of terms using \( \opn[b]{categorical} \) with real-valued log-weights and the handler algebra \( \opn[b]{E}_{\nu} \),
\begin{equation*}
  \prftree[r]{\textsc{ Categ. }}{
    \Gamma \vdash t_{i}: \tau
    \quad
    \left( i \in \left\{ 1, \dots, n \right\} \right)
  }{
    \Gamma \vdash w_{i}: \srcreal
    \quad
    \left( i \in \left\{ 1, \dots, n \right\} \right)
  }{
    \Gamma \vdash
    \fn[n]{ \opn[b]{categorical} }{
      \tuple{
        \tuple{t_{i}, w_{i}}
      }_{i = 1}^{n}
    } :
    \opn[b]{M} \tau
  }
  \quad
  \prftree[r]{\textsc{ Handler }}{
    \Gamma \vdash t : \opn[b]{M} \nu
  }{
    \Gamma \vdash \fn[n]{ \opn[b]{E}_{\nu} }{ t } : \nu
  }
\end{equation*}
Here the terms \( w_i \) are log-weights rather than weights themselves; in the \( \Set \)-interpretation below, \( \opn[b]{categorical} \) exponentiates them, so no positivity premise is needed at the source type level.

\subsection{Equational theory}
\label{sec:appendix-source-language-equational-theory}

The pure fragment has the usual \( \beta \)- and \( \eta \)-laws for \( 1 \), \( \times \), and \( \sqcup \), which are omitted (see Section II.2 of \cite{lambekIntroductionHigherOrder1986}). For the monadic and the handler fragments, we add the following rules. No additional definitional equations are imposed on the \( \opn[b]{categorical} \) constructors in the FAD source language treated here.

\subsubsection{Kleisli laws}

We consider the usual equational theory of strong monads.
In what follows, \( \equiv \) denotes definitional equality,
\begin{equation}
  \label{eq:appendix-source-language-kleisli-laws}
  \begin{gathered}
    \fn[n]{ \opn[b]{bind} }{ \fn[n]{ \opn[b]{return} }{ t }; x \; . \; s }
    \equiv
    \casubst{s}{x}{t}
    \qquad
    \fn[n]{ \opn[b]{bind} }{ t; x \; . \; \fn[n]{ \opn[b]{return} }{ x } }
    \equiv t
    \\
    \fn[n]{ \opn[b]{bind} }{
      \fn[n]{ \opn[b]{bind} }{ t; x \; . \; s };
      y \; . \; r
    }
    \equiv
    \fn[n]{ \opn[b]{bind} }{
      t;
      x \; . \; \fn[n]{ \opn[b]{bind} }{ s ; y \; . \; r }
    }
  \end{gathered}
\end{equation}
where we use \( \casubst{s}{x}{t} \) for the capture-avoiding substitution of the term \( t \) for the variable \( x \) in the expression \( s \).

\subsubsection{Eilenberg--Moore algebra laws}

We also add laws ensuring that the handler algebra \( \opn[b]{E}_{\nu} \) acts as an Eilenberg--Moore algebra.
\begin{equation}
  \label{eq:appendix-source-language-eilenberg-moore-algebra-laws}
  \fn[n]{ \opn[b]{E}_{\nu} }{ \fn[n]{ \opn[b]{return} }{ v } }
  \equiv v
  \qquad
  \fn[n]{ \opn[b]{E}_{\nu} }{ \fn[n]{ \opn[b]{bind} }{ t; x \; . \; u } }
  \equiv
  \fn[n]{ \opn[b]{E}_{\nu} }{
    \fn[n]{ \opn[b]{bind} }{
      t; x \; . \; \fn[n]{ \opn[b]{return} }{ \fn[n]{ \opn[b]{E}_{\nu} }{ u } }
    }
  }
\end{equation}

\subsection{Denotational semantics}
\label{sec:appendix-source-language-denotational-semantics}

We denote by \( \Syn \) the term model of the source language, i.e., the \emph{syntactic category of the source language}. Universally, it is the freely generated distributive category on the objects \( \srcreal[k] \), \( k \in \N \), with operations \( \srcop \) ranging in the family of built-in operations, equipped with a strong monad, a specific algebra on \( \srcreal[p] \) for some \( p \in \N \), and a categorical operation. This means that any category with this structure admits a structure-preserving functor from \( \Syn \), as is spelled out in Lemma~\ref{lem:appendix-universal-property-of-the-syntactic-category-of-the-source-language}.

In particular, we may interpret \( \Syn \) in \( \Set \), thus providing set-based semantics to the source language with a functor \( \setint{\cdot}: \Syn \to \Set \). This semantics is defined recursively on the term and type structures, and is spelled out as follows.

\subsubsection{Types}

\begin{equation*}
  \begin{aligned}
    \setint{0} &= \emptyset
    \\
    \setint{1} &= \left\{ \ast \right\}
    \\
    \setint{\srcreal[n]} &= \R[n]
  \end{aligned}
  \qquad
  \begin{aligned}
    \setint{\tau \times \sigma} &= \setint{\tau} \times \setint{\sigma}
    \\
    \setint{\tau \sqcup \sigma} &= \setint{\tau} \amalg \setint{\sigma}
    \\
    \setint{\opn[b]{M} \tau} &= \Fun[nn]{ M }{ \setint{\tau} }
  \end{aligned}
\end{equation*}

\subsubsection{Contexts and typing judgements}

Using the interpretations of types, we may define interpretations of contexts and of typing judgements: a context \( \Gamma = x_{1} : \tau_{1}, \cdots, x_{n} : \tau_{n} \) is interpreted as \( \setint{\Gamma} = \prod_{i = 1}^{n} \setint{\tau_{i}} \), and a typing judgement \( \Gamma \vdash t : \tau \) as a function \( \setint{\Gamma \vdash t : \tau}: \setint{\Gamma} \to \setint{\tau} \), whose definition is provided inductively on the term structure.

\subsubsection{Deterministic fragment}

Variables, constants, and constructors of the deterministic fragment are interpreted as follows, for all \( \gamma = \tuple{\gamma_{1}, \dots, \gamma_{n}} \in \setint{\Gamma} \),
\begin{equation*}
  \begin{aligned}
    \fn[n]{ \setint{\Gamma \vdash x_{j} : \tau_{j}} }{\tuple{\gamma_{1}, \dots, \gamma_{n}}}
    &= \gamma_{j}
      && \left( j \in \left\{ 1, \dots, n \right\} \right)
    \\
    \fn[n]{ \setint{\Gamma \vdash \tuple{\;} : 1} }{\gamma}
    &= \ast
    \\
    \fn[n]{ \setint{\Gamma \vdash \opn[b]{pr}_{i} u : \tau_{i}} }{\gamma}
    &= \fn[n]{ \pi_{i} }{ \fn[n]{ \setint{\Gamma \vdash u : \tau_{1} \times \tau_{2}} }{ \gamma } }
      && \left( i \in \left\{ 1, 2 \right\} \right)
    \\
    \fn[n]{ \setint{\Gamma \vdash \tuple{ t, s } : \tau \times \sigma} }{\gamma}
    &=
      \tuple{
      \fn[n]{ \setint{ \Gamma \vdash t : \tau } }{ \gamma },
      \fn[n]{ \setint{ \Gamma \vdash s : \sigma } }{ \gamma }
      }
    \\
    \fn[n]{ \setint{\Gamma \vdash \opn[b]{inl} u : \tau \sqcup \sigma} }{\gamma}
    &= \fn[n]{ \iota_{1} }{ \fn[n]{ \setint{\Gamma \vdash u: \tau} }{ \gamma } }
    \\
    \fn[n]{ \setint{\Gamma \vdash \opn[b]{inr} u : \tau \sqcup \sigma} }{\gamma}
    &= \fn[n]{ \iota_{2} }{ \fn[n]{ \setint{\Gamma \vdash u: \sigma} }{ \gamma } }
    \\
    \fn[n]{ \setint{\Gamma \vdash \opn[b]{abort} t : \rho} }{\gamma}
    &= \fn[n]{ \opn{abs}_{\setint{\rho}} }{
      \fn[n]{ \setint{\Gamma \vdash t : 0} }{\gamma}
    }
    \\
    \fn[n]{ \setint{\Gamma \vdash \letst{x}{t}{s} : \rho} }{ \gamma }
    &=
      \fn[n]{ \setint{\Gamma, x : \sigma \vdash s : \rho} }{
      \gamma,
      \fn[n]{ \setint{\Gamma \vdash t : \sigma} }{ \gamma }
      }
  \end{aligned}
\end{equation*}
where, for any set \( S \), \( \opn{abs}_{S}: \emptyset \to S \) is the unique function. Case statements are interpreted as branching functions,
\begin{equation*}
  \begin{aligned}
    &\fn[n]{ \setint{\Gamma \vdash \casest{p}{x}{u}{y}{v} : \rho} }{ \gamma }
      \qquad
    \\
    &=
      \begin{cases}
        \fn[n]{ \setint{\Gamma, x : \tau \vdash u : \rho} }{ \gamma, a }
        &\text{if }
          \fn[n]{ \setint{\Gamma \vdash p : \tau \sqcup \sigma} }{ \gamma }
          = \fn[n]{ \iota_{1} }{ a }
        \\
        \fn[n]{ \setint{\Gamma, y : \sigma \vdash v : \rho} }{ \gamma, b }
        &\text{if }
          \fn[n]{ \setint{\Gamma \vdash p : \tau \sqcup \sigma} }{ \gamma }
          = \fn[n]{ \iota_{2} }{ b }
      \end{cases}
  \end{aligned}
\end{equation*}
and the built-in operations are interpreted as their representing counterparts, that is, for any operation \( \srcop: \srcreal[n_{1}] \times \cdots \times \srcreal[n_{k}] \to \srcreal[m] \),
\begin{equation*}
  \fn[n]{ \setint{ \Gamma \vdash \fn[n]{ \srcop }{ t_{1}, \cdots, t_{k} } : \srcreal[m] } }{ \gamma }
  =
  \fn[n]{ \setint{\srcop} }{
    \fn[n]{ \setint{\Gamma \vdash t_{1} : \srcreal[n_{1}]} }{ \gamma },
    \cdots,
    \fn[n]{ \setint{\Gamma \vdash t_{k} : \srcreal[n_{k}]} }{ \gamma }
  }
\end{equation*}
where \( \setint{\srcop} \) is required to be a differentiable function, as we are interested in computing derivatives.

\subsubsection{Monadic fragment}

For the interpretation of monadic types and their associated terms, we assume that the category on which the semantics is to be interpreted has a strong monad \( \tuple{ M, \eta, \mu, \str } \), with strength components \( \str_{A,B}: A \times M B \to M(A \times B) \). In the case of \( \Set \), all monads are canonically strong. We interpret \( \opn[b]{return} \) and \( \opn[b]{bind} \) as
\begin{equation*}
  \begin{aligned}
    \setint{\Gamma \vdash \fn[n]{ \opn[b]{return} }{ t } : \opn[b]{M} \tau} 
    &= 
      \Nat[nn]{ \eta }{ \setint{\tau} }
      \circ
      \setint{\Gamma \vdash t : \tau}
    \\
    \setint{\Gamma \vdash \fn[n]{ \opn[b]{bind} }{ t; x \; . \; s } : \opn[b]{M} \tau} 
    &=
      \Nat[nn]{ \mu }{ \setint{\tau} }
      \circ
      \Fun[np]{ M }{
      \setint{\Gamma, x : \sigma \vdash s : \opn[b]{M} \tau}
      }
      \circ
      \Nat[nn]{ \str }{ \setint{\Gamma}, \setint{\sigma} }
      \circ
      \tupling{
      \id_{\setint{\Gamma}},
      \setint{\Gamma \vdash t : \opn[b]{M} \sigma}
      }
  \end{aligned}
\end{equation*}
Notice that these formulas are also applicable for interpretations on other distributive categories with strong monads.

\subsubsection{Measure-specific terms}

Finally, we present the interpretation of terms constructed with \( \opn[b]{categorical} \) or with the handler algebra \( \opn[b]{E}_{\nu} \). For \( \opn[b]{categorical} \), the second components are interpreted as log-weights:
\begin{equation*}
  \begin{aligned}
    &\fn[n]{ \setint{\Gamma \vdash \fn[n]{ \opn[b]{categorical} }{ \tuple{ \tuple{ t_{i}, w_{i} } }_{i = 1}^{n} } : \opn[b]{M} \tau} }{ \gamma }
    \\
    &\qquad =
      \fn[n]{ \Nat[nn]{ \categ }{ \setint{\tau} } }{
      \tuple{
        \tuple{
          \fn[n]{ \setint{\Gamma \vdash t_{i} : \tau} }{ \gamma },
          \fn[n]{ \setint{\Gamma \vdash w_{i} : \srcreal} }{ \gamma }
        }
      }_{i = 1}^{n}
      }
  \end{aligned}
\end{equation*}
and
\begin{equation*}
  \begin{aligned}
    \fn[n]{ \setint{\Gamma \vdash \fn[n]{ \opn[b]{E}_{\nu} }{ t } : \nu }}{ \gamma }
    &=
    \fn[n]{ \opn{alg}^{\Set}_{\nu} }{
    \fn[n]{ \setint{\Gamma \vdash t : \opn[b]{M} \nu} }{ \gamma }
    }
    \\
    &=
      \sum_{
      v \in 
      \opn{supp}\!
      \left(
      \fn[n]{ \setint{\Gamma \vdash t : \opn[b]{M} \nu} }{ \gamma }
      \right)
      }
      \fn[n]{
      \fn[n]{ \setint{\Gamma \vdash t : \opn[b]{M} \nu} }{ \gamma }
      }{ v }
      \cdot v
  \end{aligned}
\end{equation*}

\subsection{Soundness and initiality}

As claimed, \( \Syn \) can be seen as a free distributive category with a chosen strong monad on the objects \( \srcreal[k] \), \( k \in \N \), and the built-in family of operations \( \srcop \), as well as \( \opn[b]{categorical} \). As a freely generated category, it has a universal property, described in Lemma~\ref{lem:appendix-universal-property-of-the-syntactic-category-of-the-source-language} (restatement of Lemma~\ref{lem:universal-property-of-the-syntactic-category-of-the-source-language}), which describes the existence of unique interpretation functors. These interpretations are also \emph{sound}, meaning that any equation from the equational theory holds semantically, as stated in Lemma~\ref{lem:appendix-soundness-of-source-language-interpretations} (restatement of Lemma~\ref{lem:soundness-of-source-language-interpretations}).

\begin{lemma}[Soundness of source language interpretations]
  \label{lem:appendix-soundness-of-source-language-interpretations}

  For any interpretation \( \genint{\cdot}: \Syn \to \cat[]{X} \), if \( \Gamma \vdash t \equiv s : \tau \) is derivable in the equational theory above, then \( \genint{\Gamma \vdash t : \tau} = \genint{\Gamma \vdash s : \tau} \) as morphisms \( \genint{\Gamma} \to \genint{\tau} \) in \( \cat[]{X} \).
\end{lemma}

The proof follows by routine verification on the ways to obtain a definitional equality \( t \equiv s \) in the source language. The soundness of the (omitted) pure fragment laws follows from the distributive category laws; there are no additional \( \opn[b]{categorical} \)-equations to check.

On the other hand, the soundness of term equalities with \( \opn[b]{return} \) and \( \opn[b]{bind} \) follows from the interpretation of \( \opn[b]{bind} \) with the strength \( \str \) of the monad \( M \), together with the Kleisli equations and coherence laws of strong monads, while the definitional equalities for the handler algebra follow from the Eilenberg--Moore algebra axioms.

The universal property of the syntactic category can be given explicitly as follows.

\begin{lemma}[Universal property of \( \Syn \)]
  \label{lem:appendix-universal-property-of-the-syntactic-category-of-the-source-language}

  Let \( \cat[]{X} \) be any distributive category equipped with
  \begin{itemize}
  \item A strong monad \( \tuple{ M^{\cat[]{X}}, \eta^{\cat[]{X}}, \mu^{\cat[]{X}}, \str^{\cat[]{X}} } \), with \( \str^{\cat[]{X}}_{A,B}: A \times \Fun[nn]{M^{\cat[]{X}}}{B} \to \Fun[nn]{M^{\cat[]{X}}}{A \times B} \);
  \item For every \( n \in \N \), chosen objects \( X_{n} \);
  \item An Eilenberg--Moore algebra \( \opn{alg}: \Fun[nn]{ M^{\cat[]{X}} }{ X_{p} } \to X_{p} \) over the object \( X_{p} \);
  \item For every built-in operation \( \srcop: \srcreal[n_{1}] \times \cdots \times \srcreal[n_{k}] \to \srcreal[m] \), a chosen morphism \( \genint{\srcop}: X_{n_{1}} \times \cdots \times X_{n_{k}} \to X_{m} \);
  \item For each \( n \in \N \) and each object \( A \) interpreting a type \( \tau \), an arrow \( \categ^{\cat[]{X}}_{\tau, n}: \left( A \times X_{1} \right)^{n} \to \Fun[nn]{ M^{\cat[]{X}} }{ A } \), used to interpret \( \opn[b]{categorical} \) terms built from lists of length \( n \) of atoms of type \( \tau \) and log-weights in \( X_{1} \).
  \end{itemize}

  Then there exists a unique strict functor \( \genint{ \cdot }: \Syn \to \cat[]{X} \) preserving finite products and coproducts, the strong monad structure, the handler algebra, and \( \opn[b]{categorical} \) as above, and sending every equation of \( \Syn \) to an identity in \( \cat[]{X} \). 
\end{lemma}

The proof follows by standard algebraic presentation arguments, using the formulas for interpretation above, with \( \genint{\srcreal[n]} = X_{n} \) for all \( n \in \N \), and induction over the structure of the morphisms in \( \Syn \). Uniqueness follows from the fact that the equations in \( \Syn \) are precisely those required to enforce its freely generated structure.

\section{Target language for CHAD}
\label{sec:appendix-target-language-for-chad}

In this section we present the complete type-theoretic description of the target language for CHAD.

As explained in the abridged version in Section~\ref{sec:target-language}, the core idea of CHAD is to encode a code transformation as a categorical transformation between syntactic categories. As such code written in the source language (see Section~\ref{sec:appendix-source-language-for-chad}) is systematically converted by induction on its structural components, giving their differentiated versions. The target language is the language in which these transformed codes are written.

The target language is a dependently typed functional programming language, and it is split into a cartesian type theory and, for every context of cartesian types (i.e. a \emph{cartesian context}), an associated linear type theory.

Denoting by \( \CSyn \) the term model of the cartesian fragment, each cartesian context \( \Gamma \) (viewed as an object of \( \CSyn \), up to variable renaming) determines a syntactic category \( \LSyn[\Gamma] \) for the associated linear type theory. Any substitution \( \Gamma \to \Gamma' \) induces a reindexing functor \( \LSyn[\Gamma'] \to \LSyn[\Gamma] \) between these linear term models. Thus \( \LSyn \) is a strict indexed category \( \CSyn[o] \to \Cat \).

The syntactic category of the target language is then given by the Grothendieck construction \( \Tgt \) (see Section 2.3 of \cite{lucatellinunesCHADExpressiveTotal2023}).

As in the source language, we write \( \tau, \sigma, \rho \) for the \emph{cartesian types}, and reserve underlined Greek letters \( \lintype{\tau}, \lintype{\sigma}, \lintype{\rho} \) for linear types. We then have \emph{cartesian typing judgements} of the form \( \Gamma \vdash t : \tau \) and \emph{linear typing judgements} of the form \( \Gamma ; v : \lintype{\sigma} \vdash s : \lintype{\rho} \), with exactly one distinguished linear variable, whose name may be \(\alpha\)-renamed.

We can divide the description of the target language into three layers:
\begin{itemize}
\item \textbf{Non-monadic core}: deterministic operations. The description of this part is independent of any monad-related components.
\item \textbf{Oplax strong-monad core}: liftings possible for any strong monad in the source category.
\item \textbf{Distribution and handler}: language features that are specific to the FAD monad.
\end{itemize}
By breaking the target language into these layers, we get a generic procedure for introducing (strong) monadic effects into the automatic differentiation framework of CHAD.

\subsection{Type construction grammar}

Cartesian types are described by the following grammar:
\begin{equation}
  \label{eq:appendix-target-language-cartesian-type-construction-grammar}
  \tau, \sigma, \rho
  ::=
  \left( \text{As in Equation~\eqref{eq:source-language-type-construction-grammar}} \right)
  \mid
  \lintype{\sigma} \linto \lintype{\rho}
  \mid
  \Sigma \left( x : \sigma \right) . \tau'
\end{equation}
meaning that, in the cartesian fragment, we also have function types between linear types and dependent sums. For linear types, we use the grammar
\begin{equation}
  \label{eq:appendix-target-language-linear-type-construction-grammar}
  \lintype{\tau}, \lintype{\sigma}, \lintype{\rho} ::=
  \lintype{\srcreal[k]}
  \; \left( k \in \N \right) 
  \mid
  \lintype{1}
  \mid
  \lintype{\tau} \times \lintype{\sigma}
  \mid
  \opn[b]{case} \; t \; \opn[b]{of} \; \left\{ \opn[b]{in}_{i} x_{i} \mapsto \lintype{\sigma}_{i} \right\}_{i = 1}^{n}
  \mid
  \opn[b]{M}^{w}_{p : \Fun[nn]{ \opn[b]{M} }{ \tau }}
  \left[ x \mapsto \lintype{ \sigma } \right]
\end{equation}
meaning that the base Euclidean-space types have associated linear types (in the sense that the cotangent space of a Euclidean space at any point can be identified with the base space itself), as well as the fibre-wise zero object \( 1 \). Moreover, we allow fibre-wise binary biproducts \( \times \) and families of linear types indexed by finite disjunctions, and we introduce the cotangent linear type associated with monadic types. 

The goal of highlighting the different layers is to make explicit which constructions are general for the implementation of algebraic monadic effects and which ones are specific to the FAD monad.

\subsection{Term construction grammar}

As with types, we also distinguish between \emph{cartesian} and \emph{linear terms}. The cartesian terms are described by the grammar
\begin{equation}
  \label{eq:appendix-target-language-cartesian-term-construction-grammar}
  t, s, r ::=
  \left( \text{As in Equation~\eqref{eq:appendix-source-language-term-construction-grammar}} \right) 
  \mid
  \tuple{ a, t }
  \mid
  \opn[b]{case} \; p \; \opn[b]{of} \; \tuple{ x, y } \mapsto s
  \mid
  \lambda v \; . \; s
\end{equation}
where the constructor \( \tuple{ a, t } \) and the eliminator \( \opn[b]{case} \; p \; \opn[b]{of} \; \tuple{ x, y } \mapsto s \) are used in the typing rules for the cartesian \( \Sigma \)-type, and \( \lambda v \; . \; s \) is used for abstraction over the unique linear variable \( v \).

Since all term constructions from the source language (Equation~\eqref{eq:appendix-source-language-term-construction-grammar}) are present in the cartesian fragment of the target language, with the same typing rules and equational theory, the source syntactic category \( \Syn \) can be seen as a full subcategory of \( \CSyn \).

For linear terms, we use the following grammar:
\begin{equation}
  \label{eq:appendix-target-language-linear-term-construction-grammar}
  \begin{aligned}
    t, s, r
    ::=
    v
    &\mid
      \letst{v}{t}{s}
      \mid
      \tuple{ t, s }
      \mid
      \opn[b]{pr}_{i} t
      \mid
      \opn[b]{0}
      \mid
      t + s
      \mid
      \linapp{r}{t}
    \\
    &
      \mid
      \fn[n]{ \fn[n]{ \tgtop }{ t_{1}, \cdots, t_{k} } }{v}
      \mid
      \fn[n]{ \opn[b]{return}^{w}_{t} }{s}
      \mid
      \fn[n]{ \opn[b]{bind}^{w}_{s; x \; . \; t} }{r; u}
    \\
    &
      \mid
      \fn[n]{ \opn[b]{categorical}^{w}_{\tuple{ \tuple{ t_{i}, w_{i} } }_{i = 1}^{n}} }{
      \tuple{ \tuple{ t_{i}', w_{i}' } }_{i = 1}^{n}
      }
      \mid
      \fn[n]{ \opn[b]{E}^{w}_{\nu, t} }{ s }
  \end{aligned}
\end{equation}
The three lines correspond to the three layers: the first line consists of the non-monadic core operations (linear variables, let bindings, pairings and projections, the zero vector, sums, and \( \linapp{r}{t} \) for linear application with \( r \) a cartesian term of type \( \lintype{\sigma} \linto \lintype{\rho} \)). The second line has the derivatives of the built-in operations \( \srcop \) of the source language, as well as the linear fibre components of the monad operations \( \opn[b]{return} \) and \( \opn[b]{bind} \). Finally, the third line has the linear fibre components of \( \opn[b]{categorical} \) and of the handler algebra \( \opn[b]{E}_{\nu} \).

\subsection{Typing rules}
\subsubsection{Non-monadic core}

The cartesian typing rules are the same as the ones of the source language (see Section~\ref{sec:appendix-source-language-typing-rules}) for the common term constructors. For the dependent sums, we have the rules
\begin{equation*}
  \prftree{
    \Gamma \vdash a : \sigma
  }{
    \Gamma \vdash t : \casubst{\tau}{x}{a}
  }{
    \Gamma \vdash \tuple{ a, t }: \Sigma \left( x : \sigma \right) . \tau
  }
  \qquad
  \prftree{
    \Gamma \vdash \rho \text{ type}
  }{
    \Gamma \vdash p : \Sigma \left( x : \sigma \right) . \tau
  }{
    \Gamma, x : \sigma, y : \tau \vdash s : \rho
  }{
    \Gamma \vdash \opn[b]{case} \; p \; \opn[b]{of} \; \tuple{ x, y } \to s : \rho
  }
\end{equation*}
and for the linear function types, we introduce the typing rules
\begin{equation*}
  \prftree{
    \Gamma; v : \lintype{\sigma} \vdash s : \lintype{\rho}
  }{
    \Gamma \vdash \lambda v \; . \; s : \lintype{\sigma} \linto \lintype{\rho}
  }
  \qquad
  \prftree{
    \Gamma \vdash r : \lintype{\sigma} \linto \lintype{\rho}
  }{
    \Gamma; v : \lintype{\tau} \vdash t : \lintype{\sigma}
  }{
    \Gamma; v : \lintype{\tau} \vdash
    \linapp{r}{t} : \lintype{\rho}
  }
\end{equation*}

There are also non-monadic-layer typing rules for linear terms. 

\paragraph{Linear variable}

\begin{equation*}
  \prftree{}{
    \Gamma; v : \lintype{\sigma} \vdash v : \lintype{\sigma}
  }
  \qquad\qquad
  \textit{(one distinguished linear variable)}
\end{equation*}

\paragraph{Linear let}

\begin{equation*}
  \prftree{
    \Gamma; v : \lintype{\tau} \vdash t : \lintype{\sigma}
  }{
    \Gamma; v : \lintype{\sigma} \vdash s : \lintype{\rho}
  }{
    \Gamma; v : \lintype{\tau} \vdash \letst{v}{t}{s} : \lintype{\rho}
  }
\end{equation*}

\paragraph{Linear biproduct}

\begin{equation*}
  \begin{gathered}
  \prftree{}{
    \Gamma ; v : \lintype{\rho} \vdash \opn[b]{0} : \lintype{\sigma}
  }
  \qquad
  \prftree{
    \Gamma; v : \lintype{\rho} \vdash s : \lintype{\sigma}
  }{
    \Gamma; v : \lintype{\rho} \vdash t : \lintype{\sigma}
  }{
    \Gamma; v : \lintype{\rho} \vdash s + t : \lintype{\sigma}
  }
  \\
  \prftree{
    \Gamma; v : \lintype{\rho} \vdash s : \lintype{\tau}
  }{
    \Gamma; v : \lintype{\rho} \vdash t : \lintype{\sigma}
  }{
    \Gamma; v : \lintype{\rho} \vdash \tuple{ s, t } : \lintype{\tau} \times \lintype{\sigma}
  }
  \qquad
  \prftree{
    \Gamma; v : \lintype{\rho} \vdash
    u : \lintype{\tau}_{1} \times \lintype{\tau}_{2}
  }{
    \Gamma; v : \lintype{\rho} \vdash
    \opn[b]{pr}_{i} u : \lintype{\tau}_{i}
  }
  \end{gathered}
\end{equation*}
where \( i \in \left\{ 1, 2 \right\} \).

\paragraph{Linear case over cartesian sums}

\begin{equation*}
  \prftree{
    \Gamma \vdash t : \tau_{1} \sqcup \cdots \sqcup \tau_{n}
  }{
    \left\{
      \Gamma, x_{i} : \tau_{i}; v : \lintype{\sigma}_{i} \vdash r_{i} : \lintype{\rho}_{i}
    \right\}_{i=1}^{n} 
  }
  {
    \Gamma; v:
    \casenst{t}{i}{n}{x_{i}}{\lintype{\sigma}_{i}}
    \vdash
    \casenst{t}{i}{n}{x_{i}}{r_{i}} :
    \casenst{t}{i}{n}{x_{i}}{\lintype{\rho}_{i}}
   }
\end{equation*}

\paragraph{Built-in operations}

For any built-in operation \( \srcop: \srcreal[n_{1}] \times \cdots \times \srcreal[n_{k}] \to \srcreal[m] \) in the source language, we associate a term constructor \( \opn[b]{Dop} \), whose intended interpretation is as the transposed Jacobian of \( \srcop \). It has the following associated typing rule.
\begin{equation*}
  \prftree{
    \left\{
    \Gamma \vdash t_{j} : \srcreal[n_{j}] 
    \right\}_{j = 1}^{k}
  }{
    \Gamma ; v : \lintype{\srcreal[m]} \vdash
    \fn[n]{ \fn[n]{ \opn[b]{Dop} }{ t_{1}, \cdots, t_{k} } }{ v } :
    \lintype{\srcreal[n_{1}]}
    \times \cdots \times
    \lintype{\srcreal[n_{k}]}
  }
\end{equation*}

\subsubsection{Oplax strong-monad core}

The following typing rules concern a generic strong monad, together with its
chosen cotangent lift.  In every rule below, all displayed cartesian contexts are
assumed to be well formed.  A judgement such as
\( \Gamma, x : \tau \vdash \lintype{\sigma} \text{ type} \) means that
\( \lintype{\sigma} \) is a linear type in the fibre
\( \LSyn[\Gamma,x:\tau] \); a judgement
\( \Gamma \vdash \lintype{\rho} \text{ type} \) means that
\( \lintype{\rho} \) is a linear type in \( \LSyn[\Gamma] \).  The linear
variables occurring in the linear arguments of
\( \opn[b]{return}^{w} \), \( \opn[b]{bind}^{w} \),
\( \opn[b]{categorical}^{w} \), and \( \opn[b]{E}_{\nu}^{w} \) are local binding
occurrences; following the single-linear-variable convention, we write each of
them as \(v\), with scope determined by the surrounding judgement and with the
usual capture-avoiding substitution convention.

\paragraph{Type formation}

\begin{equation*}
  \prftree{
    \Gamma \vdash p : \Fun[nn]{ \opn[b]{M} }{ \tau }
  }{
    \Gamma,x:\tau \vdash \lintype{\sigma} \text{ type}
  }{
    \Gamma \vdash
    \opn[b]{M}^{w}_{p : \Fun[nn]{ \opn[b]{M} }{ \tau }}
    \left[ x \mapsto \lintype{\sigma} \right]
    \text{ type}
  }
\end{equation*}
Here \(x\) is the returned value on which the continuation cotangent type
may depend.

\paragraph{\( w \)-\( \opn[b]{return} \)}

\begin{equation*}
  \prftree{
    \Gamma \vdash t: \tau
  }{
    \Gamma, x : \tau \vdash \lintype{\sigma} \text{ type}
  }{
    \Gamma \vdash \lintype{\rho} \text{ type}
  }{
    \Gamma; v : \casubst{ \lintype{\sigma} }{x}{t}
    \vdash  s : \lintype{\rho}
  }{
    \Gamma; v : \opn[b]{M}^{w}_{\fn[n]{ \opn[b]{return} }{ t } : \Fun[nn]{ \opn[b]{M} }{ \tau }}
    \left[ x\mapsto \lintype{\sigma} \right]
    \vdash
    \fn[n]{ \opn[b]{return}^{w}_{t} }{ s }: \lintype{\rho}
  }
\end{equation*}
Thus the argument \(s\) is typed in the returned-value fibre
\(\casubst{\lintype{\sigma}}{x}{t}\), while the result consumes a cotangent of
\(\fn[n]{\opn[b]{return}}{t}\) whose continuation fibre is
\(x\mapsto\lintype{\sigma}\).

\paragraph{\( w \)-\( \opn[b]{bind} \)}

\begin{equation*}
  \begin{gathered}
    \Gamma \vdash s : \Fun[nn]{ \opn[b]{M} }{ \sigma }
    \qquad
    \Gamma, x : \sigma \vdash t : \Fun[nn]{ \opn[b]{M} }{ \tau }
    \\
    \Gamma, x : \sigma \vdash \lintype{\sigma} \text{ type}
    \qquad
    \Gamma, y : \tau \vdash \lintype{\tau} \text{ type}
    \qquad
    \Gamma \vdash \lintype{\rho} \text{ type}
    \\
    \prftree{
      \Gamma; v : \opn[b]{M}^{w}_{s : \Fun[nn]{ \opn[b]{M} }{ \sigma }}
      \left[ x\mapsto \lintype{\sigma} \right]
      \vdash r : \lintype{\rho}
    }{
      \Gamma, x : \sigma ; v : \opn[b]{M}^{w}_{t : \Fun[nn]{ \opn[b]{M} }{ \tau }}
      \left[ y \mapsto \lintype{\tau} \right]
      \vdash u : \lintype{\rho} \times \lintype{\sigma}
    }{
      \Gamma; v : \opn[b]{M}^{w}_{\fn[n]{ \opn[b]{bind} }{ s; x \; . \; t } : \Fun[nn]{ \opn[b]{M} }{ \tau }}
      \left[ y \mapsto \lintype{\tau} \right]
      \vdash \fn[n]{ \opn[b]{bind}^{w}_{s; x \; . \; t} }{ r; u } : \lintype{\rho}
    }
  \end{gathered}
\end{equation*}
Here \(\lintype{\sigma}\) is the continuation-cotangent type for the first
computation, depending on its returned value \(x:\sigma\), and
\(\lintype{\tau}\) is the continuation-cotangent type for the composite
computation, depending on its final returned value \(y:\tau\).  The second
linear argument \(u\) returns both a contribution to the external cotangent
\(\lintype{\rho}\) and a cotangent in the fibre \(\lintype{\sigma}\) of the
bound value \(x\).

\subsubsection{Measure-specific layer}

These typing rules relate specifically to the FAD monad, providing rules for the
transposed derivative of \( \opn[b]{categorical} \) and for the handler algebra
\( \opn[b]{E}_{\nu} \).

\paragraph{\( w \)-\( \opn[b]{categorical} \)}

For all \( n \in \N \), the finite atomic distribution constructor has the
following cotangent rule.
\begin{equation*}
  \begin{gathered}
    \Gamma, x : \tau \vdash \lintype{\sigma} \text{ type}
    \qquad
    \Gamma \vdash \lintype{\rho} \text{ type}
    \\
    \prftree{
      \left\{
        \Gamma \vdash t_{i} : \tau
      \right\}_{i = 1}^{n}
    }{
      \left\{
        \Gamma \vdash w_{i} : \srcreal
      \right\}_{i = 1}^{n}
    }{
      \left\{
        \Gamma; v : \casubst{\lintype{\sigma}}{x}{t_{i}}
        \vdash
        t_{i}' : \lintype{\rho}
      \right\}_{i = 1}^{n}
    }{
      \left\{
        \Gamma; v : \lintype{\srcreal} \vdash
        w_{i}' : \lintype{\rho}
      \right\}_{i = 1}^{n}
    }{
      \Gamma;
      v : \opn[b]{M}^{w}_{
        \fn[n]{ \opn[b]{categorical} }{
          \tuple{ \tuple{ t_{i}, w_{i} } }_{i = 1}^{n}
        } :
        \Fun[nn]{ \opn[b]{M} }{ \tau }
      }
      \left[ x \mapsto \lintype{\sigma} \right]
      \vdash
      \fn[n]{
        \opn[b]{categorical}^{w}_{
          \tuple{ \tuple{ t_{i}, w_{i} } }_{i = 1}^{n}
        }
      }{
        \tuple{ \tuple{ t_{i}', w_{i}' } }_{i = 1}^{n}
      } :
      \lintype{\rho}
    }
  \end{gathered}
\end{equation*}
Thus \(\lintype{\sigma}\) is the cotangent family for the sampled value, while
\(\lintype{\rho}\) is the ambient cotangent type to which each atom and weight
backpropagator contributes.

\paragraph{Transposed derivative of the handler}

Since \( \nu = \srcreal[p] \) for some \( p \in \N \), write
\( \lintype{\tau}_{\nu} := \lintype{\srcreal[p]} \) for its associated
linear type.  The constant family
\(y:\nu\mapsto\lintype{\tau}_{\nu}\) is a well-formed linear type in
\(\LSyn[\Gamma,y:\nu]\).  We impose the following typing rule.
\begin{equation*}
  \prftree{
    \Gamma \vdash t : \Fun[nn]{ \opn[b]{M} }{ \nu }
  }{
    \Gamma \vdash \lintype{\rho} \text{ type}
  }{
    \Gamma; v : \opn[b]{M}^{w}_{t : \Fun[nn]{ \opn[b]{M} }{ \nu }}
    \left[ y \mapsto \lintype{\tau}_{\nu} \right]
    \vdash s : \lintype{\rho}
  }{
    \Gamma; v : \lintype{\tau}_{\nu}
    \vdash \fn[n]{ \opn[b]{E}^{w}_{\nu, t} }{ s } : \lintype{\rho}
  }
\end{equation*}
\subsection{Equational theory}

In addition to the equations from the equational theory of the source language (Section~\ref{sec:appendix-source-language-equational-theory}), which also apply for the cartesian fragment, the following equations describe the equational theory of the target language. They are stated up to the usual capture-avoiding substitution on the bound variables.

\paragraph{Cartesian layer (\( \Sigma \), products, and sums)}

\begin{equation*}
  \begin{array}{rrl}
    \left( \Sigma \beta \right)
    &\quad
      \caseprod{\tuple{ a, t }}{x, y}{s}
    &\equiv
      s \left[ x := a, y := t \right]
    \\
    \left( \Sigma \eta \right)
    &\quad
      \caseprod{p}{x, y}{\tuple{ x, y }}
    &\equiv
      p
  \end{array}
\end{equation*}

\paragraph{Linear biproduct (fibrewise) and additivity}

\begin{equation*}
  \begin{array}{rrl}
    \left( \oplus \beta \right)
    &\quad
      \opn[b]{pr}_{1} \tuple{ s, t }
    &\equiv
      s
    \\
    &\quad
      \opn[b]{pr}_{2} \tuple{ s, t }
    &\equiv
      t
    \\
    \left( \oplus \eta \right)
    &\quad
      \tuple{ \opn[b]{pr}_{1} t, \opn[b]{pr}_{2} t }
    &\equiv
      t
    \\
    \left( \text{Add-unit} \right) 
    &\quad
      s + \opn[b]{0}
    &\equiv
      s
    \\
    &\quad
      \opn[b]{0} + s
    &\equiv
      s
    \\
    \left( \text{Add-assoc/comm} \right) 
    &\quad
      \left( r + s \right) + t
    &\equiv
      r + \left( s + t \right)
    \\
    &\quad
      s + t
    &\equiv
      t + s
    \\
    \left( \text{Add-lin} \right)
    &\quad
      \casubst{r}{v}{\opn[b]{0}}
    &\equiv
      \opn[b]{0}
    \\
    &\quad
      \casubst{r}{v}{s + t}
    &\equiv
      \casubst{r}{v}{s} + \casubst{r}{v}{t}
  \end{array}
\end{equation*}

\paragraph{Linear structural}

\begin{equation*}
  \begin{array}{rr}
    \left( \text{Lin. struct.} \right) 
    &\quad
      \letst{v}{t}{s}
    \equiv
      \casubst{s}{v}{t}
  \end{array}
\end{equation*}

\paragraph{Linear function space}

\begin{equation*}
  \begin{array}{rrl}
    \left( \linto \beta \right)
    &\quad
      \fn[p]{ \lambda v \; . \; s }{ t }
    &\equiv
      \casubst{s}{v}{t}
    \\
    \left( \linto \eta \right)
    &\quad
      \lambda v \; . \; \fn[n]{ r }{ v }
    &\equiv
      r
      \qq{when} \left( v \not\in \fn[n]{ \opn{FV} }{ r } \right)
  \end{array}
\end{equation*}

\paragraph{Built-in operations}
For every built-in operation \( \srcop: \srcreal[n_{1}] \times \cdots \times \srcreal[n_{k}] \to \srcreal[m] \), its associated transposed Jacobian term constructor \( \opn[b]{Dop} \) is additive in its linear argument and respects cartesian substitution.

\paragraph{Oplax fibred strong-monad and algebra equations}

We impose the following equations whenever all displayed terms are well-typed.
They are stated modulo the Kleisli equations and Eilenberg--Moore algebra equations 
in the cartesian fragment, so
that the linear types on the two sides are identified by the corresponding
cartesian equality.  All occurrences of the distinguished linear variable in
arguments of
\[
  \fn[n]{\opn[b]{return}^{w}_{t}}{r},
  \qquad
  \fn[n]{\opn[b]{bind}^{w}_{s;x.t}}{r;u},
  \qquad
  \fn[n]{\opn[b]{E}^{w}_{\nu,t}}{r}
\]
are binding occurrences and are implicitly \(\alpha\)-renamed as necessary to
avoid capture.
The equations we impose are precisely those of an oplax indexed strong monad and an oplax indexed Eilenberg--Moore algebra.
They are also the equations that are required to ensure that the interpretation of the source language in the target language is well-defined, i.e. the transposed derivatives of the Kleisli and Eilenberg--Moore algebra equations for the source language.

The \(w\)-left-unit law is the transpose of
\[
  \fn[n]{\opn[b]{bind}}{
    \fn[n]{\opn[b]{return}}{a};
    x.t
  }
  \equiv
  \casubst{t}{x}{a}.
\]
Suppose
\[
  \Gamma\vdash a:\sigma,
  \qquad
  \Gamma,x:\sigma\vdash t:\Fun[nn]{\opn[b]{M}}{\tau},
\]
\[
  \Gamma,x:\sigma\vdash \lintype{\sigma}\ \mathrm{type},
  \qquad
  \Gamma,y:\tau\vdash \lintype{\tau}\ \mathrm{type},
  \qquad
  \Gamma\vdash \lintype{\rho}\ \mathrm{type},
\]
and
\[
  \Gamma;
  v:\casubst{\lintype{\sigma}}{x}{a}
  \vdash r:\lintype{\rho},
  \qquad
  \Gamma,x:\sigma;
  v:\opn[b]{M}^{w}_{t:\Fun[nn]{\opn[b]{M}}{\tau}}
    \left[y\mapsto\lintype{\tau}\right]
  \vdash u:\lintype{\rho}\times\lintype{\sigma}.
\]
Then
\begin{equation}
  \label{eq:appendix-target-language-monad-w-left-unit}
  \begin{aligned}
    &
    \fn[n]{
      \opn[b]{bind}^{w}_{
        \fn[n]{\opn[b]{return}}{a};
        x.t
      }
    }{
      \fn[n]{\opn[b]{return}^{w}_{a}}{r};
      u
    }
    \\
    &\qquad\equiv
    \letst{v}{\casubst{u}{x}{a}}{
      \opn[b]{pr}_{1}v
      +
      \casubst{r}{v}{\opn[b]{pr}_{2}v}
    } .
  \end{aligned}
\end{equation}

The \(w\)-right-unit law is the transpose of
\[
  \fn[n]{\opn[b]{bind}}{
    s;
    x.\fn[n]{\opn[b]{return}}{x}
  }
  \equiv
  s.
\]
Suppose
\[
  \Gamma\vdash s:\Fun[nn]{\opn[b]{M}}{\sigma},
  \qquad
  \Gamma,x:\sigma\vdash \lintype{\sigma}\ \mathrm{type},
  \qquad
  \Gamma\vdash \lintype{\rho}\ \mathrm{type},
\]
and
\[
  \Gamma;
  v:\opn[b]{M}^{w}_{s:\Fun[nn]{\opn[b]{M}}{\sigma}}
    \left[x\mapsto\lintype{\sigma}\right]
  \vdash r:\lintype{\rho}.
\]
Then
\begin{equation}
  \label{eq:appendix-target-language-monad-w-right-unit}
  \fn[n]{
    \opn[b]{bind}^{w}_{
      s;
      x.\fn[n]{\opn[b]{return}}{x}
    }
  }{
    r;
    \tuple{
      \opn[b]{0},
      \fn[n]{\opn[b]{return}^{w}_{x}}{v}
    }
  }
  \equiv
  r .
\end{equation}
In the second argument, the term
\[
  \fn[n]{\opn[b]{return}^{w}_{x}}{v}
\]
is typed using the returned-value fibre
\[
  y:\sigma\mapsto \casubst{\lintype{\sigma}}{x}{y}.
\]
Thus it produces a term of type \(\lintype{\sigma}\) in the context
\(\Gamma,x:\sigma\).

The \(w\)-associativity law is the transpose of
\[
  \fn[n]{\opn[b]{bind}}{
    \fn[n]{\opn[b]{bind}}{s;x.t};
    y.q
  }
  \equiv
  \fn[n]{\opn[b]{bind}}{
    s;
    x.\fn[n]{\opn[b]{bind}}{t;y.q}
  }.
\]
Suppose
\[
  \Gamma\vdash s:\Fun[nn]{\opn[b]{M}}{\sigma},
  \qquad
  \Gamma,x:\sigma\vdash t:\Fun[nn]{\opn[b]{M}}{\tau},
  \qquad
  \Gamma,y:\tau\vdash q:\Fun[nn]{\opn[b]{M}}{\upsilon},
\]
\[
  \Gamma,x:\sigma\vdash \lintype{\sigma}\ \mathrm{type},
  \qquad
  \Gamma,y:\tau\vdash \lintype{\tau}\ \mathrm{type},
  \qquad
  \Gamma,z:\upsilon\vdash \lintype{\upsilon}\ \mathrm{type},
  \qquad
  \Gamma\vdash \lintype{\rho}\ \mathrm{type},
\]
and
\[
  \Gamma;
  v:\opn[b]{M}^{w}_{s:\Fun[nn]{\opn[b]{M}}{\sigma}}
    \left[x\mapsto\lintype{\sigma}\right]
  \vdash r:\lintype{\rho},
\]
\[
  \Gamma,x:\sigma;
  v:\opn[b]{M}^{w}_{t:\Fun[nn]{\opn[b]{M}}{\tau}}
    \left[y\mapsto\lintype{\tau}\right]
  \vdash u:\lintype{\rho}\times\lintype{\sigma},
\]
\[
  \Gamma,y:\tau;
  v:\opn[b]{M}^{w}_{q:\Fun[nn]{\opn[b]{M}}{\upsilon}}
    \left[z\mapsto\lintype{\upsilon}\right]
  \vdash e:\lintype{\rho}\times\lintype{\tau}.
\]
Then
\begin{equation}
  \label{eq:appendix-target-language-monad-w-associativity}
  \begin{aligned}
    &
    \fn[n]{
      \opn[b]{bind}^{w}_{
        \fn[n]{\opn[b]{bind}}{s;x.t};
        y.q
      }
    }{
      \fn[n]{\opn[b]{bind}^{w}_{s;x.t}}{r;u};
      e
    }
    \\
    &\qquad\equiv
    \fn[n]{
      \opn[b]{bind}^{w}_{
        s;
        x.\fn[n]{\opn[b]{bind}}{t;y.q}
      }
    }{
      r;
      \fn[n]{\opn[b]{bind}^{w}_{t;y.q}}{
        u;
        \letst{v}{e}{
          \tuple{
            \tuple{\opn[b]{pr}_{1}v,\opn[b]{0}},
            \opn[b]{pr}_{2}v
          }
        }
      }
    } .
  \end{aligned}
\end{equation}
On the right-hand side, \(q\) and \(e\) are weakened from
\(\Gamma,y:\tau\) to \(\Gamma,x:\sigma,y:\tau\), and the inserted zero has
type \(\lintype{\sigma}\).

Side note: the oplax comparison term for cartesian reindexing is definable.  It is
defined by the pure-continuation special case of \( \opn[b]{bind}^{w} \).  If
\[
  \Gamma\vdash p:\Fun[nn]{\opn[b]{M}}{\sigma},
  \qquad
  \Gamma,x:\sigma\vdash h:\tau,
  \qquad
  \Gamma,y:\tau\vdash \lintype{\tau}\ \mathrm{type},
\]
then
\begin{equation}
  \label{eq:appendix-target-language-monad-w-alpha}
  \begin{aligned}
    &
    \Gamma;
    v:
    \opn[b]{M}^{w}_{
      \fn[n]{\opn[b]{bind}}{
        p;
        x.\fn[n]{\opn[b]{return}}{h}
      }
      :
      \Fun[nn]{\opn[b]{M}}{\tau}
    }
    \left[y\mapsto\lintype{\tau}\right]
    \vdash
    \alpha_{p,x.h,\lintype{\tau}}(v)
    :
    \opn[b]{M}^{w}_{p:\Fun[nn]{\opn[b]{M}}{\sigma}}
    \left[
      x\mapsto\casubst{\lintype{\tau}}{y}{h}
    \right],
    \\
    &
    \alpha_{p,x.h,\lintype{\tau}}(v)
    :=
    \fn[n]{
      \opn[b]{bind}^{w}_{
        p;
        x.\fn[n]{\opn[b]{return}}{h}
      }
    }{
      v;
      \tuple{
        \opn[b]{0},
        \fn[n]{\opn[b]{return}^{w}_{h}}{v}
      }
    }.
  \end{aligned}
\end{equation}
The identity and composition equations for these comparison maps follow from
\eqref{eq:appendix-target-language-monad-w-right-unit} and
\eqref{eq:appendix-target-language-monad-w-associativity}, respectively.
Thus the displayed equations present an oplax fibred monad over the source
monad.  Since they are schematic in an arbitrary cartesian context, the usual
contextual strength is inherited from the given strength of the source monad.

The handler \(w\)-unit law is the transpose of
\[
  \fn[n]{\opn[b]{E}_{\nu}}{
    \fn[n]{\opn[b]{return}}{a}
  }
  \equiv
  a.
\]
Suppose
\[
  \Gamma\vdash a:\nu,
  \qquad
  \Gamma\vdash \lintype{\rho}\ \mathrm{type},
  \qquad
  \Gamma;v:\lintype{\tau}_{\nu}\vdash r:\lintype{\rho}.
\]
Then
\begin{equation}
  \label{eq:appendix-target-language-handler-w-unit}
  \fn[n]{
    \opn[b]{E}^{w}_{
      \nu,
      \fn[n]{\opn[b]{return}}{a}
    }
  }{
    \fn[n]{\opn[b]{return}^{w}_{a}}{r}
  }
  \equiv
  r .
\end{equation}

The handler \(w\)-multiplication law is the transpose of
\[
  \fn[n]{\opn[b]{E}_{\nu}}{
    \fn[n]{\opn[b]{bind}}{t;x.u}
  }
  \equiv
  \fn[n]{\opn[b]{E}_{\nu}}{
    \fn[n]{\opn[b]{bind}}{
      t;
      x.\fn[n]{\opn[b]{return}}{
        \fn[n]{\opn[b]{E}_{\nu}}{u}
      }
    }
  }.
\]
Suppose
\[
  \Gamma\vdash t:\Fun[nn]{\opn[b]{M}}{\sigma},
  \qquad
  \Gamma,x:\sigma\vdash u:\Fun[nn]{\opn[b]{M}}{\nu},
\]
\[
  \Gamma,x:\sigma\vdash \lintype{\sigma}\ \mathrm{type},
  \qquad
  \Gamma\vdash \lintype{\rho}\ \mathrm{type},
\]
and
\[
  \Gamma;
  v:
  \opn[b]{M}^{w}_{t:\Fun[nn]{\opn[b]{M}}{\sigma}}
  \left[x\mapsto\lintype{\sigma}\right]
  \vdash r:\lintype{\rho},
\]
\[
  \Gamma,x:\sigma;
  v:
  \opn[b]{M}^{w}_{u:\Fun[nn]{\opn[b]{M}}{\nu}}
  \left[y\mapsto\lintype{\tau}_{\nu}\right]
  \vdash s:\lintype{\rho}\times\lintype{\sigma}.
\]
Then
\begin{equation}
  \label{eq:appendix-target-language-handler-w-multiplication}
  \begin{aligned}
    &
    \fn[n]{
      \opn[b]{E}^{w}_{
        \nu,
        \fn[n]{\opn[b]{bind}}{t;x.u}
      }
    }{
      \fn[n]{\opn[b]{bind}^{w}_{t;x.u}}{
        r;
        s
      }
    }
    \\
    &\qquad\equiv
    \fn[n]{
      \opn[b]{E}^{w}_{
        \nu,
        \fn[n]{\opn[b]{bind}}{
          t;
          x.\fn[n]{\opn[b]{return}}{
            \fn[n]{\opn[b]{E}_{\nu}}{u}
          }
        }
      }
    }{
      \fn[n]{
        \opn[b]{bind}^{w}_{
          t;
          x.\fn[n]{\opn[b]{return}}{
            \fn[n]{\opn[b]{E}_{\nu}}{u}
          }
        }
      }{
        r;
        \fn[n]{
          \opn[b]{return}^{w}_{
            \fn[n]{\opn[b]{E}_{\nu}}{u}
          }
        }{
          \fn[n]{\opn[b]{E}^{w}_{\nu,u}}{s}
        }
      }
    } .
  \end{aligned}
\end{equation}

\subsection{Syntactic category of the target language}
\label{sec:appendix-syntactic-category-of-the-target-language}

With the proper type theory spelled out, we may now provide the proper definition of \( \LSyn[\tau] \) for each cartesian type \( \tau \):
\begin{itemize}
\item \emph{Objects}: linear types \( \lintype{\sigma} \) well-formed in the context \( p: \tau \).
\item \emph{Morphisms}: a morphism \( \lintype{\rho} \to \lintype{\sigma} \) is a linear term \( p: \tau; v : \lintype{\rho} \vdash s : \lintype{\sigma} \), considered modulo \( \alpha \)-equivalence. \emph{Composition} is performed by linear let bindings, and \emph{identities} are given by the linear variable \( v \) itself (see the \emph{Linear variable} typing rule).
\end{itemize}

Given a term \( p: \tau \vdash t : \sigma \) in \( \CSyn \), write
\(q:\sigma\) for the distinguished variable in the fibre over \(\sigma\).
Its reindexing \( t^{\ast}: \LSyn[\sigma] \to \LSyn[\tau] \) acts by substitution on objects and terms:
\begin{equation*}
  \begin{array}{rrl}
    \left( \text{Objects} \right) 
    &\quad
      \Fun[nn]{ t^{\ast} }{ \lintype{\sigma} }
    &:=
      \casubst{\lintype{\sigma}}{q}{t}
    \\
    \left( \text{Morphisms} \right) 
    &\quad
      \Fun[nn]{ t^{\ast} }{
      \left( q : \sigma; v : \lintype{\rho} \vdash s : \lintype{\sigma} \right)
      }
    &:=
      p : \tau; v : \casubst{\lintype{\rho}}{q}{t}
      \vdash
      \casubst{s}{q}{t} : \casubst{\lintype{\sigma}}{q}{t}
  \end{array}
\end{equation*}

\begin{theorem}[Split fibration structure]
  The projection \( \pi: \Tgt \to \CSyn \) is a split fibration with chosen cartesian lifts given by reindexing \( t^{\ast} \) on the nose. Identities and composition for \( \LSyn \) are strictly preserved by reindexing.
\end{theorem}

\section{\texorpdfstring{\( \Fam{\Vect[o]} \)}{Fam(VectOp)} denotational semantics}
\label{sec:appendix-fam-vectop-denotational-semantics}

In Section~\ref{sec:interpretation-of-the-target-language-in-fam-vectop}, we discuss briefly the functor \( \fvoint{\cdot}: \Tgt \to \Fam{\Vect[o]} \), which interprets the target language in the category \( \Fam{\Vect[o]} \). The goal of this Appendix is to describe this functor in further detail.

Thus cartesian types are interpreted as sets (the base part of the objects in \( \Fam{\Vect[o]} \)), whereas linear judgements over cartesian contexts are interpreted in \( \Vect[o] \). For convenience, given a morphism \( f: U \to V \) in \( \Vect[o] \), we denote by \( \opmor{f}: V \to U \) the corresponding linear map in \( \Vect \). This lets us display fibre maps in the ordinary \( \Vect \) orientation.

\subsection{Type semantics}

The semantics of the \( \CSyn \) fragment establishes the base cartesian semantics, interpreting cartesian contexts as the base sets of objects of \( \Fam{\Vect[o]} \). The term constructors shared with the source language have the same interpretation as in the source-language semantics in \( \Set \) (Sections~\ref{sec:interpretation-of-source-language-in-set} and \ref{sec:appendix-source-language-denotational-semantics}),
\begin{equation*}
  \begin{aligned}
    \setint{0} &= \emptyset
    \\
    \setint{1} &= \left\{ \ast \right\}
    \\
    \setint{\srcreal[n]} &= \R[n]
  \end{aligned}
  \qquad
  \begin{aligned}
    \setint{\tau \times \sigma} &= \setint{\tau} \times \setint{\sigma}
    \\
    \setint{\tau \sqcup \sigma} &= \setint{\tau} \amalg \setint{\sigma}
    \\
    \setint{\opn[b]{M} \tau} &= \Fun[nn]{ M }{ \setint{\tau} }
  \end{aligned}
\end{equation*}
with deterministic primitives interpreted as their chosen total functions (which we assume to be differentiable), and \( \tuple{ \opn[b]{M}, \opn[b]{return}, \opn[b]{bind} } \) as \( \tuple{ M, \eta, \mu } \). For a cartesian judgement under the context \( p : \tau \), dependent sums are interpreted as follows: for every \( \gamma \in \setint{\tau} \),
\begin{equation}
  \label{eq:appendix-target-language-csyn-dependent-sum-interpretation}
  \fn[n]{
  \setint{\Sigma x : \sigma \; . \; \tau'}
  }{ \gamma }
  := \coprod_{a \in \fn[n]{\setint{\sigma}}{\gamma}}
  \fn[n]{ \setint{\tau'} }{ \gamma, a }
\end{equation}
and the linear function types are interpreted as the set of linear morphisms between the linear interpretations of the respective linear types,
\begin{equation}
  \label{eq:appendix-target-language-csyn-linear-function-interpretation}
  \fn[n]{ \setint{ \lintype{\sigma} \linto \lintype{\rho} } }{ \gamma } :=
  \fn[n]{
    \Vect
  }{
    \fn[n]{ \linint{\lintype{\sigma}}{p : \tau} }{ \gamma },
    \fn[n]{ \linint{\lintype{\rho}}{p : \tau} }{ \gamma }
  }
\end{equation}
where the linear interpretation \( \linint{\cdot}{p : \tau} \) of linear types dependent on a cartesian variable \( p : \tau \) provides a real vector space interpretation for those linear types, making them into (non-topological) vector bundles over the base set \( \setint{\tau} \). As such, for every linear type \( \lintype{\sigma} \) well-formed in the context \( p: \tau \), we get a functor (i.e. a family)
\begin{equation*}
  \linint{\lintype{\sigma}}{p : \tau}:
  \setint{\tau} \to \Vect[o]
\end{equation*}
The linear interpretation is also described inductively over the structure of linear types, for every \( \gamma \in \setint{\tau} \),
\begin{equation}
  \label{eq:appendix-target-language-interpretation-inductive-definition}
  \begin{aligned}
    \Fun[np]{
    \linint{\lintype{\srcreal[k]}}{p : \tau}
    }{ \gamma }
    &= \R[k]
    \\
    \Fun[np]{
    \linint{\lintype{1}}{p : \tau}
    }{ \gamma }
    &= 0
    \\
    \Fun[np]{
    \linint{\lintype{\sigma} \times \lintype{\rho}}{p : \tau}
    }{ \gamma }
    &= 
      \Fun[np]{
      \linint{\lintype{\sigma}}{p : \tau}
      }{ \gamma }
      \oplus
      \Fun[np]{
      \linint{\lintype{\rho}}{p : \tau}
      }{ \gamma }
    \\
    \Fun[np]{
    \linint{
    \casenst{t}{i}{n}{x_{i}}{\lintype{\sigma}_{i}}
    }{p : \tau}
    }{ \gamma }
    &=
      \begin{cases}
        \Fun[np]{
        \linint{\lintype{\sigma}_{i}}{p : \tau, x_{i} : \tau_{i}}
        }{ \gamma, a }
        &\text{if }
          \fn[n]{ \setint{t} }{ \gamma } = \opn[b]{in}_{i} a
      \end{cases}
    \\
    \Fun[np]{
    \linint{
    \opn[b]{M}^{w}_{q : \Fun[nn]{ \opn[b]{M} }{ \rho }}
    \left[ x \mapsto \lintype{\sigma} \right]
    }{p : \tau}
    }{ \gamma }
    &=
      \bigoplus_{x \in \supp \fn[n]{ \setint{q} }{\gamma} }
      \left(
      \fn[n]{ \linint{\lintype{\sigma}}{p : \tau, x : \rho} }{ \gamma, x } \oplus \R
      \right)
  \end{aligned}
\end{equation}
Compare the definition of the monadic cotangent type semantics to the definition of the fibre part of the monad on objects in \( \Fam{\Vect[o]} \) given in Equation~\ref{eq:fad-monad-on-objects-in-fam-vectop}.

With these definitions, we get a functor \( \fvoint{\cdot}: \Tgt \to \Fam{\Vect[o]} \) defined on objects, which, as described in Section~\ref{sec:appendix-syntactic-category-of-the-target-language}, are given by a pair \( \tuple{ p : \tau, \lintype{\sigma} } \) where \( p : \tau \) is a cartesian context and \( \lintype{\sigma} \) a well-formed linear type under this context, as
\begin{equation}
  \label{eq:appendix-target-language-interpretation-functor-on-objects}
  \fvoint{\tuple{ p : \tau, \lintype{\sigma} }} :=
  \tuple{
    \setint{\tau},
    \linint{\lintype{\sigma}}{p : \tau}
  }
\end{equation}
and on any morphism \( \tuple{ t, s }: \tuple{ p : \tau, \lintype{\rho} } \to \tuple{ q : \sigma, \lintype{\chi} } \), which consists of a cartesian arrow \( t : p \to q \) and a linear arrow in the fibre \( s: \casubst{\lintype{\chi}}{q}{t} \to \lintype{\rho} \) in \( \LSyn[ p : \tau ] \), as
\begin{equation}
  \label{eq:appendix-target-language-fam-vectop-interpretation-morphisms}
  \fvoint{\tuple{ t, s }} :=
  \tuple{
    \setint{t},
    \left(
      \gamma \mapsto
      \left(
        \opmor[p]{
          \fn[n]{ \linint{s}{p : \tau} }{ \gamma }
        }
        \in
        \fn[n]{ \Vect }{
          \fn[n]{ \linint{\lintype{\chi}}{q : \sigma} }{ \fn[n]{ \setint{t} }{ \gamma } },
          \fn[n]{ \linint{\lintype{\rho}}{p : \tau} }{ \gamma }
        }
      \right) 
    \right) 
  }
\end{equation}

\subsection{Term semantics}

Next, we provide interpretations for the terms themselves, which will be mapped to morphisms in \( \Fam{\Vect[o]} \). Again, this is done inductively over the structure of the terms, and the cartesian part of the language is interpreted in \( \Set \) as the base part of the morphisms.

\paragraph{Cartesian fragment}
For the term constructors in common with the source language, the interpretation is the same as the one described in Section~\ref{sec:appendix-source-language-denotational-semantics}. For the other cartesian term constructors, in a cartesian context \(\Gamma=p:\tau\), we define, for all \( \gamma \in \setint{\tau} \),
\begin{equation}
  \label{eq:appendix-target-language-new-cartesian-term-constructors-semantics}
  \begin{aligned}
    \fn[n]{
    \setint{
    \Gamma \vdash
    \tuple{ a, t }: \Sigma \left( x : \sigma \right) . \tau'
    }
    }{ \gamma }
    &=
      \tuple{
      \fn[n]{ \setint{a} }{ \gamma },
      \fn[n]{ \setint{t} }{ \gamma }
      }
    \\
    \fn[n]{
    \setint{
    \Gamma \vdash
    \opn[b]{case} \; p \; \opn[b]{of} \; \tuple{x,y} \to s : \rho
    }
    }{ \gamma }
    &=
      \fn[n]{ \setint{s} }{ \gamma, a, b }
      \qq{where}
      \tuple{ a, b } = \fn[n]{ \setint{p} }{ \gamma }
    \\
    \fn[n]{
    \setint{
    \Gamma \vdash
    \lambda v \; . \; s:
    \lintype{\sigma} \linto \lintype{\rho}
    }
    }{ \gamma }
    &=
      \opmor[p]{ \fn[n]{ \linint{s}{p : \tau} }{ \gamma } }
  \end{aligned}
\end{equation}
Linear application is a linear term constructor; its semantics is given below.

\paragraph{Linear core}
Linear terms are interpreted fibrewise in \( \Vect[o] \), or, equivalently, as linear maps in \( \Vect \), but with domain and codomain reversed. The first display below is written in the \(\Vect[o]\)-orientation; when we need the ordinary \(\Vect\)-orientation, we apply \(\opmor{\blankdash}\). In what follows, we consider the linear terms under a cartesian context \( \Gamma = p : \tau \),
{
  \allowdisplaybreaks
  \begin{alignat*}{2}
    \fn[n]{
    \linint{
    \Gamma; v : \lintype{\sigma}
    \vdash v : \lintype{\sigma}
    }{p : \tau}
    }{ \gamma }
    &=
      \id_{
      \Fun[np]{
      \linint{
      \lintype{\sigma}
      }{p : \tau}
      }{ \gamma }
      }
    \\
    \fn[n]{
    \linint{
    \opn[b]{0}
    }{p : \tau}
    }{ \gamma }
    &= 0
    \\
    \fn[n]{
    \linint{
    s + t
    }{p : \tau}
    }{ \gamma }
    &=
      \fn[n]{
      \linint{
      s
      }{p : \tau}
      }{ \gamma }
      +
      \fn[n]{
      \linint{
      t
      }{p : \tau}
      }{ \gamma }
    \\
    \fn[n]{
    \linint{
    \tuple{ s, t }
    }{p : \tau}
    }{ \gamma }
    &=
      \cotupling{
      \fn[n]{
      \linint{
      s
      }{p : \tau}
      }{ \gamma },
      \fn[n]{
      \linint{
      t
      }{p : \tau}
      }{ \gamma }
      }
    \\
    \fn[n]{
    \linint{
    \opn[b]{pr}_{i} u
    }{p : \tau}
    }{ \gamma }
    &=
      \fn[n]{
      \linint{
      u
      }{p : \tau}
      }{ \gamma }
      \circ
      \iota_{i}
      &&
      \quad
      \left( i \in \left\{ 1, 2 \right\} \right) 
    \\
    \fn[n]{
    \linint{
    \letst{v}{t}{s}
    }{p : \tau}
    }{ \gamma }
    &=
      \fn[n]{ \linint{t}{p : \tau} }{ \gamma }
      \circ
      \fn[n]{ \linint{s}{p : \tau} }{ \gamma }
      &&
      \quad
      \text{(Composition in \( \Vect[o] \))}
    \\
    \fn[n]{
    \linint{
    \casenst{t}{i}{n}{x_{i}}{r_{i}}
    }{p : \tau}
    }{ \gamma }
    &=
      \fn[n]{ \linint{r_{i}}{p : \tau, x_{i} : \tau_{i}} }{ \gamma, a }
      &&
      \quad
      \text{if }
         \fn[n]{ \setint{t} }{ \gamma } = \fn[n]{ \opn[b]{in}_{i} }{ a }
  \end{alignat*}
}
For linear application, where \( \Gamma \vdash r : \lintype{\sigma} \linto \lintype{\rho} \) and \( \Gamma; v : \lintype{\alpha} \vdash t : \lintype{\sigma} \), we use the ordinary \(\Vect\)-orientation
\begin{equation*}
  \opmor[p]{
    \fn[n]{
      \linint{\linapp{r}{t}}{p : \tau}
    }{ \gamma }
  }
  =
  \fn[n]{ \setint{r} }{ \gamma }
  \circ
  \opmor[p]{
    \fn[n]{ \linint{t}{p : \tau} }{ \gamma }
  }.
\end{equation*}
Finally, for any built-in operation \( \opn[b]{op}: \srcreal[n_{1}] \times \cdots \times \srcreal[n_{k}] \to \srcreal[m] \),
\begin{equation*}
    \opmor[p]{
    \fn[n]{
    \linint{
    \fn[n]{ \fn[n]{ \opn[b]{Dop} }{ t_{1}, \cdots, t_{k} } }{v}
    }{p : \tau}
    }{ \gamma }
    }
    =
    \fn[p]{ \grad \setint{\opn[b]{op}} }{
      \fn[n]{ \setint{t_{1}} }{\gamma},
      \cdots,
      \fn[n]{ \setint{t_{k}} }{\gamma}
    }:
    \R[m] \to \R[n_{1}] \oplus \cdots \oplus \R[n_{k}]
\end{equation*}
where \( \setint{\opn[b]{op}} \) is the intended set-based interpretation of the operation \( \opn[b]{op} \), assumed to be differentiable, so that taking its transpose derivative \( \grad \setint{\opn[b]{op}} \) is well-defined.

\paragraph{Transposed derivative of \( \opn[b]{return} \)}
Let the cartesian context be \(\Gamma=p:\chi\).  Suppose
\(p:\chi\vdash t:\tau\),
\(p:\chi,x:\tau\vdash \lintype{\sigma}\text{ type}\), and
\(p:\chi; v:\casubst{\lintype{\sigma}}{x}{t}\vdash s:\lintype{\rho}\).
For \(\gamma\in\setint{\chi}\), put
\[
  a:=\fn[n]{\setint{t}}{\gamma},
  \qquad
  A_{a}:=\fn[n]{\linint{\lintype{\sigma}}{p:\chi,x:\tau}}{\gamma,a},
  \qquad
  R_{\gamma}:=\fn[n]{\linint{\lintype{\rho}}{p:\chi}}{\gamma}.
\]
Then the ordinary \(\Vect\)-orientation of the fibre map is
\begin{equation}
  \label{eq:appendix-target-language-transpose-derivative-of-return}
  \opmor[p]{
    \fn[n]{
      \linint{
        \fn[n]{ \opn[b]{return}^{w}_{t} }{ s }
      }{p : \chi}
    }{ \gamma }
  }
  =
  \opmor[p]{
    \fn[n]{ \linint{s}{p : \chi} }{ \gamma }
  }
  \circ
  \pi_{A_{a}}
  :
  A_{a}\oplus\R
  \to
  R_{\gamma}.
\end{equation}
Here \(\pi_{A_{a}}:A_{a}\oplus\R\to A_{a}\) is the projection onto the returned-value cotangent.  The log-weight component is discarded, as expected for the derivative of a unit distribution with constant log-weight.

\paragraph{Transposed derivative of \( \opn[b]{bind} \)}
Let the cartesian context again be \(\Gamma=p:\chi\), and consider a well-typed term
\(
  \fn[n]{ \opn[b]{bind}^{w}_{s; x \; . \; t}}{ r ; u }
\)
with
\(p:\chi\vdash s:\Fun[nn]{\opn[b]{M}}{\sigma}\) and
\(p:\chi,x:\sigma\vdash t:\Fun[nn]{\opn[b]{M}}{\tau}\).
Fix \(\gamma\in\setint{\chi}\), and define
\begin{equation*}
  \begin{aligned}
    w
    &:= \fn[n]{\setint{s}}{\gamma}
      &&\in \Fun[nn]{M}{\setint{\sigma}},
    \\
    \mu_{x}
    &:= \fn[n]{\setint{t}}{\gamma,x}
      &&\in \Fun[nn]{M}{\setint{\tau}},
    \\
    f_{\gamma}(x)
    &:= \mu_{x}
      &&:\setint{\sigma}\to\Fun[nn]{M}{\setint{\tau}},
    \\
    \widehat w
    &:= \fn[n]{\Fun[nn]{M}{f_{\gamma}}}{w}
      &&\in \Fun[nn]{M}{\Fun[nn]{M}{\setint{\tau}}},
    \\
    \bar w
    &:= \fn[n]{\Nat[nn]{\mu}{\setint{\tau}}}{\widehat w}
      &&\in \Fun[nn]{M}{\setint{\tau}}.
  \end{aligned}
\end{equation*}
Thus \(\bar w\) is the set-level denotation of
\(\fn[n]{\opn[b]{bind}}{s;x.t}\) at \(\gamma\).  Put
\begin{equation*}
  A_{x}:=\fn[n]{\linint{\lintype{\sigma}}{p:\chi,x:\sigma}}{\gamma,x},
  \qquad
  B_{y}:=\fn[n]{\linint{\lintype{\tau}}{p:\chi,y:\tau}}{\gamma,y},
  \qquad
  R_{\gamma}:=\fn[n]{\linint{\lintype{\rho}}{p:\chi}}{\gamma}.
\end{equation*}
The two linear arguments give ordinary-orientation maps
\begin{equation*}
  \begin{aligned}
    r^{\sharp}_{\gamma}
    &:={}
      \opmor[p]{
        \fn[n]{\linint{r}{p:\chi}}{\gamma}
      }
      :
      \bigoplus_{x\in\supp w}
      \left(A_{x}\oplus\R\right)
      \to R_{\gamma},
    \\
    u^{\sharp}_{\gamma,x}
    &:={}
      \opmor[p]{
        \fn[n]{\linint{u}{p:\chi,x:\sigma}}{\gamma,x}
      }
      :
      \bigoplus_{y\in\supp \mu_{x}}
      \left(B_{y}\oplus\R\right)
      \to R_{\gamma}\oplus A_{x}.
  \end{aligned}
\end{equation*}
The family \(\{u^{\sharp}_{\gamma,x}\}_{x\in\setint{\sigma}}\) is the fibre part of a morphism in \(\Fam{\Vect[o]}\)
\begin{equation*}
  \widehat f_{\gamma}:
  \tuple{
    \setint{\sigma},
    x\mapsto R_{\gamma}\oplus A_{x}
  }
  \to
  \Fun[nn]{M^{\cat[]{F}}}{
    \tuple{
      \setint{\tau},
      y\mapsto B_{y}
    }
  },
\end{equation*}
whose base map is \(f_{\gamma}\).  Applying the lifted FAD monad to this morphism and evaluating the fibre map at \(w\) gives
\begin{equation*}
  \Theta^{u}_{w}
  :=
  \fn[n]{
    \Nat[nn]{
      \fibr{\Fun[nn]{M^{\cat[]{F}}}{\widehat f_{\gamma}}}
    }{w}
  }:
  \bigoplus_{\xi\in\supp \widehat w}
  \left(
    \left(
      \bigoplus_{y\in\supp \xi}
      \left(B_{y}\oplus\R\right)
    \right)
    \oplus\R
  \right)
  \to
  \bigoplus_{x\in\supp w}
  \left(
    \left(R_{\gamma}\oplus A_{x}\right)
    \oplus\R
  \right).
\end{equation*}
Similarly, let
\begin{equation*}
  \varpi^{B}_{\widehat w}
  :=
  \fn[n]{
    \Nat[nn]{
      \fibr{
        \Nat[nn]{\mu^{\cat[]{F}}}{
          \tuple{
            \setint{\tau},
            y\mapsto B_{y}
          }
        }
      }
    }{\widehat w}
  }:
  \bigoplus_{y\in\supp \bar w}
  \left(B_{y}\oplus\R\right)
  \to
  \bigoplus_{\xi\in\supp \widehat w}
  \left(
    \left(
      \bigoplus_{y\in\supp \xi}
      \left(B_{y}\oplus\R\right)
    \right)
    \oplus\R
  \right)
\end{equation*}
be the ordinary-orientation fibre map of the lifted monad multiplication at \(\widehat w\).
For \(C_{x}:=R_{\gamma}\oplus A_{x}\), define
\begin{equation*}
  \begin{aligned}
    \Omega_{\gamma}
    &: 
    \bigoplus_{x\in\supp w}\left(C_{x}\oplus\R\right)
    \to R_{\gamma},
    \\
    \Omega_{\gamma}
    &:={}
    \left(
      \nabla_{R_{\gamma}}
      \circ
      \bigoplus_{x\in\supp w}
      \left(\pi_{R_{\gamma}}\circ\pi_{1}\right)
    \right)
    +
    \left(
      r^{\sharp}_{\gamma}
      \circ
      \bigoplus_{x\in\supp w}
      \tupling{
        \pi_{A_{x}}\circ\pi_{1},
        \pi_{2}
      }
    \right).
  \end{aligned}
\end{equation*}
Here \(\pi_{1}:C_{x}\oplus\R\to C_{x}\), \(\pi_{2}:C_{x}\oplus\R\to\R\), and \(\pi_{R_{\gamma}}\), \(\pi_{A_x}\) are the projections from \(C_x=R_\gamma\oplus A_x\).  The first summand of \(\Omega_\gamma\) adds the direct ambient cotangents produced by the continuation backpropagators; the second feeds the cotangent of the bound value and the outer log-weight cotangent through \(r\).

The semantics of \(\opn[b]{bind}^{w}\) is then
\begin{equation}
  \label{eq:appendix-target-language-transpose-derivative-of-bind}
  \opmor[p]{
    \fn[n]{
      \linint{
        \fn[n]{ \opn[b]{bind}^{w}_{s; x \; . \; t}}{ r ; u }
      }{p : \chi}
    }{ \gamma }
  }
  =
  \Omega_{\gamma}
  \circ
  \Theta^{u}_{w}
  \circ
  \varpi^{B}_{\widehat w}
  :
  \bigoplus_{y\in\supp \bar w}\left(B_{y}\oplus\R\right)
  \to R_{\gamma}.
\end{equation}
Because \(\Theta^{u}_{w}\) and \(\varpi^{B}_{\widehat w}\) are the actual fibre maps of the lifted FAD functor and multiplication, this definition also covers collisions of atoms and collisions among the intermediate FADs \(\mu_x\).

\paragraph{Transposed derivative of \( \opn[b]{categorical} \)}
Again work in a cartesian context \(p:\chi\).  Consider a well-typed term
\[
  \fn[n]{
    \opn[b]{categorical}^{w}_{
      \tuple{\tuple{t_{i},w_{i}}}_{i=1}^{n}
    }
  }{
    \tuple{\tuple{t_{i}',w_{i}'}}_{i=1}^{n}
  }
\]
generated by the rule in the target language.  Fix \(\gamma\in\setint{\chi}\), and put
\[
  a_{i}:=\fn[n]{\setint{t_{i}}}{\gamma},
  \qquad
  \ell_{i}:=\fn[n]{\setint{w_{i}}}{\gamma},
  \qquad
  q:=
  \fn[n]{\setint{
    \fn[n]{\opn[b]{categorical}}{
      \tuple{\tuple{t_{i},w_{i}}}_{i=1}^{n}
    }
  }}{\gamma}.
\]
Let
\[
  A_{x}:=\fn[n]{\linint{\lintype{\sigma}}{p:\chi,x:\tau}}{\gamma,x},
  \qquad
  R_{\gamma}:=\fn[n]{\linint{\lintype{\rho}}{p:\chi}}{\gamma}.
\]
The lifted categorical operation \(\categ^{\cat[]{F}}\) has an ordinary-orientation fibre map at the list \(\tuple{\tuple{a_{i},\ell_{i}}}_{i=1}^{n}\), which we write as
\begin{equation*}
  \kappa_{\gamma}:
  \bigoplus_{x\in\supp q}\left(A_{x}\oplus\R\right)
  \to
  \bigoplus_{i=1}^{n}\left(A_{a_i}\oplus\R\right).
\end{equation*}
This is the fibre part of the left factor in the FAD factorisation; in particular, it performs the correct splitting when several displayed atoms denote the same value.  Define
\begin{equation*}
  \begin{aligned}
    t^{\sharp}_{i,\gamma}
    &:={}
    \opmor[p]{
      \fn[n]{\linint{t_{i}'}{p:\chi}}{\gamma}
    }
    : A_{a_i}\to R_{\gamma},
    \\
    w^{\sharp}_{i,\gamma}
    &:={}
    \opmor[p]{
      \fn[n]{\linint{w_{i}'}{p:\chi}}{\gamma}
    }
    : \R\to R_{\gamma}.
  \end{aligned}
\end{equation*}
Then
\begin{equation}
  \label{eq:appendix-target-language-transpose-derivative-of-categorical}
  \begin{aligned}
  &\opmor[p]{
    \fn[n]{
      \linint{
        \fn[n]{
          \opn[b]{categorical}^{w}_{
            \tuple{\tuple{t_{i},w_{i}}}_{i=1}^{n}
          }
        }{
          \tuple{\tuple{t_{i}',w_{i}'}}_{i=1}^{n}
        }
      }{p:\chi}
    }{\gamma}
  }
  \\
  &\qquad=
  \nabla_{R_{\gamma}}
  \circ
  \left(
    \bigoplus_{i=1}^{n}
    \left(
      t^{\sharp}_{i,\gamma}\circ\pi_{A_{a_i}}
      +
      w^{\sharp}_{i,\gamma}\circ\pi_{\R}
    \right)
  \right)
  \circ
  \kappa_{\gamma}.
  \end{aligned}
\end{equation}

\paragraph{Transposed derivative of the handler algebra}
Let \(\nu=\srcreal[d]\), and write \(V_{\nu}:=\R[d]\) for the fibre of \(\lintype{\tau}_{\nu}\).  Suppose \(p:\chi\vdash t:\Fun[nn]{\opn[b]{M}}{\nu}\), and fix \(\gamma\in\setint{\chi}\).  Put \(m:=\fn[n]{\setint{t}}{\gamma}\).  The lifted handler algebra \(\opn[b]{E}^{\cat[]{F}}_{\nu}:\Fun[nn]{M^{\cat[]{F}}}{\TR[d]}\to\TR[d]\) has ordinary-orientation fibre map
\begin{equation*}
  \epsilon_{m}:V_{\nu}\to
  \bigoplus_{y\in\supp m}\left(V_{\nu}\oplus\R\right),
\end{equation*}
which, for the expectation handler, is explicitly
\begin{equation*}
  \fn[n]{\epsilon_{m}}{c}
  =
  \sum_{y\in\supp m}
  \left(
    \fn[n]{m}{y}\,c\cdot\vect{e}_{y}
    +
    \fn[n]{m}{y}\,\left\langle c,y\right\rangle\cdot\vect{\varepsilon}_{y}
  \right),
  \qquad c\in V_{\nu}.
\end{equation*}
Here \(\vect{e}_{y}\) denotes the atom-cotangent summand and \(\vect{\varepsilon}_{y}\) denotes the log-weight summand.  Therefore, for
\(p:\chi;v:\opn[b]{M}^{w}_{t}\left[y\mapsto\lintype{\tau}_{\nu}\right]\vdash s:\lintype{\rho}\),
\begin{equation}
  \label{eq:appendix-target-language-transpose-derivative-of-handler}
  \opmor[p]{
    \fn[n]{
      \linint{
        \fn[n]{\opn[b]{E}^{w}_{\nu,t}}{s}
      }{p:\chi}
    }{\gamma}
  }
  =
  \opmor[p]{
    \fn[n]{\linint{s}{p:\chi}}{\gamma}
  }
  \circ
  \epsilon_{m}
  : V_{\nu}\to
  \fn[n]{\linint{\lintype{\rho}}{p:\chi}}{\gamma}.
\end{equation}

\section{Code transformation rules}
\label{sec:appendix-code-transformation-rules}

This section gives the full reverse-mode CHAD translation used in
Equation~\eqref{eq:chad-code-transformation-terms-under-context}.  For every source judgement
\( \Gamma \vdash t : \tau \), the translation is a target-language term
\begin{equation}
  \label{eq:appendix-chad-term-typing-shape}
  \Fun[np]{\chad}{\Gamma}_{1}
  \vdash
  \Fun[np]{\chad[\Gamma]}{t}
  :
  \Sigma \left( p : \Fun[np]{\chad}{\tau}_{1} \right)
  .
  \left(
    \casubst{\Fun[np]{\chad}{\tau}_{2}}{p_{\tau}}{p}
    \linto
    \Fun[np]{\chad}{\Gamma}_{2}
  \right).
\end{equation}
The first component is the primal computation and the second component is the
backpropagator.  To avoid capture in dependent cotangent types, we write
\(p_{\tau} : \Fun[np]{\chad}{\tau}_{1}\) for the distinguished primal
variable of the type \(\tau\).  Thus \(\Fun[np]{\chad}{\tau}_{2}\) is read
in the cartesian context
\(p_{\tau} : \Fun[np]{\chad}{\tau}_{1}\), and every substitution into a
cotangent type below is understood after this \(\alpha\)-renaming.

\paragraph{Convention for the distinguished linear variable.}
The target language has one distinguished linear variable.  In clauses involving
\(\opn[b]{return}^{w}\), \(\opn[b]{bind}^{w}\),
\(\opn[b]{categorical}^{w}\), or \(\opn[b]{E}^{w}\), the variable bound by the
outer expression \(\lambda v\; .\;(-)\) has the conclusion type of the
corresponding \(w\)-typing rule.  The occurrences of \(v\) inside the arguments
of the \(w\)-constructor are checked at the premise types of that rule.  This is
exactly the typing convention used in the rules of
Section~\ref{sec:appendix-target-language-for-chad}.

\subsection{Type and context translation}

\paragraph{Base types, products, and sums.}
The primal translation \(\Fun[np]{\chad}{-}_{1}\) is cartesian, while
\(\Fun[np]{\chad}{-}_{2}\) gives a linear cotangent type indexed by the
corresponding distinguished primal variable:
\begin{equation}
  \label{eq:appendix-chad-type-translation-core-basis}
  \begin{array}{rcl}
    \Fun[np]{\chad}{1}_{1} &:=& 1
    \\[2pt]
    \Fun[np]{\chad}{0}_{1} &:=& 0
    \\[2pt]
    \Fun[np]{\chad}{\srcreal[k]}_{1} &:=& \srcreal[k]
    \\[6pt]
    \Fun[np]{\chad}{\tau \times \sigma}_{1}
    &:=&
    \Fun[np]{\chad}{\tau}_{1} \times \Fun[np]{\chad}{\sigma}_{1}
    \\[6pt]
    \Fun[np]{\chad}{\tau \sqcup \sigma}_{1}
    &:=&
    \Fun[np]{\chad}{\tau}_{1} \sqcup \Fun[np]{\chad}{\sigma}_{1}
  \end{array}
\end{equation}
and
\begin{equation}
  \label{eq:appendix-chad-type-translation-core-cotangents}
  \begin{array}{rcl}
    \Fun[np]{\chad}{1}_{2} &:=& \lintype{1}
    \\[2pt]
    \Fun[np]{\chad}{0}_{2} &:=& \lintype{1}
    \\[2pt]
    \Fun[np]{\chad}{\srcreal[k]}_{2} &:=& \lintype{\srcreal[k]}
    \\[6pt]
    \Fun[np]{\chad}{\tau \times \sigma}_{2}
    &:=&
    \casubst{\Fun[np]{\chad}{\tau}_{2}}{p_{\tau}}{\opn[b]{pr}_{1}p_{\tau\times\sigma}}
    \times
    \casubst{\Fun[np]{\chad}{\sigma}_{2}}{p_{\sigma}}{\opn[b]{pr}_{2}p_{\tau\times\sigma}}.
    \\[6pt]
    \Fun[np]{\chad}{\tau \sqcup \sigma}_{2}
    &:=&
    \opn[b]{case}\;p_{\tau\sqcup\sigma}\;\opn[b]{of}
    \left\{
    \begin{array}{l}
      \opn[b]{inl}\,x \mapsto
      \casubst{\Fun[np]{\chad}{\tau}_{2}}{p_{\tau}}{x}
      \\
      \opn[b]{inr}\,y \mapsto
      \casubst{\Fun[np]{\chad}{\sigma}_{2}}{p_{\sigma}}{y}
    \end{array}
    \right\}.
  \end{array}
\end{equation}

For the linear case-type former we use the usual dependent \(\beta\)-rules.  For
example,
\begin{equation*}
  \casubst{\Fun[np]{\chad}{\tau \sqcup \sigma}_{2}}{p_{\tau\sqcup\sigma}}{\opn[b]{inl}a}
  \equiv
  \casubst{\Fun[np]{\chad}{\tau}_{2}}{p_{\tau}}{a},
  \quad
  \casubst{\Fun[np]{\chad}{\tau \sqcup \sigma}_{2}}{p_{\tau\sqcup\sigma}}{\opn[b]{inr}b}
  \equiv
  \casubst{\Fun[np]{\chad}{\sigma}_{2}}{p_{\sigma}}{b}.
\end{equation*}
The same convention applies to the \(n\)-ary coproduct notation introduced in
Section~\ref{sec:appendix-source-language-for-chad}.

\paragraph{The monadic type.}
For the FAD monad, the primal remains a monadic computation and the cotangent is
the lifted monadic cotangent type of Section~\ref{sec:appendix-target-language-for-chad}:
\begin{equation}
  \label{eq:appendix-chad-type-translation-monad}
  \begin{aligned}
    \Fun[np]{\chad}{\Fun[nn]{\opn[b]{M}}{\tau}}_{1}
    &:=
    \Fun[nn]{\opn[b]{M}}{\Fun[np]{\chad}{\tau}_{1}},
    \\
    \Fun[np]{\chad}{\Fun[nn]{\opn[b]{M}}{\tau}}_{2}
    &:=
    \opn[b]{M}^{w}_{p_{\opn[b]{M}\tau} : \Fun[nn]{\opn[b]{M}}{\Fun[np]{\chad}{\tau}_{1}}}
    \left[
      x \mapsto \casubst{\Fun[np]{\chad}{\tau}_{2}}{p_{\tau}}{x}
    \right].
  \end{aligned}
\end{equation}
Here \(p_{\opn[b]{M}\tau}\) is the primal variable for the source type
\(\Fun[nn]{\opn[b]{M}}{\tau}\), while \(p_{\tau}\) is the primal variable
for the returned value type \(\tau\).  This distinction is necessary because the
continuation cotangent family is indexed by returned values, not by the whole
monadic computation.

\paragraph{Contexts.}
Contexts are translated componentwise on primals and by fibrewise biproducts on
cotangents:
\begin{equation}
  \label{eq:appendix-chad-context-translation}
  \begin{aligned}
    \Fun[np]{\chad}{\cdot}_{1} &:= \cdot,
    &
    \Fun[np]{\chad}{\Gamma, x : \tau}_{1}
    &:=
    \Fun[np]{\chad}{\Gamma}_{1},\; x : \Fun[np]{\chad}{\tau}_{1},
    \\
    \Fun[np]{\chad}{\cdot}_{2} &:= \lintype{1},
    &
    \Fun[np]{\chad}{\Gamma, x : \tau}_{2}
    &:=
    \Fun[np]{\chad}{\Gamma}_{2}
    \times
    \casubst{\Fun[np]{\chad}{\tau}_{2}}{p_{\tau}}{x}.
  \end{aligned}
\end{equation}
The first factor \(\Fun[np]{\chad}{\Gamma}_{2}\) is implicitly reindexed
along the cartesian weakening
\(\Fun[np]{\chad}{\Gamma}_{1},x:\Fun[np]{\chad}{\tau}_{1}\to
\Fun[np]{\chad}{\Gamma}_{1}\).  Thus \(\opn[b]{pr}_{1}\) selects the cotangent
tuple for the previous context and \(\opn[b]{pr}_{2}\) selects the cotangent for
the newest variable.

For \(\Gamma = x_{1}:\tau_{1},\dots,x_{n}:\tau_{n}\), define the derived
injection of a single variable cotangent into the context cotangent by recursion.
Writing \(\fn[n]{\opn{idx}}{x_{j};\Gamma}=j\), if
\(\Gamma=\Gamma',x_{n}:\tau_{n}\), then
\begin{equation}
  \label{eq:appendix-chad-coproj-definition}
  \begin{aligned}
    \opn[b]{coproj}^{\Gamma}_{n}(v)
    &:=
    \tuple{
      \opn[b]{0}_{\Fun[np]{\chad}{\Gamma'}_{2}},
      v
    },
    \\
    \opn[b]{coproj}^{\Gamma}_{j}(v)
    &:=
    \tuple{
      \opn[b]{coproj}^{\Gamma'}_{j}(v),
      \opn[b]{0}_{\casubst{\Fun[np]{\chad}{\tau_{n}}_{2}}{p_{\tau_n}}{x_{n}}}
    }
    \qquad (1\leq j<n).
  \end{aligned}
\end{equation}
We write \(\opn[b]{coproj}_{\fn[n]{\opn{idx}}{x;\Gamma}}v\) when \(\Gamma\) is
clear from context.

\subsection{Term translation}

All zero terms below are the typed zeroes of the relevant fibrewise biproduct.
The term translation is by structural recursion on the source typing derivation.

\paragraph{Nullary terms and variables.}
\begin{align}
  \label{eq:appendix-chad-unit}
  \Fun[np]{\chad[\Gamma]}{\tuple{\;}}
  &:=
  \tuple{
    \tuple{\;},
    \lambda v\; .\; \opn[b]{0}_{\Fun[np]{\chad}{\Gamma}_{2}}
  },
  \\
  \label{eq:appendix-chad-abort}
  \Fun[np]{\chad[\Gamma]}{\fn[n]{\opn[b]{abort}}{t}}
  &:=
  \opn[b]{case}\; \Fun[np]{\chad[\Gamma]}{t}\; \opn[b]{of}\;
  \tuple{p,b}\to
  \fn[n]{\opn[b]{abort}}{p},
  \\
  \label{eq:appendix-chad-variable}
  \Fun[np]{\chad[\Gamma]}{x}
  &:=
  \tuple{
    x,
    \lambda v\; .\;
    \opn[b]{coproj}_{\fn[n]{\opn{idx}}{x;\Gamma}}v
  }.
\end{align}

\paragraph{Products.}
\begin{align}
  \label{eq:appendix-chad-pair}
  \Fun[np]{\chad[\Gamma]}{\tuple{t,s}}
  &:=
  \opn[b]{case}\;\Fun[np]{\chad[\Gamma]}{t}\;\opn[b]{of}\;\tuple{p,b}\to
  \opn[b]{case}\;\Fun[np]{\chad[\Gamma]}{s}\;\opn[b]{of}\;\tuple{q,c}\to
  \notag\\[-2pt]
  &\qquad
  \tuple{
    \tuple{p,q},
    \lambda v\; .\;
    \fn[n]{b}{\opn[b]{pr}_{1}v}
    +
    \fn[n]{c}{\opn[b]{pr}_{2}v}
  },
  \\
  \label{eq:appendix-chad-pr1}
  \Fun[np]{\chad[\Gamma]}{\opn[b]{pr}_{1}t}
  &:=
  \opn[b]{case}\;\Fun[np]{\chad[\Gamma]}{t}\;\opn[b]{of}\;\tuple{p,b}\to
  \notag\\[-2pt]
  &\qquad
  \tuple{
    \opn[b]{pr}_{1}p,
    \lambda v\; .\;
    \fn[n]{b}{
      \tuple{
        v,
        \opn[b]{0}_{\casubst{\Fun[np]{\chad}{\sigma}_{2}}{p_{\sigma}}{\opn[b]{pr}_{2}p}}
      }
    }
  }
  \quad (t: \tau\times\sigma),
  \\
  \label{eq:appendix-chad-pr2}
  \Fun[np]{\chad[\Gamma]}{\opn[b]{pr}_{2}t}
  &:=
  \opn[b]{case}\;\Fun[np]{\chad[\Gamma]}{t}\;\opn[b]{of}\;\tuple{p,b}\to
  \notag\\[-2pt]
  &\qquad
  \tuple{
    \opn[b]{pr}_{2}p,
    \lambda v\; .\;
    \fn[n]{b}{
      \tuple{
        \opn[b]{0}_{\casubst{\Fun[np]{\chad}{\tau}_{2}}{p_{\tau}}{\opn[b]{pr}_{1}p}},
        v
      }
    }
  }
  \quad (t: \tau\times\sigma).
\end{align}

\paragraph{Coproducts.}
For \(n\)-ary sums, the following introduction rule specializes to
\(\opn[b]{inl}\) and \(\opn[b]{inr}\) in the binary grammar:
\begin{equation}
  \label{eq:appendix-chad-sum-intro}
  \Fun[np]{\chad[\Gamma]}{\opn[b]{in}_{i}t}
  :=
  \opn[b]{case}\;\Fun[np]{\chad[\Gamma]}{t}\;\opn[b]{of}\;\tuple{x,x'}\to
  \tuple{\opn[b]{in}_{i}x,x'}.
\end{equation}
For case analysis, suppose
\(\Gamma\vdash t:\tau_{1}\sqcup\cdots\sqcup\tau_{n}\) and
\(\Gamma,x_{i}:\tau_{i}\vdash s_{i}:\rho\).  Then
\begin{align}
  \label{eq:appendix-chad-sum-elim}
  &\Fun[np]{\chad[\Gamma]}{
    \opn[b]{case}\;t\;\opn[b]{of}\;
    \left\{\opn[b]{in}_{i}x_{i}\to s_{i}\right\}_{i=1}^{n}
  }
  :=
  \opn[b]{case}\;\Fun[np]{\chad[\Gamma]}{t}\;\opn[b]{of}\;\tuple{y,y'}\to
  \notag\\
  &\quad
  \opn[b]{case}\;y\;\opn[b]{of}
  \left\{
  \begin{array}{l}
    \opn[b]{in}_{i}x_{i}\mapsto
    \opn[b]{case}\;\Fun[np]{\chad[\Gamma,x_{i}:\tau_{i}]}{s_{i}}\;\opn[b]{of}\;\tuple{z_{i},z_{i}'}\to
    \\
    \qquad
    \tuple{
      z_{i},
      \lambda v\; .\;
      \letst{v}{\fn[n]{z_{i}'}{v}}{
        \opn[b]{pr}_{1}v
        +
        \fn[n]{y'}{\opn[b]{pr}_{2}v}
      }
    }
  \end{array}
  \right\}_{i=1}^{n}.
\end{align}
The application of \(y'\) in the branch uses the case-type \(\beta\)-rule above,
so \(\opn[b]{pr}_{2}v\) has precisely the cotangent type of the active summand.

\paragraph{Let bindings and deterministic operations.}
If \(\Gamma\vdash t:\sigma\) and \(\Gamma,x:\sigma\vdash s:\rho\), then
\begin{align}
  \label{eq:appendix-chad-let}
  \Fun[np]{\chad[\Gamma]}{\letst{x}{t}{s}}
  &:=
  \opn[b]{case}\;\Fun[np]{\chad[\Gamma]}{t}\;\opn[b]{of}\;\tuple{p,b}\to
  \opn[b]{case}\;\casubst{\Fun[np]{\chad[\Gamma,x:\sigma]}{s}}{x}{p}\;\opn[b]{of}\;\tuple{q,c}\to
  \notag\\[-2pt]
  &\qquad
  \tuple{
    q,
    \lambda v\; .\;
    \letst{v}{\fn[n]{c}{v}}{
      \opn[b]{pr}_{1}v + \fn[n]{b}{\opn[b]{pr}_{2}v}
    }
  }.
\end{align}
For a built-in operation
\(\srcop: \srcreal[n_{1}]\times\cdots\times\srcreal[n_{k}]\to\srcreal[m]\),
write \(\tgtop\) for the chosen transposed Jacobian term constructor.  Its CHAD
clause is
\begin{align}
  \label{eq:appendix-chad-op}
  &\Fun[np]{\chad[\Gamma]}{\fn[n]{\srcop}{t_{1},\dots,t_{k}}}
  :=
  \notag\\
  &\quad
  \opn[b]{case}\;\Fun[np]{\chad[\Gamma]}{t_{1}}\;\opn[b]{of}\;\tuple{x_{1},x_{1}'}\to
  \cdots
  \opn[b]{case}\;\Fun[np]{\chad[\Gamma]}{t_{k}}\;\opn[b]{of}\;\tuple{x_{k},x_{k}'}\to
  \notag\\[-2pt]
  &\quad\qquad
  \tuple{
    \fn[n]{\srcop}{x_{1},\dots,x_{k}},
    \lambda v\; .\;
    \letst{v}{\fn[n]{\fn[n]{\tgtop}{x_{1},\dots,x_{k}}}{v}}{
      \sum_{j=1}^{k}
      \fn[n]{x_{j}'}{\opn[b]{pr}_{j}v}
    }
  }.
\end{align}
Here \(\opn[b]{pr}_{j}\) denotes the derived projection from the cotangent product
returned by \(\tgtop\).

\paragraph{Monad unit and multiplication.}
The unit clause is the one displayed in Equation~\eqref{eq:chad-code-transformation-return}:
\begin{equation}
  \label{eq:appendix-chad-return}
  \Fun[np]{\chad[\Gamma]}{\fn[n]{\opn[b]{return}}{t}}
  :=
  \opn[b]{case}\;\Fun[np]{\chad[\Gamma]}{t}\;\opn[b]{of}\;\tuple{p,b}\to
  \tuple{
    \fn[n]{\opn[b]{return}}{p},
    \lambda v\; .\;
    \fn[n]{\opn[b]{return}^{w}_{p}}{\fn[n]{b}{v}}
  }.
\end{equation}
For bind, let \(\Gamma\vdash t:\Fun[nn]{\opn[b]{M}}{\sigma}\) and
\(\Gamma,x:\sigma\vdash s:\Fun[nn]{\opn[b]{M}}{\tau}\).  In the scope of
\(x\), put
\[
  r_{x}
  :=
  \opn[b]{case}\;\Fun[np]{\chad[\Gamma,x:\sigma]}{s}\;\opn[b]{of}\;\tuple{q,c}\to q,
  \qquad
  c_{x}
  :=
  \opn[b]{case}\;\Fun[np]{\chad[\Gamma,x:\sigma]}{s}\;\opn[b]{of}\;\tuple{q,c}\to c.
\]
Then the bind clause is
\begin{equation}
  \label{eq:appendix-chad-bind}
  \Fun[np]{\chad[\Gamma]}{\fn[n]{\opn[b]{bind}}{t; x\; .\; s}}
  :=
  \opn[b]{case}\;\Fun[np]{\chad[\Gamma]}{t}\;\opn[b]{of}\;\tuple{p,b}\to
  \tuple{
    \fn[n]{\opn[b]{bind}}{p; x\; .\; r_{x}},
    \lambda v\; .\;
    \fn[n]{\opn[b]{bind}^{w}_{p; x\; .\; r_{x}}}{
      \fn[n]{b}{v};
      \fn[n]{c_{x}}{v}
    }
  }.
\end{equation}
The first argument of \(\opn[b]{bind}^{w}\) accumulates the cotangent contribution
through the sampled computation \(t\), while the second argument accumulates the
cotangent contribution through the continuation \(s\), including the cotangent of
the bound variable \(x\).

\paragraph{FAD constructors and handler.}
For the finite atomic distribution constructor, assume
\(\Gamma\vdash t_{i}:\tau\) and \(\Gamma\vdash w_{i}:\srcreal\) for
\(1\leq i\leq n\).  Then
\begin{align}
  \label{eq:appendix-chad-categorical}
  &\Fun[np]{\chad[\Gamma]}{
    \fn[n]{\opn[b]{categorical}}{\tuple{\tuple{t_{i},w_{i}}}_{i=1}^{n}}
  }
  :=
  \notag\\
  &\quad
  \opn[b]{case}\;\Fun[np]{\chad[\Gamma]}{t_{1}}\;\opn[b]{of}\;\tuple{u_{1},u_{1}'}\to
  \cdots
  \opn[b]{case}\;\Fun[np]{\chad[\Gamma]}{t_{n}}\;\opn[b]{of}\;\tuple{u_{n},u_{n}'}\to
  \notag\\
  &\quad
  \opn[b]{case}\;\Fun[np]{\chad[\Gamma]}{w_{1}}\;\opn[b]{of}\;\tuple{a_{1},a_{1}'}\to
  \cdots
  \opn[b]{case}\;\Fun[np]{\chad[\Gamma]}{w_{n}}\;\opn[b]{of}\;\tuple{a_{n},a_{n}'}\to
  \notag\\[-2pt]
  &\quad\qquad
  \tuple{
    \fn[n]{\opn[b]{categorical}}{\tuple{\tuple{u_{i},a_{i}}}_{i=1}^{n}},
    \lambda v\; .\;
    \fn[n]{
      \opn[b]{categorical}^{w}_{\tuple{\tuple{u_{i},a_{i}}}_{i=1}^{n}}
    }{
      \tuple{\tuple{\fn[n]{u_{i}'}{v},\fn[n]{a_{i}'}{v}}}_{i=1}^{n}
    }
  }.
\end{align}
Finally, for the handler algebra \(\opn[b]{E}_{\nu}\), where
\(\nu=\srcreal[p]\), we define
\begin{equation}
  \label{eq:appendix-chad-handler}
  \Fun[np]{\chad[\Gamma]}{\fn[n]{\opn[b]{E}_{\nu}}{t}}
  :=
  \opn[b]{case}\;\Fun[np]{\chad[\Gamma]}{t}\;\opn[b]{of}\;\tuple{p,b}\to
  \tuple{
    \fn[n]{\opn[b]{E}_{\nu}}{p},
    \lambda v\; .\;
    \fn[n]{\opn[b]{E}^{w}_{\nu,p}}{\fn[n]{b}{v}}
  }.
\end{equation}

\subsection{Typing and universal characterization}

\begin{theorem}[Well-typedness and uniqueness of the CHAD translation]
  \label{th:appendix-chad-well-typedness-uniqueness}
  For every well-typed source judgement \(\Gamma\vdash t:\tau\), the term
  \(\Fun[np]{\chad[\Gamma]}{t}\) defined above is well typed as in
  Equation~\eqref{eq:appendix-chad-term-typing-shape}.  Moreover, these clauses
  determine the unique structure-preserving functor \(\chad:\Syn\to\Tgt\) that
  preserves finite products and coproducts, sends each source monadic
  type \(\Fun[nn]{\opn[b]{M}}{\tau}\) to the displayed pair consisting of the
  primal \(\opn[b]{M}\)-type and its cotangent lift \(\opn[b]{M}^{w}\),
  interprets each \(\srcop\) together with its transposed Jacobian
  \(\tgtop\), and interprets the FAD constructors and the handler algebra by
  the displayed pairs \(\left(\opn[b]{categorical},\opn[b]{categorical}^{w}\right)\)
  and \(\left(\opn[b]{E}_{\nu},\opn[b]{E}^{w}_{\nu}\right)\), respectively.
\end{theorem}

The typing statement follows by induction on the source typing derivation.  The
only non-structural cases are the monadic and FAD constructors, where the typing
rules for \(\opn[b]{return}^{w}\), \(\opn[b]{bind}^{w}\),
\(\opn[b]{categorical}^{w}\), and \(\opn[b]{E}^{w}\) are exactly the target rules
used in the displayed clauses.  The uniqueness statement is the universal
property of the freely generated source syntactic category
(Lemma~\ref{lem:appendix-universal-property-of-the-syntactic-category-of-the-source-language})
applied to the corresponding structure in \(\Tgt\).  The target equations of
Section~\ref{sec:appendix-target-language-for-chad} ensure that the source
Kleisli, distributive, and Eilenberg--Moore equations are preserved.

\section{Example of differentiated code}
\label{sec:appendix-example-of-differentiated-code}

This section gives the complete differentiated example code from Equation~\eqref{eq:example-source-code-encoding}, following the inductive rules of code differentiation outlined in Appendix~\ref{sec:appendix-code-transformation-rules}. As in Section~\ref{sec:returning-to-the-example}, we adopt the following abbreviations:
\begin{equation*}
  \begin{aligned}
    a_{0} &:= 0
    \\
    a_{1} &:= P \cdot \fn[n]{ \opn[b]{sig} }{ -s_{p} \cdot \left( p - p_{0} \right) }
  \end{aligned}
  \qquad
  \begin{aligned}
    \gamma_{0} &:= \fn[n]{ \opn[b]{lsig} }{ -s_{m} \cdot \left( m - m_{0} \right) }
    \\
    \gamma_{1} &:= \fn[n]{ \opn[b]{lsig} }{ s_{m} \cdot \left( m - m_{0} \right) }
  \end{aligned}
\end{equation*}
The full code transformation is then given by
{
  \allowdisplaybreaks
  \begin{alignat*}{2}
    &\opn[b]{case}
    \\
    &\qquad \opn[b]{case} \; \tuple{ a_{0}, \grad a_{0} } \; \opn[b]{of}
      \tuple{ \theta_{0}, \theta_{0}' } \to
    \\
    &\qquad \opn[b]{case} \; \tuple{ a_{1}, \grad a_{1} } \; \opn[b]{of}
      \tuple{ \theta_{1}, \theta_{1}' } \to
    \\
    &\qquad \opn[b]{case} \; \tuple{ \gamma_{0}, \grad \gamma_{0} } \; \opn[b]{of}
      \tuple{ \rho_{0}, \rho_{0}' } \to
    \\
    &\qquad \opn[b]{case} \; \tuple{ \gamma_{1}, \grad \gamma_{1} } \; \opn[b]{of}
      \tuple{ \rho_{1}, \rho_{1}' } \to
    \\
    &\qquad \qquad
      \tuple{
      \fn[n]{ \opn[b]{categorical} }{
      \tuple{ \tuple{ \theta_{j}, \rho_{j} } }_{j = 0}^{1}
      },
      \lambda v \; . \;
      \fn[n]{
      \opn[b]{categorical}^{w}_{
      \tuple{ \tuple{ \theta_{j}, \rho_{j} } }_{j = 0}^{1}
      }
      }{
      \tuple{ \tuple{ \fn[n]{ \theta_{j}' }{ v }, \fn[n]{ \rho_{j}' }{ v } } }_{j = 0}^{1}
      }
      }
    \\
    &\opn[b]{of} \tuple{ \alpha, \alpha' } \to
    \\
    &\qquad
      \opn[b]{case}
    \\
    &\qquad \qquad
      \opn[b]{case}
    \\
    &\qquad \qquad \qquad
      \opn[b]{case}
      \tuple{ s, \lambda v \; . \; \opn[b]{coproj}_{\fn[n]{ \opn{idx} }{ s; \Gamma, s }} v }
    \\
    &\qquad \qquad \qquad
      \opn[b]{of} \tuple{ \eta, \eta' } \to
    \\
    &\qquad \qquad \qquad \qquad
      \opn[b]{let} \xi =
    \\
    &\qquad \qquad \qquad \qquad \qquad
      \opn[b]{case}
      \\
    &\qquad \qquad \qquad \qquad \qquad \qquad
      \opn[b]{case}
      \tuple{ \left( p \cdot x - m \right), \grad_{p, m, x} \left( p \cdot x - m \right) }
      \\
    &\qquad \qquad \qquad \qquad \qquad \qquad
      \opn[b]{of} \tuple{ \mu, \mu' } \to
      \tuple{
      \fn[n]{ \opn[b]{return} }{ \mu },
      \lambda v \; . \;
      \fn[n]{ \opn[b]{return}^{w}_{\mu} }{ \fn[n]{ \mu' }{ v } }
      }
      \\
    &\qquad \qquad \qquad \qquad \qquad
      \opn[b]{of} \tuple{ \kappa, \kappa' } \to \kappa
    \\
    &\qquad \qquad \qquad \qquad
      \opn[b]{in}
    \\
    &\qquad \qquad \qquad \qquad \qquad
      \left\langle
      \opn[b]{bind}
      \left(
      \eta; x \; . \; \xi
      \right),
      \lambda v \; . \;
      \opn[b]{bind}_{\eta; x \; . \; \xi}^{w}
      \left(
      \fn[n]{ \eta' }{ v };
      \right. 
      \right.
    \\
    &\qquad \qquad \qquad \qquad \qquad \qquad
      \left(
      \opn[b]{case}
      \right.
      \\
    &\qquad \qquad \qquad \qquad \qquad \qquad \qquad
      \opn[b]{case}
      \tuple{ \left( p \cdot x - m \right), \grad_{p, m, x} \left( p \cdot x - m \right) }
      \\
    &\qquad \qquad \qquad \qquad \qquad \qquad \qquad
      \opn[b]{of} \tuple{ \mu, \mu' } \to
      \tuple{
      \fn[n]{ \opn[b]{return} }{ \mu },
      \lambda v \; . \;
      \fn[n]{ \opn[b]{return}^{w}_{\mu} }{ \fn[n]{ \mu' }{ v } }
      }
      \\
    &\qquad \qquad \qquad \qquad \qquad \qquad
      \left.
      \left.\left.
      \opn[b]{of} \tuple{ \kappa, \kappa' } \to \kappa'
      \right) 
      \left( v \right)
      \right) 
      \right\rangle
    \\
    &\qquad \qquad
      \opn[b]{of} \tuple{ \delta, \delta' } \to
    \\
    &\qquad \qquad \qquad
      \opn[b]{case}
    \\
    &\qquad \qquad \qquad \qquad
      \opn[b]{case}
      \tuple{ t, \lambda v \; . \; \opn[b]{coproj}_{\fn[n]{ \opn{idx} }{ t; \Gamma, s, t }} v}
    \\
    &\qquad \qquad \qquad \qquad
      \opn[b]{of} \tuple{ \zeta, \zeta' } \to
      \tuple{
      \fn[n]{ \opn[b]{E}_{\srcreal} }{ \zeta },
      \lambda v \; . \; 
      \fn[n]{ \opn[b]{E}_{\srcreal, \zeta}^{w} }{ \fn[n]{ \zeta' }{ v } }
      }
    \\
    &\qquad \qquad \qquad
      \opn[b]{of} \tuple{ \epsilon, \epsilon' } \to
      \tuple{
      \casubst{\epsilon}{t}{\delta},
      \lambda v \; . \;
      \letst{w}{
      \fn[n]{ \casubst{\epsilon'}{t}{\delta} }{ v }
      }{
      \opn[b]{pr}_{1} w + \fn[n]{ \delta' }{ \opn[b]{pr}_{2} w }
      }
      }
    \\
    &\qquad
      \opn[b]{of} \tuple{ \beta, \beta' } \to
      \tuple{
      \casubst{\beta}{s}{\alpha},
      \lambda v \; . \;
      \letst{ w }{ \fn[n]{ \casubst{\beta'}{s}{\alpha} }{v} }{
      \opn[b]{pr}_{1} w + \fn[n]{ \alpha' }{ \opn[b]{pr}_{2} w }
      } 
      }
  \end{alignat*}
}

\section{Deriving concrete gradient estimation formulas}
We now explain how the usual gradient rule for weighted particles is
already contained in the lifted FAD monad and in the lifted canonical
algebra.

In this section, we write the second component of a morphism in
\(\Fam{\Vect[o]}\) using the symbol \(\nabla\).  Thus a morphism
\(\hat h:T^*X\to T^*Y\) is written as
\[
  \hat h=(h,\nabla h),
\]
where \(h:X\to Y\) is its base part and
\[
  (\nabla h)_x:T^*_{h(x)}Y\to T^*_xX
\]
is the \(x\)-fibre of its fibre part.  This notation is suggestive:
\(\nabla h\) denotes the chosen cotangent map in
\(\Fam{\Vect[o]}\), whether or not \(X\) and \(Y\) carry genuine smooth
structure.

Write
\[
  \opn{E}^{\cat[]{F}}_{\R}:M^{\cat[]{F}}\TR\to\TR
\]
for the lifted canonical algebra on real values.  Its base part is the
unnormalised expectation map
\[
  u\longmapsto
  \sum_{r\in\supp u}u(r)r,
  \qquad
  u\in\Fun[nn]{M}{\R},
\]
and its fibre part is
\begin{equation}
  \label{eq:appendix-lifted-real-expectation-fibre}
  \left(\nabla\opn{E}^{\cat[]{F}}_{\R}\right)_u(t)
  =
  \sum_{r\in\supp u}
  \left(
    u(r)t\cdot\vect{e}_{r}
    +
    u(r)r t\cdot\vect{\varepsilon}_{r}
  \right).
\end{equation}
Here the \(\vect{e}_{r}\)-component is the cotangent of the real value
\(r\), while the \(\vect{\varepsilon}_{r}\)-component is the cotangent of
the log-weight attached to \(r\).

Let
\[
  T^*X=\tuple{X,\{T^*_xX\}_{x\in X}}
\]
be an object of \(\Fam{\Vect[o]}\), and let
\[
  \hat g=(g,\nabla g):T^*X\to\TR
\]
be a real-valued test morphism.  

For \(w\in\Fun[nn]{M}{X}\), set
\[
  u:=\fn[n]{\Fun[nn]{M}{g}}{w}.
\]
The Frobenius, or pushforward, identity says on bases that
\[
  \int_X g(x)\,dw
  =
  \int_{\R} r\,d\fn[n]{\Fun[nn]{M}{g}}{w},
\]
or explicitly
\[
  \sum_{x\in\supp w}w(x)g(x)
  =
  \sum_{r\in\supp u}u(r)r.
\]
In categorical form, this is the base part of the composite
\[
  M^{\cat[]{F}}T^*X
  \xrightarrow{M^{\cat[]{F}}\hat g}
  M^{\cat[]{F}}\TR
  \xrightarrow{\opn{E}^{\cat[]{F}}_{\R}}
  \TR.
\]

The fibre part of this composite is obtained by composing
Equation~\eqref{eq:appendix-lifted-real-expectation-fibre} with the fibre part of
\(M^{\cat[]{F}}\hat g\).
For \(x\in\supp w\), the denominator is nonzero, and the
corresponding summand is
\[
  \frac{w(x)}{u(g(x))}
  \left(
    (\nabla g)_x
    \left(
      u(g(x))t
    \right)
    \cdot\vect{e}_{x}
    +
    u(g(x))g(x)t\cdot\vect{\varepsilon}_{x}
  \right).
\]
Since \((\nabla g)_x\) is linear, the factor \(u(g(x))\) in the
atom-cotangent part cancels with the denominator; the log-weight part
cancels directly.  Hence
\begin{equation}
  \label{eq:appendix-frobenius-cotangent-rule}
  \left(
    \nabla\!\left(
      \opn{E}^{\cat[]{F}}_{\R}
      \circ
      M^{\cat[]{F}}\hat g
    \right)
  \right)_w(t)
  =
  \sum_{x\in\supp w}
  \left(
    w(x)(\nabla g)_x(t)\cdot\vect{e}_{x}
    +
    w(x)g(x)t\cdot\vect{\varepsilon}_{x}
  \right).
\end{equation}
Thus the lifted algebra and the lifted FAD monad already produce the two
basic cotangent contributions: a value cotangent weighted by \(w(x)\),
and a log-weight cotangent weighted by \(w(x)g(x)\).

Now apply this to a parameterised family of particles.  Let
\[
  T^*\Theta=\tuple{\Theta,\{T^*_\theta\Theta\}_{\theta\in\Theta}},
  \qquad
  T^*Z=\tuple{Z,\{T^*_zZ\}_{z\in Z}},
\]
and let
\[
  \hat q=(q,\nabla q):T^*\Theta\to M^{\cat[]{F}}T^*Z
\]
be a parameterised FAD, equipped with its chosen cotangent map.  Locally,
away from collisions of particles, choose a labelling of the support and
write the base part of \(\hat q\) as
\[
  \fn[n]{ q}{\theta}
  =
  \sum_{i\in I}w_i(\theta)\delta_{z_i(\theta)}.
\]
The corresponding strengthened FAD over
\(T^*\Theta\times T^*Z\) is
\[
  \hat{\underline q}=\tuple{\underline q,\nabla\underline q}
  :=
  \Nat[nn]{\str^{\cat[]{F}}}{T^*\Theta,T^*Z}
  \circ
  \tupling{\id_{T^*\Theta},\hat q}
  :
  T^*\Theta
  \to
  M^{\cat[]{F}}(T^*\Theta\times T^*Z),
\]
whose base part is
\[
  \fn[n]{\underline q}{\theta}
  =
  \sum_{i\in I}w_i(\theta)\delta_{(\theta,z_i(\theta))}.
\]

In the chosen local labelling, the fibre of
\(M^{\cat[]{F}}(T^*\Theta\times T^*Z)\) at
\(\fn[n]{\underline q}{\theta}\) identifies with
\[
  \bigoplus_{i\in I}
  \left(
    \left(
      T^*_\theta\Theta\oplus T^*_{z_i(\theta)}Z
    \right)
    \oplus\R
  \right).
\]
We write elements of this fibre as finite sums
\[
  \sum_{i\in I}
  \left(
    \tuple{d_i,a_i}\cdot\vect{e}_{(\theta,z_i(\theta))}
    +
    s_i\cdot\vect{\varepsilon}_{(\theta,z_i(\theta))}
  \right),
\]
where
\[
  d_i\in T^*_\theta\Theta,
  \qquad
  a_i\in T^*_{z_i(\theta)}Z,
  \qquad
  s_i\in\R.
\]
Let
\[
  \ell_i:\Theta\to\R,
  \qquad
  \ell_i(\theta)=\log w_i(\theta).
\]
In these coordinates, the fibre part of \(\hat{\underline q}\) is
\begin{equation}
  \label{eq:appendix-suggestive-fibre-parametrised-fad}
  \begin{aligned}
  &
  (\nabla\underline q)_\theta
  \left(
    \sum_{i\in I}
    \left(
      \tuple{d_i,a_i}\cdot\vect{e}_{(\theta,z_i(\theta))}
      +
      s_i\cdot\vect{\varepsilon}_{(\theta,z_i(\theta))}
    \right)
  \right)
  \\
  &\qquad =
  \sum_{i\in I}d_i
  +
  \sum_{i\in I}(\nabla z_i)_\theta(a_i)
  +
  \sum_{i\in I}(\nabla\ell_i)_\theta(s_i).
  \end{aligned}
\end{equation}
Here
\[
  (\nabla z_i)_\theta:T^*_{z_i(\theta)}Z\to T^*_\theta\Theta
\]
is the \(\theta\)-fibre of the fibre part of the chosen morphism
\(\hat z_i=(z_i,\nabla z_i):T^*\Theta\to T^*Z\) representing the
parameterised particle, and
\[
  (\nabla\ell_i)_\theta:\R\to T^*_\theta\Theta
\]
is the \(\theta\)-fibre of the fibre part of the chosen morphism
\(\hat{\ell}_i=(\ell_i,\nabla\ell_i):T^*\Theta\to\TR\) representing the
parameterised log-weight.  Thus an atom cotangent is sent backwards
through the moving particle \(z_i\), while a log-weight cotangent is sent
backwards through \(\ell_i=\log w_i\).

Finally, let
\[
  \widehat f=(f,\nabla f):
  T^*\Theta\times T^*Z\to\TR
\]
be a parameter-dependent test morphism, with base part
\[
  (\theta,z)\longmapsto f(\theta,z).
\]
We write the \((\theta,z)\)-fibre of its fibre part as
\[
  (\nabla f)_{(\theta,z)}(t)
  =
  \tuple{
    (\nabla f)^\Theta_{(\theta,z)}(t),
    (\nabla f)^Z_{(\theta,z)}(t)
  },
\]
where
\[
  (\nabla f)^\Theta_{(\theta,z)}(t)\in T^*_\theta\Theta,
  \qquad
  (\nabla f)^Z_{(\theta,z)}(t)\in T^*_zZ.
\]
Consider the composite
\[
  T^*\Theta
  \xrightarrow{\hat{\underline q}}
  M^{\cat[]{F}}(T^*\Theta\times T^*Z)
  \xrightarrow{M^{\cat[]{F}}\widehat f}
  M^{\cat[]{F}}\TR
  \xrightarrow{\opn{E}^{\cat[]{F}}_{\R}}
  \TR.
\]
By the Frobenius identity, its base part is
\[
  \theta
  \longmapsto
  \sum_{i\in I}
  w_i(\theta)f(\theta,z_i(\theta)).
\]

Applying Equation~\eqref{eq:appendix-frobenius-cotangent-rule} to
\(\widehat f:T^*\Theta\times T^*Z\to\TR\), an output cotangent
\(t\in\R\) is first sent to the following cotangent of the FAD over
\(\Theta\times Z\):
\[
  \sum_{i\in I}
  \left(
    w_i(\theta)
    \tuple{
      (\nabla f)^\Theta_{(\theta,z_i(\theta))}(t),
      (\nabla f)^Z_{(\theta,z_i(\theta))}(t)
    }
    \cdot\vect{e}_{(\theta,z_i(\theta))}
    +
    w_i(\theta)f(\theta,z_i(\theta))t
    \cdot\vect{\varepsilon}_{(\theta,z_i(\theta))}
  \right).
\]
Composing with Equation~\eqref{eq:appendix-suggestive-fibre-parametrised-fad}
therefore gives
\begin{equation}
  \label{eq:appendix-three-term-weighted-particle-gradient}
  \begin{aligned}
  &
  \left(
    \nabla\!\left(
      \opn{E}^{\cat[]{F}}_{\R}
      \circ
      M^{\cat[]{F}}\widehat f
      \circ
      \hat{\underline q}
    \right)
  \right)_\theta(t)
  \\
  &\qquad =
  \sum_{i\in I}
  w_i(\theta)(\nabla f)^\Theta_{(\theta,z_i(\theta))}(t)
  +
  \sum_{i\in I}
  w_i(\theta)
  (\nabla z_i)_\theta
  \left(
    (\nabla f)^Z_{(\theta,z_i(\theta))}(t)
  \right)
  \\
  &\qquad\qquad
  +
  \sum_{i\in I}
  (\nabla\ell_i)_\theta
  \left(
    w_i(\theta)f(\theta,z_i(\theta))t
  \right).
  \end{aligned}
\end{equation}

The three summands in
Equation~\eqref{eq:appendix-three-term-weighted-particle-gradient} correspond to
three different cotangent routes.  The first term comes from the
\(T^*_\theta\Theta\)-component of the test morphism
\(\widehat f:T^*\Theta\times T^*Z\to\TR\), and represents the explicit
dependence of \(f\) on the parameter.  The second term comes from the
\[
  T^*_{z_i(\theta)}Z
\]
part of the atom component inside the fibre of
\(M^{\cat[]{F}}(T^*\Theta\times T^*Z)\), and represents how moving the
particle \(z_i(\theta)\) changes the value of the test function.  The
third term comes from the real component
\[
  \R
\]
inside the same fibre, and represents how changing the log-weight
\(\ell_i(\theta)=\log w_i(\theta)\) changes the weighted sum.  This
third contribution is the usual score, or likelihood-ratio, term.

If the weights are already normalised, \(w_i(\theta)=p_i(\theta)\), and
we write
\[
  \ell_i(\theta)=\log p_i(\theta),
\]
then Equation~\eqref{eq:appendix-three-term-weighted-particle-gradient} gives
\[
  \begin{aligned}
  &
  \left(
    \nabla (\theta\mapsto
    \mathbb E_{i\sim p_\theta}
    \left[
      f(\theta,z_i(\theta))
    \right])
  \right)_\theta(t)
  \\
  &\qquad =
  \mathbb E_{i\sim p_\theta}
  \left[
    (\nabla f)^\Theta_{(\theta,z_i(\theta))}(t)
    +
    (\nabla z_i)_\theta
    \left(
      (\nabla f)^Z_{(\theta,z_i(\theta))}(t)
    \right)
    +
    f(\theta,z_i(\theta))(\nabla\ell_i)_\theta(t)
  \right].
  \end{aligned}
\]

If normalised probabilities are instead obtained from unnormalised
weights by a separate normalisation map, let
\[
  p_i(\theta)=\frac{w_i(\theta)}{\sum_{j\in I}w_j(\theta)},
  \qquad
  \ell_i(\theta)=\log w_i(\theta).
\]
Then the lifted fibre map of that normalisation gives
\[
  \begin{aligned}
  &
  \left( 
    \nabla (\theta\mapsto 
    \mathbb E_{i\sim p_\theta}
    \left[
      f(\theta,z_i(\theta))
    \right])
  \right)_\theta(t)
  \\
  &\qquad =
  \mathbb E_{i\sim p_\theta}
  \left[
    (\nabla f)^\Theta_{(\theta,z_i(\theta))}(t)
    +
    (\nabla z_i)_\theta
    \left(
      (\nabla f)^Z_{(\theta,z_i(\theta))}(t)
    \right)
  \right.
  \\
  &\qquad\qquad\left.
    +
    \left(
      f(\theta,z_i(\theta))
      -
      \mathbb E_{j\sim p_\theta}
      [f(\theta,z_j(\theta))]
    \right)
    (\nabla\ell_i)_\theta(t)
  \right].
  \end{aligned}
\]

\section{Algebras and strong morphisms into the continuation monad}
\label{sec:appendix-algebras-and-continuations}

We give a proof of
Lemma~\ref{lem:algebras-and-strong-monad-morphisms-to-the-continuation-monad}.
Throughout this appendix, let
\[
  K := \cont{V}^{\cat[]{X}} = \inthom{\inthom{\blankdash}{V}}{V}.
\]
We use the internal language of the cartesian closed category
\(\cat[]{X}\).  Thus the elementwise formulas below denote the
corresponding morphism equalities obtained by currying and evaluation.

Recall first the continuation monad structure.  For a morphism
\(f:A\to B\), the functorial action of \(K\) is contravariant
precomposition twice:
\[
  \Fun{K}{f}(c)(\ell)=c(\ell\circ f),
  \qquad
  c:\Fun{K}{A},
  \quad
  \ell:\inthom{B}{V}.
\]
Its unit and multiplication are
\[
  \Nat{\eta^{K}}{A}(x)(k)=k(x),
  \qquad
  x:A,
  \quad
  k:\inthom{A}{V},
\]
and
\[
  \Nat{\mu^{K}}{A}(\Phi)(k)
  =
  \Phi(c\mapsto c(k)),
  \qquad
  \Phi:\Fun{K}{\Fun{K}{A}},
  \quad
  k:\inthom{A}{V}.
\]
The canonical strength of \(K\) is
\[
  \Nat{\str^{K}}{A,B}(x,c)(h)
  =
  c(b\mapsto h(x,b)),
\]
where \(x:A\), \(c:\Fun{K}{B}\), and
\(h:\inthom{A\times B}{V}\).

A \(T\)-algebra structure on \(V\) is a morphism
\[
  a:\Fun{T}{V}\to V
\]
satisfying
\[
  a\circ\Nat{\eta^{T}}{V}=\id_{V}
  \qquad
  \text{and}
  \qquad
  a\circ\Nat{\mu^{T}}{V}
  =
  a\circ\Fun{T}{a}.
\]

Given such an algebra \(a\), define a natural transformation
\[
  \theta^{a}:T\to K
\]
as follows.  Its component at \(A\) is the transpose of the composite
\[
  \Fun{T}{A}\times\inthom{A}{V}
  \xrightarrow{\simeq}
  \inthom{A}{V}\times\Fun{T}{A}
  \xrightarrow{\Nat{\str^{T}}{\inthom{A}{V},A}}
  \Fun[np]{T}{\inthom{A}{V}\times A}
  \xrightarrow{\Fun{T}{\operatorname{ev}_{A,V}}}
  \Fun{T}{V}
  \xrightarrow{a}
  V,
\]
where
\(\operatorname{ev}_{A,V}:\inthom{A}{V}\times A\to V\)
is evaluation.  Equivalently,
\[
  \Nat{\theta^{a}}{A}(t)(k)
  =
  a(\Fun{T}{k}(t)),
  \qquad
  t:\Fun{T}{A},
  \quad
  k:\inthom{A}{V}.
\]
Naturality in \(A\) follows immediately from functoriality of \(T\): for
\(f:A\to B\) and \(\ell:\inthom{B}{V}\),
\[
  \Nat{\theta^{a}}{B}(\Fun{T}{f}(t))(\ell)
  =
  a(\Fun{T}{\ell\circ f}(t))
  =
  \Fun{K}{f}(\Nat{\theta^{a}}{A}(t))(\ell).
\]

The unit law for a monad morphism is
\[
\begin{aligned}
  \Nat{\theta^{a}}{A}(\Nat{\eta^{T}}{A}(x))(k)
  &=
  a(\Nat{\eta^{T}}{V}(k(x)))        \\
  &=
  k(x)                              \\
  &=
  \Nat{\eta^{K}}{A}(x)(k).
\end{aligned}
\]
For multiplication, let \(u:\Fun{T}{\Fun{T}{A}}\).  Then
\[
\begin{aligned}
  \Nat{\theta^{a}}{A}(\Nat{\mu^{T}}{A}(u))(k)
  &=
  a(\Fun{T}{k}(\Nat{\mu^{T}}{A}(u)))                                  \\
  &=
  a(\Nat{\mu^{T}}{V}(\Fun{T}{\Fun{T}{k}}(u)))                          \\
  &=
  a(\Fun{T}{a}(\Fun{T}{\Fun{T}{k}}(u)))                                \\
  &=
  a(\Fun{T}{(a\circ\Fun{T}{k})}(u))                                    \\
  &=
  \Nat{\theta^{a}}{\Fun{T}{A}}(u)
    (s\mapsto \Nat{\theta^{a}}{A}(s)(k))                               \\
  &=
  \bigl(
    \Nat{\mu^{K}}{A}
    \circ
    \Fun{K}{\Nat{\theta^{a}}{A}}
    \circ
    \Nat{\theta^{a}}{\Fun{T}{A}}
  \bigr)(u)(k).
\end{aligned}
\]
Thus \(\theta^{a}\) is a monad morphism.

It remains to check that it is strong.  Let
\(x:A\), \(t:\Fun{T}{B}\), and
\(h:\inthom{A\times B}{V}\).  Using the coherence of the strength of \(T\),
applying \(h\) after adding the pure parameter \(x\) to the computation \(t\)
is the same as applying the function \(b\mapsto h(x,b)\) to the computation
\(t\).  Hence
\[
\begin{aligned}
  \Nat{\theta^{a}}{A\times B}
  (\Nat{\str^{T}}{A,B}(x,t))(h)
  &=
  a(\Fun{T}{h}(\Nat{\str^{T}}{A,B}(x,t)))       \\
  &=
  a(\Fun{T}{(b\mapsto h(x,b))}(t))              \\
  &=
  \Nat{\theta^{a}}{B}(t)(b\mapsto h(x,b))        \\
  &=
  \Nat{\str^{K}}{A,B}
  (x,\Nat{\theta^{a}}{B}(t))(h).
\end{aligned}
\]
Therefore \(\theta^{a}\) preserves strength.

Conversely, suppose that
\[
  \theta:T\to K
\]
is a strong monad morphism.  Define
\[
  a^{\theta}:\Fun{T}{V}\to V
\]
by evaluating \(\Nat{\theta}{V}\) at the identity continuation:
\[
  a^{\theta}
  :=
  \operatorname{ev}_{\id_{V}}\circ\Nat{\theta}{V}.
\]
In internal notation,
\[
  a^{\theta}(t)=\Nat{\theta}{V}(t)(\id_{V}).
\]
Here \(\operatorname{ev}_{\id_{V}}\) denotes evaluation at the point
\(\id_{V}:1\to\inthom{V}{V}\) corresponding to the identity morphism on \(V\).

The algebra unit law follows from the unit law for the monad morphism
\(\theta\):
\[
\begin{aligned}
  a^{\theta}(\Nat{\eta^{T}}{V}(v))
  &=
  \Nat{\theta}{V}(\Nat{\eta^{T}}{V}(v))(\id_{V}) \\
  &=
  \Nat{\eta^{K}}{V}(v)(\id_{V})                  \\
  &=
  v.
\end{aligned}
\]
For multiplication, let \(u:\Fun{T}{\Fun{T}{V}}\).  Then
\[
\begin{aligned}
  a^{\theta}(\Nat{\mu^{T}}{V}(u))
  &=
  \Nat{\theta}{V}(\Nat{\mu^{T}}{V}(u))(\id_{V})                       \\
  &=
  \Nat{\mu^{K}}{V}
  \bigl(
    \Fun{K}{\Nat{\theta}{V}}
    (\Nat{\theta}{\Fun{T}{V}}(u))
  \bigr)(\id_{V})                                                     \\
  &=
  \Nat{\theta}{\Fun{T}{V}}(u)
  (s\mapsto \Nat{\theta}{V}(s)(\id_{V}))                              \\
  &=
  \Nat{\theta}{\Fun{T}{V}}(u)(a^{\theta})                             \\
  &=
  \Nat{\theta}{V}(\Fun{T}{a^{\theta}}(u))(\id_{V})                    \\
  &=
  a^{\theta}(\Fun{T}{a^{\theta}}(u)).
\end{aligned}
\]
Thus \(a^{\theta}\) is a \(T\)-algebra structure on \(V\).

Finally, the two constructions are inverse.  Starting from an algebra \(a\),
we have
\[
  a^{\theta^{a}}(t)
  =
  \Nat{\theta^{a}}{V}(t)(\id_{V})
  =
  a(\Fun{T}{\id_{V}}(t))
  =
  a(t).
\]
Starting from a strong monad morphism \(\theta\), let
\(t:\Fun{T}{A}\) and \(k:\inthom{A}{V}\).  By the definition of
\(\theta^{a^{\theta}}\), by naturality of \(\theta\), and then by preservation
of strength, we obtain
\[
\begin{aligned}
  \Nat{\theta^{a^{\theta}}}{A}(t)(k)
  &=
  a^{\theta}
  \left(
    \Fun{T}{\operatorname{ev}_{A,V}}
    \bigl(
      \Nat{\str^{T}}{\inthom{A}{V},A}(k,t)
    \bigr)
  \right)                                                            \\
  &=
  \Nat{\theta}{\inthom{A}{V}\times A}
  \bigl(
    \Nat{\str^{T}}{\inthom{A}{V},A}(k,t)
  \bigr)
  (\operatorname{ev}_{A,V})                                           \\
  &=
  \Nat{\str^{K}}{\inthom{A}{V},A}
  (k,\Nat{\theta}{A}(t))
  (\operatorname{ev}_{A,V})                                           \\
  &=
  \Nat{\theta}{A}(t)(k).
\end{aligned}
\]
Hence \(\theta^{a^{\theta}}=\theta\).  Therefore \(T\)-algebra structures on
\(V\) are in one-to-one correspondence with strong monad morphisms
\(T\to\cont{V}^{\cat[]{X}}\), as claimed.

\section{Other algebraic effects}
\label{sec:appendix-additional-effects}

This appendix gives the generic CHAD technique for discrete-output effects including the ones mentioned in Section~\ref{sec:beyond-probability}. The main body develops the case of the FAD monad in detail.
Here we identify the reusable part: if an algebraic effect only has values returned to continuations through discrete output types, then
the raw AST monad has a fibred polynomial lift and the CHAD proof extends uniformly as soon as we give a compatible lifted generic effects and a compatible lifted handler algebra.
We assume some standard background about (many-sorted) Lawvere theories, algebraic effects, and continuation monads~\cite{hylandCategoryTheoreticUnderstanding2007,plotkinAlgebraicOperationsGeneric2003,hylandCombiningAlgebraicEffects2007}. 

The presentation is organised as follows. We first explain what we mean by discrete-output effects and we specify what data need to be supplied to apply CHAD to them.
Then we state the source and target language extensions and give a generic CHAD code transformation and correctness theorem.
We next specialise to the three concrete instances from Section~\ref{sec:beyond-probability}.
Finally, we demonstrate why we impose the limitation to non-discrete outputs.

\subsection{Discrete-output setup and admissible data}
\label{sec:appendix-general-effect-recipe}

\subsubsection{Polynomial ASTs with discrete outputs}

Let \( \Sigma \) be an algebraic signature with a small set of sorts. 
A generic effect symbol is written
\[
  \omega:I_{\omega}\to T O_{\omega}.
\]
Here, we call \(I_{\omega}\) the effect parameter type or input type and \(O_{\omega}\) the effect output type.

Equivalently, with a \( \Set \) interpretation, it determines by Yoneda algebraic-operation families
\[
  \setint{I_{\omega}}\times X^{\setint{O_{\omega}}}
  \to TX,
\]
and 
\[
  \setint{I_{\omega}}\times (TX)^{\setint{O_{\omega}}}
  \to TX .
\]
We see that \( O_{\omega} \) gets the role of the type of values returned to the continuation.

Write
\[
  \widehat I_{\omega}=\fvoint{I_{\omega}}
\]
for the chosen interpretation in \( \Fam{\Vect[o]} \), with base set \(\setint{I_{\omega}}\).
We allow \( I_{\omega} \) to be a genuine differentiable type in the sense that it may have non-trivial cotangent fibres. 
We require \( O_{\omega} \) to be a discrete type, in the sense that it has zero cotangent fibres:
\[
  \fvoint{O_{\omega}}=\Delta\setint{O_{\omega}}=\tuple{\setint{O_{\omega}},\underline{1}}.
\]
For the CHAD rules below we use the finitary case, so \(\setint{O_{\omega}}\) is finite; more generally the same formula applies whenever the required products and direct sums exist.

This induces, before quotienting by equations, the raw syntax tree monad with the polynomial form
\[
  \AST_{\Sigma}X
  \simeq
  \mu Y.\;
  X\sqcup
  \coprod_{\omega\in\Sigma}
  \setint{I_{\omega}}\times
  Y^{\setint{O_{\omega}}}
  \simeq
  \mu Y.\;
  X\sqcup
  \coprod_{\omega\in\Sigma}
  \setint{I_{\omega}}\times
  \prod_{o\in\setint{O_{\omega}}}Y .
\]
The corresponding fibred polynomial lift to \(\Fam{\Vect[o]}\), for an object \(A\in\Fam{\Vect[o]}\) is
\[
  \widehat{\AST}_{\Sigma}A
  \simeq
  \mu B.\;
  A\sqcup
  \coprod_{\omega\in\Sigma}
  \widehat I_{\omega}\times
  B^{\Delta\setint{O_{\omega}}} .
\]
Its base is canonically \(\AST_{\Sigma}(\base{A})\). More explicitly, over a tree-node base element \(\omega(i,k)\), where \(i\in\setint{I_{\omega}}\) and \(k:\setint{O_{\omega}}\to\base{B}\), the fibre is
\[
  \Fun[np]{\fibr{ \widehat{I_{\omega}} }}{i}
  \oplus
  \bigoplus_{o\in\setint{O_{\omega}}}
  \Fun[np]{\fibr{B}}{k(o)} .
\]
This is the point at which the discreteness of \(O_{\omega}\) is used: a continuation branch is indexed only by the set-level value \(o\), and the variable \(o:O_{\omega}\) contributes no cotangent of its own. The same polynomial construction applies componentwise to \(\cat[]{LR}\).

Thus continuous operation-inputs are harmless, provided they are treated as parameters with their own cotangent fibres; what matters for this direct fibred-polynomial construction is discreteness of the operation-output object \( O_{\omega} \). This covers finite choice, exception operations \( \opn{raise}:E\to T_E0 \), writer operations such as \( \opn{tell}:\R[n]\to W_{\R[n]}1 \), and state/reader operations over finite discrete state/reader objects. It does not directly cover continuous state, continuous reader, or continuous probability, where the canonical operations have genuine output objects with nonzero fibres. For such an output object \( \widehat O \), a continuation is a morphism \( \widehat O\to B \), not merely a function \( \setint{O}\to\base{B} \); hence the simple set-indexed product \(B^{\Delta\setint O}\) above no longer describes the desired continuation object.


\subsubsection{Admissible lifted data}

A \emph{CHAD-admissible discrete-output presentation} consists of the following data.

First, we assume a signature \( \Sigma \) as above and a set of equations \( E \) presenting a strong monad \(T\) on \(\Set\) as a quotient of the raw tree monad,
\[
  q:\AST_{\Sigma}\twoheadrightarrow T,
\]
with \(q\) a morphism of strong monads. Equivalently, \(T\) is the monad presented by the generic effects of \(\Sigma\) modulo \(E\).
(In fact, the equations \(E\) are often precisely obtained from the algebra/handler below.)

Second, we assume a chosen handler carrier type \(\nu\) and a \(T\)-algebra
\[
  h:T\nu\to \nu .
\]
This algebra is what discharges the effect in real-output programs.
(We can also allow a family of handler carriers and handler algebras, in which case we obtain a strong monad morphism from ASTs to a product of continuation monads, but we avoid this generality to simplify the presentation.)

Third, we assume chosen strong monad lifts
\[
  T^{\cat[]{F}},\quad T^{\cat[]{LR}}
\]
over \(T\), in the sense that their projections to the underlying set semantics are naturally isomorphic to \(T\) and the isomorphisms respect unit, multiplication, and strength. For each operation symbol \(\omega\), these lifted monads are equipped with lifted generic effects, for example in \(\Fam{\Vect[o]}\)
\[
  \widehat\omega^{\cat[]{F}}_A:
  \widehat I_{\omega}\times
  \left(T^{\cat[]{F}}A\right)^{\Delta\setint{O_{\omega}}}
  \to
  T^{\cat[]{F}}A,
\]
and analogously in \(\cat[]{LR}\). Under the base projection, \(\widehat\omega^{\cat[]{F}}_A\) is required to be the set-level generic-effect constructor
\[
  \setint{I_{\omega}}\times
  \left(\Fun[np]{T}{\base{A}}\right)^{\setint{O_{\omega}}}
  \to
  \Fun[np]{T}{\base{A}}
\]
induced by \(q\); the \(\cat[]{LR}\)-operation satisfies the corresponding logical-relations compatibility. The lifted operations must satisfy the equations \(E\) and be coherent with the lifted unit, multiplication, and strength.

Finally, we assume lifted handler algebras
\[
  h^{\cat[]{F}}:T^{\cat[]{F}}\nu^{\cat[]{F}}\to \nu^{\cat[]{F}},
  \qquad
  h^{\cat[]{LR}}:T^{\cat[]{LR}}\nu^{\cat[]{LR}}\to \nu^{\cat[]{LR}}.
\]
They are required to be Eilenberg--Moore algebras for the lifted monads and to project to the set-level handler algebra \(h\), with the evident logical-relations compatibility for \(h^{\cat[]{LR}}\).

These assumptions are exactly the data consumed by the CHAD proof. The proof does not use how the lifted monad was constructed. It may be obtained by a direct quotient/coimage construction, by an image construction, or directly from products and coproducts, provided the resulting operations and handler satisfy the compatibility conditions just listed.

\subsection{Language extension and CHAD theorem}
\label{sec:appendix-language-extension-chad-theorem}

Fix a CHAD-admissible discrete-output presentation as in
Section~\ref{sec:appendix-general-effect-recipe}.  The monadic type
\(\opn[b]{M}\) is now interpreted by the presented monad \(T\), together with
its chosen lifts \(T^{\cat[]{F}}\) and \(T^{\cat[]{LR}}\).  The rules below do
not construct these lifts from the set-level operations alone; they expose the
backwards fibre maps of the chosen lifted generic effects and lifted handler
algebras.  This point is essential when an operation parameter
\(I_{\omega}\) has nonzero cotangent fibres.  All equations involving the added
operations are interpreted modulo the equations \(E\) of the presentation and,
in the target, through the corresponding lifted operations.

\subsubsection{Source operations, equations, and handlers}

For every generic effect symbol \(\omega:I_{\omega}\to T O_{\omega}\), the
source language is extended by the operation form
\[
\prftree[r]{\textsc{Op}\(_\omega\)}{
  \Gamma\vdash i:I_{\omega}
  \qquad
  \Gamma,o:O_{\omega}\vdash k_o:\Fun[nn]{\opn[b]{M}}{\tau}
}{
  \Gamma\vdash
  \fn[n]{\opn[b]{op}_{\omega}}{i;\,o.k_o}
  :\Fun[nn]{\opn[b]{M}}{\tau}
}.
\]
The binder \(o.k_o\) denotes a finite family of continuation terms indexed by
\(\setint{O_{\omega}}\).  Since \(O_{\omega}\) is discrete, extending the context
by \(o:O_{\omega}\) changes the base context but adds no cotangent component.
Write
\[
  \omega_X(i,k)
  :=
  \mu_X\bigl(T(k)(\omega(i))\bigr),
  \qquad
  k:\setint{O_{\omega}}\to TX,
\]
for the quotient-induced Kleisli constructor associated with the generic effect
\(\omega:I_{\omega}\to T O_{\omega}\).  The set-level semantics is
\[
  \setint{\fn[n]{\opn[b]{op}_{\omega}}{i;\,o.k_o}}(\gamma)
  =
  \omega_{\setint{\tau}}
  \left(
    \setint{i}(\gamma),
    o\mapsto \setint{k_o}(\gamma,o)
  \right).
\]

The source equational theory is extended in three ways.  First, the equations
\(E\) presenting \(T\) are imposed on the generic effects.  Secondly, the
operation form is algebraic, i.e. it commutes with Kleisli extension.  Explicitly,
for \(\Gamma,x:\tau\vdash s:\Fun[nn]{\opn[b]{M}}{\sigma}\), we impose the equation
schema
\begin{equation}
  \label{eq:appendix-generic-operation-algebraicity-law}
  \fn[n]{\opn[b]{bind}}{
    \fn[n]{\opn[b]{op}_{\omega}}{i;\,o.k_o};
    x.s
  }
  \equiv
  \fn[n]{\opn[b]{op}_{\omega}}{
    i;\,o.\fn[n]{\opn[b]{bind}}{k_o;x.s}
  } .
\end{equation}
The right-hand side is well typed because the variable \(o:O_{\omega}\) is in
scope in the continuation branch.  Finally, the handler associated with the
chosen algebra \(h:T\nu\to\nu\) is added as
\[
\prftree[r]{\textsc{Handle}\(_h\)}{
  \Gamma\vdash t:\Fun[nn]{\opn[b]{M}}{\nu}
}{
  \Gamma\vdash \fn[n]{\opn[b]{handle}_{h}}{t}:\nu
},
\]
and satisfies the Eilenberg--Moore algebra laws.  In term notation, for
\(\Gamma\vdash v:\nu\), \(\Gamma\vdash t:\Fun[nn]{\opn[b]{M}}{\sigma}\), and
\(\Gamma,x:\sigma\vdash u:\Fun[nn]{\opn[b]{M}}{\nu}\), these are the equation schemas
\begin{equation}
  \label{eq:appendix-generic-handler-em-laws}
  \begin{aligned}
    \fn[n]{\opn[b]{handle}_{h}}{\fn[n]{\opn[b]{return}}{v}}
    &\equiv v,
    \\
    \fn[n]{\opn[b]{handle}_{h}}{
      \fn[n]{\opn[b]{bind}}{t;x.u}
    }
    &\equiv
    \fn[n]{\opn[b]{handle}_{h}}{
      \fn[n]{\opn[b]{bind}}{
        t;x.\fn[n]{\opn[b]{return}}{
          \fn[n]{\opn[b]{handle}_{h}}{u}
        }
      }
    } .
  \end{aligned}
\end{equation}
These laws are exactly the assertion that \(\opn[b]{handle}_{h}\) interprets the
\(T\)-algebra \(h\); the second equation is \(h\circ\mu=h\circ T h\) written
in Kleisli form.

\subsubsection{Linear target operations}

The target language contains, for each generic effect \(\omega\), a linear
constructor
\[
  \fn[n]{\opn[b]{op}^{w}_{\omega,i;\,o.k_o}}{
    u.i' \; ; \; o.z.k'_o
  } .
\]
The notation \(\lintype{I_\omega}\) in the rule below denotes the fibre of the
chosen parameter object \(\widehat I_\omega\); if \(I_\omega\) is discrete, this
is the zero linear type.  The constructor represents the reverse fibre component
of the chosen lifted operation (the transposed derivative in differentiable examples)
\[
  \widehat\omega^{\cat[]{F}}_A:
  \widehat I_{\omega}\times
  \left(T^{\cat[]{F}}A\right)^{\Delta\setint{O_{\omega}}}
  \to
  T^{\cat[]{F}}A,
\]
and similarly in \(\cat[]{LR}\).  Thus, for
\(i\in\setint{I_\omega}\) and a branch family
\(p:\setint{O_\omega}\to\base{T^{\cat[]{F}}A}\), writing \(p_o=p(o)\), its reverse fibre map is
\[
  \Fun[np]{\fibr{ T^{\cat[]{F}}A }}{\omega(i,p)}
  \longrightarrow
  \Fun[np]{\fibr{ \widehat{ I_\omega } }}{i}
  \oplus
  \bigoplus_{o\in\setint{O_\omega}}
  \Fun[np]{\fibr{ T^{\cat[]{F}}A }}{p(o)},
\]
where \(\omega(i,p)=\mu_{\base{A}}(T(p)(\omega(i)))\) is the induced Kleisli constructor.  In
one-sorted notation it has the following rule, in the same continuation-passing
style as \(\opn[b]{categorical}^{w}\):
\begin{equation}
  \label{eq:appendix-generic-operation-w-typing}
  \begin{gathered}
    \Gamma,y:I_{\omega}\vdash \lintype{I_{\omega}} \text{ type}
    \qquad
    \Gamma,x:\tau\vdash \lintype{\tau} \text{ type}
    \qquad
    \Gamma\vdash \lintype{\rho} \text{ type}
    \\
    \prftree{
      \Gamma\vdash i:I_{\omega}
    }{
      \Gamma,o:O_{\omega}\vdash k_o:\Fun[nn]{\opn[b]{M}}{\tau}
    }{
      \Gamma;u:\casubst{\lintype{I_{\omega}}}{y}{i}
      \vdash i':\lintype{\rho}
    }{
      \Gamma,o:O_{\omega};
      z:
      \opn[b]{M}^{w}_{k_o:\Fun[nn]{\opn[b]{M}}{\tau}}
      \left[x\mapsto \lintype{\tau}\right]
      \vdash k'_o:\lintype{\rho}
    }{
      \Gamma;
      v:
      \opn[b]{M}^{w}_{
        \fn[n]{\opn[b]{op}_{\omega}}{i;\,o.k_o}
        :
        \Fun[nn]{\opn[b]{M}}{\tau}
      }
      \left[x\mapsto \lintype{\tau}\right]
      \vdash
      \fn[n]{\opn[b]{op}^{w}_{\omega,i;\,o.k_o}}{
        u.i' \; ; \; o.z.k'_o
      }
      :
      \lintype{\rho}
    }
  \end{gathered}
\end{equation}
The variable \(u\) receives the cotangent contribution for the operation
parameter, while \(z\) receives the cotangent contribution for the corresponding
continuation branch.  Since \(O_{\omega}\) is discrete, extending the context by
\(o:O_{\omega}\) does not add a linear component; hence each \(k'_o\) returns a
cotangent for the same ambient context \(\Gamma\).  The finite direct sum above
is where the finiteness of \(O_{\omega}\) enters the target syntax.
Semantically, if the reverse fibre map sends the incoming monadic cotangent \(v\) to
\(\tuple{u,(z_o)_{o\in\setint{O_{\omega}}}}\), then the constructor denotes the finite sum
\[
  \fn[n]{i'}{u}
  +
  \sum_{o\in\setint{O_{\omega}}}
  \fn[n]{k'_o}{z_o}.
\]
For \(O_{\omega}=0\), this branch sum is empty.

With this operation former, the CHAD clause for algebraic operations is the
following.  We use the canonical identification
\[
  \Fun[np]{\chad}{\Gamma,o:O_{\omega}}_{2}
  \cong
  \Fun[np]{\chad}{\Gamma}_{2},
\]
since \(O_{\omega}\) is discrete and has zero cotangent fibres (syntactically, \(\lintype{O_{\omega}}\) is the terminal zero-cotangent type \(\underline{1}\)).  Then
\begin{equation}
  \label{eq:appendix-generic-operation-chad}
  \begin{aligned}
    &\Fun[np]{\chad[\Gamma]}{
      \fn[n]{\opn[b]{op}_{\omega}}{i;\,o.k_o}
    }
    :=
    \\
    &\quad
    \opn[b]{case}\;\Fun[np]{\chad[\Gamma]}{i}\;\opn[b]{of}\;
    \tuple{i_0,i'}\to
    \\
    &\quad
    \opn[b]{let}\;\tuple{r_o,r'_o}
    =
    \Fun[np]{\chad[\Gamma,o:O_{\omega}]}{k_o}
    \;\opn[b]{in}
    \\
    &\quad\quad
    \tuple{
      \fn[n]{\opn[b]{op}_{\omega}}{i_0;\,o.r_o},
      \lambda v\; .\;
      \fn[n]{\opn[b]{op}^{w}_{\omega,i_0;\,o.r_o}}{
        u.\fn[n]{i'}{u}
        \; ; \;
        o.z.\fn[n]{r'_o}{z}
      }
    } .
  \end{aligned}
\end{equation}
The \(\opn[b]{let}\)-binding of \(\tuple{r_o,r'_o}\) is meta-notation for the
family obtained by translating the continuation in the context
\(\Gamma,o:O_{\omega}\).  The incoming cotangent \(v\) is consumed by
\(\opn[b]{op}^{w}_{\omega}\), which separates it into the parameter cotangent
and the branch cotangents; the backpropagators \(i'\) and \(r'_o\) then return
these contributions to the ambient context.  No cotangent for \(o\) is omitted:
there is none.

\subsubsection{Linear handler operations}

Similarly, the target language contains a linear constructor
\( \opn[b]{handle}^{w}_{h,p} \).  To avoid overloading the incoming output
cotangent with the monadic cotangent produced by the handler algebra, we write
this constructor with an explicit binder \(w.s\):
\begin{equation}
  \label{eq:appendix-generic-handler-w-typing}
  \begin{gathered}
    \Gamma,x:\nu\vdash \lintype{\nu} \text{ type}
    \qquad
    \Gamma\vdash \lintype{\rho} \text{ type}
    \\
    \prftree{
      \Gamma\vdash p:\Fun[nn]{\opn[b]{M}}{\nu}
    }{
      \Gamma;
      w:
      \opn[b]{M}^{w}_{p:\Fun[nn]{\opn[b]{M}}{\nu}}
      \left[x\mapsto \lintype{\nu}\right]
      \vdash s:\lintype{\rho}
    }{
      \Gamma;
      v:
      \casubst{\lintype{\nu}}{x}{\fn[n]{\opn[b]{handle}_{h}}{p}}
      \vdash
      \fn[n]{\opn[b]{handle}^{w}_{h,p}}{w.s}
      :\lintype{\rho}
    }
  \end{gathered}
\end{equation}
Semantically, if \(B=\nu^{\cat[]{F}}\), the lifted handler algebra
\(h^{\cat[]{F}}:T^{\cat[]{F}}B\to B\) has reverse fibre map
\[
  \Fun[np]{\fibr{B}}{h(p)}
  \longrightarrow
  \Fun[np]{\fibr{ T^{\cat[]{F}}B }}{p}.
\]
The constructor \(\opn[b]{handle}^{w}_{h,p}\) binds \(w\) to the image of
\(v\) under this reverse fibre map and then evaluates \(s\).  Thus it uses the
supplied lifted algebra, not merely the underlying set map \(h\).  In the FAD
instance \(\nu=\srcreal[p]\), this is the same continuation-passing rule as
Equation~\eqref{eq:target-language-typing-rule-handler-algebra}; the explicit
binder only makes the two cotangent types visible.

As for FAD, we again impose the transposed derivatives of the equations for the operations and handlers from the source language, e.g., the Eilenberg--Moore laws for the handler algebra.

The CHAD clause is the exact analogue of the canonical FAD handler
clause:
\begin{equation}
  \label{eq:appendix-generic-handler-chad}
  \Fun[np]{\chad[\Gamma]}{\fn[n]{\opn[b]{handle}_{h}}{t}}
  :=
  \opn[b]{case}\;\Fun[np]{\chad[\Gamma]}{t}\;\opn[b]{of}\;\tuple{p,b}\to
  \tuple{
    \fn[n]{\opn[b]{handle}_{h}}{p},
    \lambda v\; .\;
    \fn[n]{\opn[b]{handle}^{w}_{h,p}}{w.\fn[n]{b}{w}}
  }.
\end{equation}
Here \(b\) is a backpropagator expecting a monadic cotangent over \(p\).  The
incoming output cotangent \(v\) is first pulled back along the lifted handler
algebra to the monadic cotangent \(w\), and only then passed to \(b\).

\subsubsection{The generic CHAD result}

\begin{theorem}[CHAD for discrete-output algebraic effects]
  \label{th:appendix-generic-discrete-output-chad}
  For every CHAD-admissible discrete-output presentation, the source and target
  languages, the CHAD translation, and the denotational semantics of
  Sections~\ref{sec:source-language}--\ref{sec:interpretation-of-the-target-language-in-fam-vectop}
  extend to the operations of \(\Sigma\) and to the handler \(h\).  The
  fundamental logical-relations lemma continues to hold for the extended
  language.  Consequently, for every well-typed real-output program whose added
  effects are discharged by the chosen handler, the conclusion of
  Theorem~\ref{th:correctness-theorem} continues to hold.
\end{theorem}

\begin{proof}
  The source interpretation sends \(\opn[b]{op}_{\omega}\) to the
  quotient-induced Kleisli constructor \(\omega_X\), and sends
  \(\opn[b]{handle}_{h}\) to the algebra \(h:T\nu\to\nu\).  The equations \(E\), the algebraicity law
  \eqref{eq:appendix-generic-operation-algebraicity-law}, and the
  Eilenberg--Moore laws \eqref{eq:appendix-generic-handler-em-laws} are sound by
  construction.

  The target interpretation sends \(\opn[b]{M}^{w}\),
  \(\opn[b]{op}^{w}_{\omega}\), and \(\opn[b]{handle}^{w}_{h,p}\) to the fibre
  components of the chosen lifted monads, lifted generic effects, and lifted
  handler algebras in \(\Fam{\Vect[o]}\).  The \(\cat[]{LR}\)-interpretation uses
  the corresponding supplied \(\cat[]{LR}\)-data.  Hence the target equations are
  sound because these data are assumed to be lifted strong monads, lifted
  operations satisfying \(E\), and lifted Eilenberg--Moore algebras.

  It remains to check the fundamental lemma for the new typing rules.  In the
  operation case, the induction hypotheses give the related parameter and the
  related family of branch computations.  Applying the lifted generic effect
  in \(\cat[]{LR}\) gives the related transformed computation.  Its fibre
  component is precisely the operation represented syntactically by
  \(\opn[b]{op}^{w}_{\omega}\).  Because \(O_{\omega}\) is discrete, the branch
  variable contributes no cotangent; the fibre component therefore produces only
  the parameter cotangent and the finite family of branch cotangents appearing in
  Equation~\eqref{eq:appendix-generic-operation-chad}.  The displayed
  backpropagators are exactly the chain rule for these components.

  In the handler case, the \(\cat[]{LR}\)-lifted algebra over \(h\) maps a
  related monadic computation to a related handled value, and its fibre component
  is represented by \(\opn[b]{handle}^{w}_{h,p}\).  Thus
  Equation~\eqref{eq:appendix-generic-handler-chad} is again the chain rule.  All
  other cases are the original cases of Theorem~\ref{th:correctness-theorem}, so
  the correctness statement follows.
\end{proof}

\subsection{Concrete instances}
\label{sec:appendix-additional-effects-examples}

The following examples instantiate the admissible data of
Section~\ref{sec:appendix-general-effect-recipe}.  In each case the point to
check is not merely that the set-level monad is algebraic, but that the chosen
fibred monad, the lifted generic effects, and the lifted handler are genuine
Eilenberg--Moore data over the set-level ones.  We display the
\(\Fam{\Vect[o]}\)-fibre maps; the \(\cat[]{LR}\)-lifts used in the logical
relations proof are obtained by the same constructions with the relation data
carried componentwise.

\subsubsection{Finite non-determinism via the finite multiset monad}
\label{sec:appendix-multiset-nondeterminism}

We model multiplicity-sensitive finite non-determinism by the finite multiset monad
\[
  \MFin X = \coprod_{k \in \N} \left(X^{k}/\mathfrak{S}_{k}\right),
\]
where \( \mathfrak{S}_{k} \) acts by permuting the entries of a \( k \)-tuple.
Multiplicity is retained but order is forgotten.  If one additionally quotients
by idempotence, one obtains the finite-powerset variant of non-determinism,
whose algebras are join-semilattices with bottom rather than commutative
monoids.  We use the multiset, not the idempotent, variant below; in particular,
the additive handler \(\opn{sum}_n\) below would not factor through the
idempotent quotient unless the chosen target algebra were idempotent.

The unit sends \(x\) to the singleton multiset
\(\left\{\!\left\{x\right\}\!\right\}\), and bind maps each occurrence to a
finite multiset and then takes multiset union.  An algebra
\(a:\MFin A\to A\) is equivalently a commutative monoid object on \(A\).  Given
\(a\), the unit is \(a(\varnothing)\) and multiplication is
\((a_1,a_2)\mapsto a(\left\{\!\left\{a_1,a_2\right\}\!\right\})\).  Conversely,
a commutative monoid folds finite multisets by iterated multiplication; the
algebra laws are exactly unit, associativity, and commutativity.

The generic-operation presentation may be taken to have operations
\[
  \opn{choose}_k:1\to \MFin\{1,\dots,k\},
  \qquad
  \ast\mapsto \left\{\!\left\{1,\dots,k\right\}\!\right\},
\]
for \(k\in\N\), modulo the usual unit, associativity, and permutation equations.
The induced Kleisli constructor sends a finite family of branch computations to
their multiset union; \(k=0\) gives the empty choice.

The quotient by permutations must also be reflected in the fibres.  Let
\(\tuple{X,A}\in\Fam{\Vect[o]}\), let
\(m\in\MFin X\), and choose a representative tuple
\(\vec x=(x_1,\dots,x_k)\) of \(m\).  The free commutative-monoid lift obtained
by quotienting the lifted list monad has fibre
\begin{equation}
  \label{eq:appendix-multiset-fibre-stabiliser}
  \Fun[np]{\fibr{ \MFin^{\cat[]{F}}\tuple{X,A} }}{m}
  =
  \left(
    \bigoplus_{i=1}^{k}\Fun[nn]{A}{x_i}
  \right)^{\operatorname{Stab}(\vec x)},
\end{equation}
where \(\operatorname{Stab}(\vec x)\leq\mathfrak{S}_k\) is the subgroup
permuting equal entries of \(\vec x\), acting by the corresponding symmetries of
the finite direct sum.  This description is independent of the representative.
Equivalently, if \(m\) is written as a finite-support multiplicity function,
then the fibre is canonically the direct sum over the support,
\[
  \Fun[np]{\fibr{ \MFin^{\cat[]{F}}\tuple{X,A} }}{m}
  \cong
  \bigoplus_{x\in\operatorname{supp}(m)}\Fun[nn]{A}{x},
\]
with the quotient map from ordered lists copying the support cotangent to every
indistinguishable occurrence.  This last copying is essential: the tempting
formula with an independent summand for every occurrence is a labelled-trace
semantics before quotienting by stabilisers, not the canonical free
commutative-monoid lift on the finite-multiset quotient.

For a morphism \(\tuple{f,\alpha}:\tuple{X,A}\to\tuple{Y,B}\), the base map is
multiset push-forward.  Under the support description, the fibre map over
\(m\in\MFin X\) sends
\[
  (b_y)_{y\in\operatorname{supp}(\MFin f(m))}
  \longmapsto
  \bigl(\Nat[nn]{\alpha}{x}(b_{f(x)})\bigr)_{x\in\operatorname{supp}(m)} .
\]
Thus a cotangent at a target support point is copied to all source support
points that map to it, and the quotient map then copies it to the corresponding
occurrences.  The unit, multiplication, and strength are induced from the lifted
list monad and descend because the permutation equations act by symmetries of
finite direct sums.

For every \(n\), the additive commutative monoid
\(\tuple{\R[n],0,+}\) gives the handler
\[
  \opn{sum}_{n}:\MFin\R[n]\to\R[n],
  \qquad
  \left\{\!\left\{v_{1},\dots,v_{k}\right\}\!\right\}\mapsto
  \sum_{i=1}^{k} v_i .
\]
Its lifted fibre map is characterised by the list quotient: after pulling back
to an ordered representative, it sends an incoming cotangent \(\xi\in\R[n]\) to
one copy of \(\xi\) at every occurrence.  Under the support identification
above, this is the invariant cotangent with component \(\xi\) at every support
point.  Hence, when two syntactic branches produce the same base value, the
same support cotangent is copied back to both branches and the surrounding
context addition accounts for the multiplicity.

The source operation \(\opn[b]{choose}_k\) is therefore interpreted by multiset
union of the \(k\) branch computations.  Its linear operation
\(\opn[b]{choose}^{w}_k\) is the reverse fibre map of the descended union map:
it restricts the incoming invariant multiset cotangent to the branch multisets,
copying it along collisions.  The handler clause for
\(\opn[b]{fold}_{\nu}\), and in particular \(\opn[b]{sum}_n\), is the generic
algebra clause using the lifted commutative-monoid fold.  These data satisfy the
commutative-monoid equations in \(\Fam{\Vect[o]}\), so the generic CHAD theorem
applies.

\subsubsection{Exception monads}
\label{sec:appendix-exception-monads}

Fix an exception object \(\widehat E=\tuple{E,B}\) in \(\Fam{\Vect[o]}\); for a
purely discrete exception set one takes \(B_e=0\).  The underlying set-level
exception monad is
\[
  T_E X=X\sqcup E .
\]
Its signature has one generic effect
\[
  \opn{raise}:E\to T_E0,
\]
equivalently one nullary operation for each point of \(E\).  Since the output
type is \(0\), this is a discrete-output effect even when the parameter object
\(\widehat E\) has nonzero cotangent fibres.  The raw AST monad is already
\[
  \AST_E X
  \simeq
  \mu Y.\;X\sqcup E\times Y^0
  \simeq
  X\sqcup E,
\]
so there is no non-trivial quotient to impose.

A \(T_E\)-algebra on a set \(Y\) is a map \(a:Y\sqcup E\to Y\) whose restriction
to \(Y\) is the identity.  Equivalently it is a recovery map
\[
  h_0:E\to Y,
  \qquad
  a=[\id_Y,h_0].
\]
For instance, taking \(Y=A\sqcup E'\) allows a handler to recover some
exceptions as values in \(A\) and translate others into values tagged by a new
exception set \(E'\); it is still an algebra for the original monad \(T_E\) on
the carrier \(Y\).

The exception monad lifts by the coproduct formula
\[
  T_E^{\cat[]{F}}\tuple{X,A}
  =
  \tuple{
    X\sqcup E,
    z\mapsto
    \begin{cases}
      \Fun[nn]{A}{x}, & z=\opn{inl}(x),\\
      \Fun[nn]{B}{e}, & z=\opn{inr}(e).
    \end{cases}
  }.
\]
The lifted unit, multiplication, and strength are the coproduct maps induced by
those of \(T_E\).  The lifted \(\opn{raise}\) operation
\(\widehat E\to T_E^{\cat[]{F}}\tuple{X,A}\) has reverse fibre map the identity
\[
  \Fun[nn]{B}{e}\longrightarrow \Fun[nn]{B}{e},
\]
because there are no continuation branches.

Let \(\tuple{Y,C}\) be a lifted carrier and let
\(\widehat h:\widehat E\to\tuple{Y,C}\) be a morphism over the recovery map
\(h_0:E\to Y\).  The lifted algebra over \([\id_Y,h_0]\) is case analysis.  Its
reverse fibre map sends \(\gamma\in\Fun[np]{C}{y}\) over a returned value
\(\opn{inl}(y)\) to \(\gamma\), while over an exception \(\opn{inr}(e)\) it sends
\(\gamma\in\Fun[np]{C}{h_0(e)}\) to the fibre map of \(\widehat h\),
\[
  \left( \fibr{ \widehat{h} } \right)_e^{\ast}(\gamma)
  \in
  \Fun[np]{B}{e} .
\]
When \(h_0\) is an ordinary differentiable recovery map and \(\widehat h\) is
its cotangent lift, this fibre map is \((Dh_0(e))^{\mathsf T}\gamma\).  The
Eilenberg--Moore laws are immediate from
\([\id_Y,h_0]\circ\eta=\id_Y\) and from the fact that multiplication for
\(X\sqcup E\) is case analysis that leaves the first encountered exception
unchanged.  Thus exceptions fit the admissible pattern directly; no quotient is
needed for the direct coproduct lift.  The free-algebra/product-continuation
construction also recovers \(T_E\) because free exception algebras separate
returned values and exception labels, but a single chosen recovery handler need
not be separating if it identifies exceptions or misses some cotangent direction.

\subsubsection{Writer monads and \texorpdfstring{\(\R[n]\)}{R[n]}-accumulation}
\label{sec:appendix-writer-accumulation}

Let \(M=\tuple{M,e,\cdot}\) be a monoid.  We use the right-action convention
for the monad on \(X\times M\): bind sends \(\tuple{x,m}\), together
with \(f:X\to Y\times M\), to \(\tuple{y,n\cdot m}\) when
\(\fn[n]{f}{x}=\tuple{y,n}\).  Equivalently, this is the usual writer monad
for the opposite monoid; for the additive monoids \(\R[n]\) used below, the two
conventions coincide.  With this convention, an Eilenberg--Moore algebra for
\(W_M\) is precisely a right \(M\)-action
\[
  a:A\times M\to A
\]
with \(a(a_0,e)=a_0\) and \(a(a(a_0,m),n)=a(a_0,m\cdot n)\).

The generic effect is
\[
  \opn{tell}:M\to W_M1,
  \qquad
  r\mapsto\tuple{\ast,r} .
\]
The induced Kleisli constructor sends a parameter \(r\in M\) and a branch
\(k(\ast)=\tuple{x,s}\in W_MX\) to \(\tuple{x,s\cdot r}\).  Hence the
operation-output type is \(1\), so the writer effect is in the discrete-output
fragment even when the parameter object \(M\) is continuous.

For the accumulation example take \(M=\tuple{\R[n],0,+}\) with its standard
cotangent bundle.  The lifted writer monad maps
\[
  \tuple{X,A}
  \mapsto
  \tuple{
    X\times\R[n],
    \left(\tuple{x,r}\mapsto \Fun[nn]{A}{x}\oplus\R[n]\right)
  }.
\]
For a morphism \(\tuple{f,\alpha}:\tuple{X,A}\to\tuple{Y,B}\), the base map is
\(\tuple{x,r}\mapsto\tuple{\fn[n]{f}{x},r}\), and the reverse fibre map is
\[
  \Fun[np]{B}{\fn[n]{f}{x}}\oplus\R[n]
  \longrightarrow
  \Fun[nn]{A}{x}\oplus\R[n],
  \qquad
  \tuple{\beta,\rho}\mapsto\tuple{\Nat[nn]{\alpha}{x}(\beta),\rho}.
\]
The lifted multiplication is the transpose derivative of
\[
  \tuple{\tuple{x,s},r}\mapsto\tuple{x,s+r},
\]
so an incoming cotangent \(\tuple{x^{\ast},u^{\ast}}\) is sent to
\(\tuple{\tuple{x^{\ast},u^{\ast}},u^{\ast}}\).  This proves the lifted monad
laws because the displayed maps are the transpose derivatives of the unit,
associative addition, and strength maps.

For the lifted \(\opn{tell}\) operation, if the branch computation has base
\(\tuple{x,s}\), then the operation has base \(\tuple{x,s+r}\).  Its reverse
fibre map sends an incoming cotangent \(\tuple{x^{\ast},u^{\ast}}\) to the
parameter cotangent \(u^{\ast}\) and the branch cotangent
\(\tuple{x^{\ast},u^{\ast}}\).  Thus the newly written value and the writer value
produced by the continuation both receive the same accumulator cotangent.

For the commutative additive accumulation example, the useful carrier is
\(\sigma\times\R[n]\) with action
\begin{equation}
  \label{eq:canonical-accumulation-action}
  (\sigma\times\R[n])\times\R[n]
  \to
  \sigma\times\R[n],
  \qquad
  \tuple{\tuple{a,u},r}\mapsto\tuple{a,u+r}.
\end{equation}
If \(\sigma\) is lifted with fibre \(A_a\), the transposed derivative of
Equation~\eqref{eq:canonical-accumulation-action} sends a cotangent
\(\tuple{a^{\ast},u^{\ast}}\) at \(\tuple{a,u+r}\) to
\[
  \tuple{\tuple{a^{\ast},u^{\ast}},u^{\ast}}
  \in
  \left(A_a\oplus\R[n]\right)\oplus\R[n],
\]
where the two copies of \(u^{\ast}\) correspond to the old accumulated value and
the newly written value.  The action laws are exactly the
Eilenberg--Moore algebra laws for the writer monad, so this gives the lifted
handler used by \(\opn[b]{run}_{\tau}\) and its linear operation
\(\opn[b]{run}^{w}_{\tau,p}\).  This primitive handles computations whose
returned values already lie in \(\tau\times\srcreal[n]\); the usual
\(\opn{runWriter}:\opn[b]{M}\tau\to\tau\times\srcreal[n]\) is obtained by
first mapping \(a:\tau\) to \(\tuple{a,0}\) and then applying this algebra.
For a general differentiable monoid object, the same statements hold with
addition replaced by the transpose derivative of the chosen multiplication-order
map.

\subsection{Scope and non-discrete output limitations}
\label{sec:appendix-additional-effects-scope}

The constructions above are deliberately discrete-output constructions. They
allow differentiable operation-input objects, such as continuous exception
labels or continuous writer values, because these are parameters
\(I_{\omega}\) and their cotangents are accounted for at the operation node.
They also allow arbitrary zero-fibre output indices \(\Delta O_{\omega}\),
provided the required set-indexed products and direct sums exist; in the
finitary CHAD theorem these output sets are finite.

They do not cover genuine operation-output objects with nonzero cotangent
fibres. For such an object \(\widehat O_{\omega}\), a continuation is a
morphism \(\widehat O_{\omega}\to B\), not just a function
\(P\widehat O_{\omega}\to PB\), so its base point can contain reverse-linear
fibre data. If syntax or handlers can inspect that data, the base of the lifted
construction need not be the intended set-level monad. The obstruction remains
even if genuine powers \(B^{\widehat O_{\omega}}\) are supplied.

\begin{lemma}[Genuine output fibres can destroy fibredness]
  \label{lem:appendix-nondiscrete-output-not-fibred}
  There are a one-operation raw theory, an output object with nonzero fibre
  \(\widehat O\in\Fam{\Vect[o]}\), and objects
  \(A,\widehat R\in\Fam{\Vect[o]}\) for which the following holds.
  If operation nodes are indexed by genuine continuations into the recursively
  generated syntax object, then the resulting direct raw-syntax construction is
  not fibred over the set-level raw monad. Moreover, assuming a genuine power
  \(\widehat R^{\widehat O}\) whose base points are morphisms
  \(\widehat O\to\widehat R\), there is a \(\Fam{\Vect[o]}\)-algebra
  for the raw operation which does not lie over any set-level unary-operation
  algebra, and whose induced image and coimage have bases that cannot be
  obtained from any quotient of the set-level unary-operation monad.
\end{lemma}

\begin{proof}
  Work one-sorted and put
  \[
    \widehat O=(\{*\},\R),
    \qquad
    A=(\{*\},\R).
  \]
  Consider the raw theory with one operation \(\omega\), no parameters, output
  object \(\widehat O\), and no equations. Its set-level shadow is the free
  unary-operation theory, since \(P\widehat O=\{*\}\). Hence the set-level raw
  monad is
  \[
    SX\cong \N\times X,
  \]
  and, over the one-point set \(PA=\{*\}\), there is exactly one depth-one
  set-level term.

  In \(\Fam{\Vect[o]}\), a morphism
  \(k:\widehat O\to A\) consists of the unique base map and a reverse fibre map
  \[
    A_*=\R\longrightarrow \widehat O_*=\R.
  \]
  Thus such variable continuations are parametrised by scalars
  \(\lambda\in\R\); write \(k_{\lambda}:\widehat O\to A\) for the morphism whose
  reverse fibre map is multiplication by \(\lambda\). Let \(\widehat S\) denote
  the direct raw construction obtained by replacing the discrete continuation
  object \(B^{\Delta(P\widehat O)}\) by the genuine power \(B^{\widehat O}\).
  The operation constructor in \(\widehat S A\) is therefore indexed by
  morphisms \(\widehat O\to\widehat S A\). Composing the \(k_{\lambda}\)'s with
  the unit \(\eta_A:A\to\widehat S A\) gives distinct morphisms
  \(\widehat O\to\widehat S A\), because the variable summand carries the
  original fibre and \(\eta_A\) has the identity reverse fibre map there. Hence
  we obtain distinct depth-one base terms
  \[
    t_{\lambda}=\omega(\eta_A\circ k_{\lambda}),
    \qquad \lambda\in\R,
  \]
  because the genuine continuation morphism, including its reverse fibre map, is
  part of the operation-node data and there are no equations identifying
  different \(k_{\lambda}\)'s.
  Consequently \(P\widehat S A\) has continuum many depth-one operation nodes,
  and is in particular uncountable, whereas \(S(PA)\cong\N\) is countable. Thus
  the direct raw construction cannot be fibred over the set-level raw monad.

  The same issue appears for image and coimage semantics once handlers are
  allowed to observe the fibre component of a genuine continuation. Assume a
  genuine power object \(\widehat R^{\widehat O}\) is available, with base
  points the morphisms \(\widehat O\to\widehat R\). Let
  \[
    \widehat R=(\R,\R)
  \]
  be the constant one-dimensional family over the set \(\R\). A base point of
  \(\widehat R^{\widehat O}\) is then a morphism
  \(\widehat O\to\widehat R\), equivalently a pair
  \[
    (r,\ell),
    \qquad
    r\in\R,
    \quad
    \ell:\R\to\R \text{ linear},
  \]
  where \(r\) is the base value and \(\ell\) is the reverse fibre map. Define a
  lifted algebra map
  \[
    h:\widehat R^{\widehat O}\to\widehat R
  \]
  for \(\omega\) by the base function
  \[
    (r,\ell)\longmapsto \ell(1),
  \]
  and choose the zero reverse fibre map at every base point. This is a morphism
  in \(\Fam{\Vect[o]}\), since every fibre is a vector space. Its base map
  depends on the reverse fibre map \(\ell\), so it does not factor through the
  projection to the set-level branch value \(r\); hence it is not a lift of any
  set-level unary-operation algebra. Because the raw theory imposes no
  equations, this operation structure nevertheless extends uniquely to an
  algebra for the free raw-syntax monad in \(\Fam{\Vect[o]}\).

  For each \(\mu\in\R\), let
  \(\rho_{\mu}:A\to\widehat R\) be the valuation with base value \(0\) and
  reverse fibre map multiplication by \(\mu\). The composite
  \[
    \widehat O
      \xrightarrow{k_{\lambda}}
    A
      \xrightarrow{\rho_{\mu}}
    \widehat R
  \]
  has base value \(0\) and reverse fibre map multiplication by
  \(\lambda\mu\), where \(\lambda\mu\) denotes the multiplication map by
  that scalar. Therefore the base denotation of the depth-one term
  \(t_{\lambda}=\omega(\eta_A\circ k_{\lambda})\) under the valuation
  \(\rho_{\mu}\) is
  \[
    (P h)(0,\lambda\mu)=\lambda\mu .
  \]
  Taking \(\mu=1\) separates all \(\lambda\). Hence image and coimage contain
  continuum many depth-one denotations, respectively kernel-pair classes. They
  cannot be bases of lifts of any quotient of the set-level unary-operation
  semantics over \(PA\), since every quotient of \(S(PA)\cong\N\) is countable
  and cannot create additional depth-one base points.

  Thus nonzero output fibres can make continuation-fibre data visible on bases,
  exactly what zero-fibre output indices prevent. Genuine nonzero-fibre outputs
  therefore need extra structure or equations forcing syntax and handlers to
  quotient or ignore that data.
\end{proof}


\end{document}